\documentclass[orivec]{llncs}
\usepackage{xspace}
\newcommand{\scf}{\ensuremath{\mathcal{SC}_{\fpsi}}\xspace}
\newcommand{\SCe}{\ensuremath{\mathcal{SC}_{\epsi}}\xspace}

\newcommand{\resizeS}{\scalebox{.52}}

\newcommand{\st}{\scriptscriptstyle}

\usepackage{lipsum}

\newcommand{\et}{\textit{et al.}\xspace}
\newcommand{\prf}{\ensuremath{\mathtt{PRF}}}

\newcommand{\PRP}{\mathtt{PRP}}
\newcommand{\minn}{\resizeS{{min}}}

\newcommand{\fpsi}{\ensuremath{\mathtt{JUS}}\xspace}
\newcommand{\epsi}{\ensuremath{\mathtt{ANE}}\xspace}

\newcommand{\withRew}{{Anesidora}\xspace}
\newcommand{\withFai}{{Justitia}\xspace}

\newcommand{\aud}{${Aud}$\xspace}
\newcommand{\ct}{\ensuremath{\mathtt{CT}}\xspace}
\newcommand{\fct}{\ensuremath{f_{\st \ct}}\xspace}

\newcommand{\cl}{\ensuremath{{CL}}\xspace}

\newcommand{\ole}{\ensuremath{\mathtt{OLE}}\xspace}

\newcommand{\vopr}{\ensuremath{\mathtt{VOPR}}\xspace}

\newcommand{\zspa}{\ensuremath{\mathtt{ZSPA}}\xspace}
\newcommand{\zspaa}{\ensuremath{\mathtt{ZSPA\text{-}A}}\xspace}

\newcommand{\negl}{\ensuremath{\epsilon}\xspace}

\newcommand{\Smin}{\ensuremath{S_{\st min}}\xspace}
\newcommand{\Smax}{\ensuremath{S_{\st max}}\xspace}

\usepackage{boxedminipage}

\newcommand{\mkver}{\ensuremath{\mathtt{MT.verify}}\xspace}
\newcommand{\mkgen}{\ensuremath{\mathtt{MT.genTree}}\xspace}
\newcommand{\mkprove}{\ensuremath{\mathtt{MT.prove}}\xspace}
\newcommand{\comcom}{\ensuremath{\mathtt{Com}}\xspace}
\newcommand{\comver}{\ensuremath{\mathtt{Ver}}\xspace}
\usepackage[table]{xcolor}
\usepackage{amssymb}
\usepackage[most]{tcolorbox}
\newtcolorbox{mybox}[2][]{%
  attach boxed title to top center
               = {yshift=-10pt},
  colframe     =black,
  colbacktitle = black,
  title        = #2,#1,
  enhanced,
}

\newcommand{\yc}{\ensuremath\ddot{y}\xspace}
\newcommand{\vc}{\ensuremath\ddot{v}\xspace}
\newcommand{\lc}{\ensuremath\ddot{l}\xspace}
\newcommand{\rc}{\ensuremath\ddot{r}\xspace}
\newcommand{\fc}{\ensuremath\ddot{f}\xspace}
\newcommand{\bc}{\ensuremath\ddot{b}\xspace}
\newcommand{\cc}{\ensuremath\ddot{c}\xspace}
\newcommand{\chc}{\ensuremath\ddot{ch}\xspace}
\newcommand{\dc}{\ensuremath\ddot{d}\xspace}
\newcommand{\tc}{\ensuremath\ddot{t}\xspace}
\newcommand{\wc}{\ensuremath\ddot{w}\xspace}
\newcommand{\xc}{\ensuremath\ddot{x}\xspace}
\newcommand{\xci}{\ensuremath\ddot{x}_{\st i}\xspace}
\newcommand{\depsc}{\ensuremath\ddot{deps}\xspace}
\newcommand{\rewci}{\ensuremath\ddot{rew}_{\st i}\xspace}
\newcommand{\Xc}{\ensuremath\ddot{X}\xspace}

\newcommand{\SCpc}{\ensuremath{\mathcal{SC}_{\pc}}\xspace}
\newcommand{\SCtc}{\ensuremath{\mathcal{SC}_{\tcx}}\xspace}
\newcommand{\SCcc}{\ensuremath{\mathcal{SC}_{\ccx}}\xspace}

\newcommand{\pc}{\resizeS{{PC}}}
\newcommand{\tcx}{\resizeS{{TC}}}
\newcommand{\ccx}{\resizeS{{CC}}}
\usepackage{enumitem}
\usepackage{makeidx}  % allows for indexgeneration
\usepackage{amsmath}

\usepackage{amssymb}
\usepackage[ruled,linesnumbered]{algorithm2e}
\usepackage{subfig}
\usepackage{graphicx}
\usepackage{framed}
\usepackage{esvect}
\usepackage{tikz}
\usepackage{latexsym}

\usepackage{tikz}
\usepackage{pgfplots}
\usepackage{pgfplots, pgfplotstable}

\usepackage{blkarray}% http://ctan.org/pkg/blkarra
\usepackage{mathtools}
\usepackage{amsmath}
\usepackage{bm}
\usepackage{adjustbox}
\usepackage{blindtext}
\usepackage{multicol}
\usepackage{cancel}

\usepackage[hidelinks]{hyperref}

\usepackage[title]{appendix}

\usepackage{multirow}

\usepackage{fullpage}

\usepackage{adjustbox}
\usepackage{blindtext}
\usepackage{multicol}
\usepackage{tablefootnote}

\newcommand{\prp}{\ensuremath{\mathtt{PRP}}\xspace}

\newcommand{\p}{$\mathcal{PSI}^{\st \mathcal{FC}}$\xspace}
\newcommand{\ep}{$\mathcal{PSI}^{\st \mathcal{FCR}}$\xspace}
\newcommand{\qdel}{\ensuremath{Q^{\st \text{Del}}}\xspace}
\newcommand{\qdelwr}{\ensuremath{Q^{\st \text{Del}}_{\st\text{R}}}\xspace}
\newcommand{\qinit}{\ensuremath{Q^{\st \text{Init}}}\xspace}
\newcommand{\qUnFAbtwr}{\ensuremath{Q^{\st \text{UF-A}}_{\st \text{R}}}\xspace}
\newcommand{\qUnFAbt}{\ensuremath{Q^{\st \text{UF-A}}}\xspace}
\newcommand{\qFAbt}{\ensuremath{Q^{\st \text{F-A}}}\xspace}

\begin{document}
  \setlength\abovedisplayskip{0pt}
  \setlength\belowdisplayskip{0pt}

\newenvironment{packed_item}{
\begin{itemize}
	\setlength{\topsep}{0pt}
	\setlength{\partopsep}{0pt}
  \setlength{\itemsep}{0pt}
  \setlength{\parskip}{0pt}
  \setlength{\parsep}{0pt}
}{\end{itemize}}

\newenvironment{packed_enum}{
\begin{enumerate}
	\setlength{\topsep}{0pt}
	\setlength{\partopsep}{0pt}
  \setlength{\itemsep}{0pt}
  \setlength{\parskip}{0pt}
  \setlength{\parsep}{0pt}
}{\end{enumerate}}

\title{Earn While You Reveal: \\ Private Set Intersection that Rewards Participants}

\author{%
Aydin Abadi\thanks{aydin.abadi@ucl.ac.uk}\hspace{1mm} 
%Steven J. Murdoch\thanks{s.murdoch@ucl.ac.uk}
\institute{University College London}}

%\author{}
%\institute{}
\date{}
\maketitle{}

\begin{abstract}
In Private Set Intersection protocols (PSIs), a non-empty result always reveals something about the private input sets of the parties. Moreover, in various variants of PSI, not all parties necessarily receive or are interested in the result. Nevertheless, to date, the literature has assumed that those parties who do not receive or are not interested in the result still contribute their private input sets to the PSI for free, although doing so would cost them their privacy. In this work, for the first time, we propose a multi-party PSI, called ``\withRew'', that \emph{rewards} parties who contribute their private input sets to the protocol.  
\withRew is efficient; it relies on symmetric key primitives and its computation and communication complexities are linear with the number of parties and set cardinality. It remains secure even if the majority of parties are corrupted by active colluding adversaries.  Along the way, we devise the first fair multi-party PSI called ``\withFai'' and also propose the notion of \emph{unforgeable polynomials} each of which can be of independent interest.

\end{abstract}
 % !TEX root =main.tex

\section{Introduction}

Secure Multi-Party Computation (MPC)  allows multiple mutually distrustful parties to jointly compute a certain functionality on their private inputs without revealing anything beyond the result. Private Set Intersection (PSI) is a subclass of MPC that aims to efficiently achieve the same security property as MPC does. %lets parties compute the intersection of their private sets without  revealing anything beyond the sets' intersection. 
PSI has numerous applications. For instance, it has  been used  in Vertical Federated Learning (VFL) \cite{LuD20}, COVID-19 contact tracing schemes  \cite{DBLP:conf/asiacrypt/DuongPT20},  remote diagnostics \cite{BrickellPSW07}, and finding leaked credentials \cite{ThomasPYRKIBPPB19}. %Apple's child safety solution to combat ``Child Sexual Abuse Material''  \cite{Apple-PSI}. %PSI has been considered by the ``Financial Action Task Force'' (FATF) as one of the vital tools for enabling collaborative analytics between financial institutions to  strengthen ``Anti-Money Laundering'' (AML)  
%and ``Countering the Financing of Terrorism'' (CFT) compliance
 %\cite{FATF}. 

There exist two facts about PSIs: (i) a non-empty result always reveals something about the parties' private input sets (i.e., the set elements that are in the intersection), and (ii) various variants of PSIs do not output the result to all parties, even in those PSIs that do,  not all of the parties are necessarily interested in it.  Given these facts, one may ask a natural question:

%\begin{center}
%\emph{Why do the parties that do not receive the result or are not interested in it participate in a PSI which would ultimately reveal some information about their private inputs?}
%\end{center}

\begin{center}
\emph{How can we incentivise the parties that do not receive the result or\\ are not interested in it to participate in a PSI?}
\end{center}

To date, the literature has not answered the above question. The literature has assumed that all parties will participate in a PSI for free and bear the privacy cost (in addition to computation and computation overheads imposed by the PSI).\footnote{Combining a multi-party PSI with a payment mechanism that always charges a buyer (who initiates the PSI computation and interested in the result) a fixed amount is problematic. Because, it (1) forces the buyer to pay even if some malicious clients affect the results' correctness letting them learn the result without letting it learn correct result, as there exists no fair multi-party PSI in the literature, and (2) forces the buyer to always pay independent  of the exact size of the intersection; if it has to pay more than what it should pay for the exact intersection, then the buyer would be discouraged to participate in the protocol. If the buyer has to pay less than what it should pay for the exact intersection, then other  clients would be discouraged to participate in the protocol.} 
In this work, we answer the above question for the first time. We present a multi-party PSI, called ``\withRew'', that allows a buyer who initiates the PSI computation (and is interested in the result) to pay other parties proportionate to the number of elements it learns about other parties' private inputs.\footnote{Anesidora is in Greek and Roman mythology an epithet of several goddesses. It means sender of gifts. We call our protocol which sends rewards (or gifts) to honest parties \withRew.}  \withRew is efficient and mainly \emph{uses symmetric key primitives}.  Its computation and communication complexities are linear with the number of parties and set cardinality. \withRew remains secure even if the majority of parties are corrupt by active adversaries which may collude with each other.

%Smart-PSI is  mainly based on symmetric key primitives, modified F-PSI as well as a game theory based approach. The latter approach is leveraged to create tension, betrayal and distrust    between the clients (who  want to collude with the buyer to increase their shares) and buyer (who wants to pay less). This is the first time a game theory based approach utilised in a PSI protocol.

%In this paper,  we provide the first efficient \emph{multi-party fair} PSI protocol (F-PSI).  It allows either all clients to get the result or if the protocol aborts in an unfair manner (where only dishonest parties learn the result), then honest parties will be financially compensated. The protocol is mainly based on symmetric key primitives and a smart contract. This is the first time a smart contract is used in a PSI protocol.

%Moreover, we provide another PSI protocol (E-PSI)  that  allows a buyer  who initiates the PSI computation (and interested in the result) to pay other parties in a fair manner, where the amount each party receives is proportional to the number of elements the buyer learns about their inputs.  Smart-PSI is  mainly based on symmetric key primitives, modified F-PSI as well as a game theory based approach. The latter approach is leveraged to create tension, betrayal and distrust    between the clients (who  want to collude with the buyer to increase their shares) and buyer (who wants to pay less). This is the first time a game theory based approach utilised in a PSI protocol.

We develop \withRew in a modular fashion. Specifically, we propose the formal notion of ``PSI with Fair Compensation'' (\p) and devise the first construction, called ``\withFai'', that realises the notion.\footnote{\withFai is the ancient Roman personification of justice. We call our protocol which ensures that parties are treated fairly \withFai.} \p ensures that either all parties get the result or if the protocol aborts in an unfair manner (where only dishonest parties learn the result), then honest parties will receive financial compensation, i.e., adversaries are penalised. Justitia is the first fair multi-party PSI. \withFai itself relies on our new primitive called unforgeable polynomials, that can be of independent interest.

Next, we enhance \p to the notion of ``PSI with Fair Compensation and Reward'' (\ep) and develop \withRew that realises \ep. The latter notion ensures that honest parties (a) are rewarded regardless of whether all parties are honest, or a set of them aborts in an unfair manner and (b) are compensated in the case of an unfair abort. We formally prove the two PSIs using the simulation-based model. To devise efficient PSIs, we have developed a primitive, called ``unforgeable polynomial'' that might be of independent interest. 

A PSI, like Anesidora, that supports more than two parties and rewards set contributors can create opportunities for much richer analytics and incentivise parties to participate. It can be used (1) by an advertiser who wants to conduct advertisements targeted at certain customers by first finding their common shopping patterns distributed across different e-commerce companies' databases \cite{IonKNPSS0SY20}, (2) by a malware detection service that allows a party to send a query to a collection of malware databases held by independent antivirus companies to find out whether all of them consider a certain application as malware \cite{TamrakarLPEPA17}, or (3) by a bank, like ``WeBank'', that uses VFL and PSI to gather information about certain customers from various partners (e.g., national electronic invoice and other financial institutions) to improve its risk management of loans \cite{ChengLCY20}. In all these cases, the set contributors will be rewarded by such a PSI.

We hope that our work initiates future research on developing reward mechanisms for participants of \emph{generic MPC}, as well. Such reward mechanisms have the potential to increase MPC's real-world adoption.  

\begin{paragraph}
{\textbf{Our Contributions Summary.}} In this work, we: (1) devise \withRew, the first PSI that lets participants receive a reward for contributing their set elements to the intersection, (2) develop \withFai, the first PSI that lets either all parties receive the result or if the protocol aborts in an unfair manner,  honest parties receive compensation, and (3) propose formal definitions of the above constructions.
\end{paragraph}

%We develop the PSI that rewards participants in a modular fashion; specifically, first, we propose (i) the notion of ``PSI with fair compensation'' (\p) and (ii) devise a construction, F-PSI, that realises the notion. \p ensures that either all clients get the result or if the protocol aborts in an unfair manner (where only dishonest parties learn the result), then honest parties will be financially compensated, i.e., adversaries are penalised. Then, we enhance \p to the notion of ``PSI with fair compensation and reward'' (\ep) and develop a construction that realises \ep. The latter notion ensures that honest parties are (1) rewarded regardless of whether all parties are honest or a set of them aborts in an unfair manner and (2) compensated in the case of an unfair abort. 

% !TEX root =main.tex

\section{Related Work}\label{sec::related-work}

Since their introduction in \cite{DBLP:conf/eurocrypt/FreedmanNP04}, various PSIs have been designed. PSIs can be divided into \textit{traditional} and \textit{delegated} ones.

In \textit{traditional} PSIs, data owners interactively compute the result using their local data. 
Very recently, Raghuraman and Rindal \cite{RaghuramanR22} proposed two two-party PSIs, one secure against semi-honest/passive and the other against malicious/active adversaries. To date, these two protocols are the fastest two-party PSIs. They mainly rely on  Oblivious Key-Value Stores (OKVS) data structure and Vector Oblivious Linear Evaluation (VOLE). The protocols' computation cost is $O(c)$, where $c$ is  a set's cardinality.  They also impose $O(c\log c^{\st 2}+\kappa)$ and $O(c\cdot \kappa)$ communication costs in the semi-honest and malicious security models respectively, 
where 
%$l$ is a set element's bit-size, and  
$\kappa$ is a security parameter.  
Also, researchers designed PSIs that  allow multiple (i.e., more than two) parties to efficiently compute the intersection. The multi-party PSIs in  \cite{DBLP:conf/scn/InbarOP18,DBLP:conf/ccs/KolesnikovMPRT17} are secure against  passive adversaries while those in \cite{Ben-EfraimNOP21,GhoshN19,ZhangLLJL19,DBLP:conf/ccs/KolesnikovMPRT17,NevoTY21} were designed to remain secure against  active ones. However, Abadi \et  \cite{AbadiMZ21} showed that the PSIs in  \cite{GhoshN19} are susceptible to several attacks.  To date, the  protocols  in   \cite{DBLP:conf/ccs/KolesnikovMPRT17} and  \cite{NevoTY21} are the most  efficient multi-party PSIs  designed to be  secure against passive and active  adversaries respectively. These protocols are secure even if  the majority of parties are corrupt.  
%
%The PSIs in  \cite{DBLP:conf/ccs/KolesnikovMPRT17,NevoTY21} to keep the overall costs low and avoid requiring all clients to interact with each other in all steps of the protocols, use a ``leader'' client which interacts with the rest of the clients. 
%
The former relies on inexpensive symmetric key primitives such as  Programmable Pseudorandom Function (OPPRF) and Cuckoo Hashing, while the latter mainly uses OPPRF and OKVS. 

The overall computation and communication complexities of the PSI in  \cite{DBLP:conf/ccs/KolesnikovMPRT17} are  $O(c\cdot m^{\st 2}+c\cdot m )$ and $O(c\cdot m^{\st 2})$ respectively, as each client needs to interact with the rest (in the ``Conditional Zero-Sharing'' phase), where $m$ is the number of clients. Later, to achieve efficiency, Chandran \et \cite{ChandranD0OSS21} proposed a multi-party PSI that remains secure only if the minority of the parties is corrupt by a semi-honest adversary; thus, it offers a weaker security guarantee than  the one in \cite{DBLP:conf/ccs/KolesnikovMPRT17} does. The PSI in \cite{NevoTY21} has a parameter $t$ that determines how many parties can collude with each other and must be set before the protocol's execution, where $t\in [2, m)$. The protocol divides the parties into three groups, clients: $A_{\st 1},..A_{\st m-t-1}$, leader: $A_{\st m-t}$, and servers: $A_{\st m-t+1},..A_{\st m}$. Each client needs to send a set of messages to every server and the leader which jointly compute the final result. Hence, this protocol's overall computation and communication complexities are $O(c\cdot \kappa(m+t^{\st 2}-t(m+1)))$ and $O(c\cdot m\cdot \kappa)$ respectively.

Dong \et proposed a ``fair'' two-party PSIs \cite{DBLP:conf/dbsec/DongCCR13} that ensure either both parties receive the result or neither does, even if a malicious party aborts prematurely during the protocol's execution. The protocol relies on homomorphic encryption,  zero-knowledge proofs, and polynomial representation of sets. The protocol's  computation and communication complexities are $O(c^{\st 2})$ and $O(c)$  respectively. Since then, various fair two-party PSIs have been proposed, e.g.,  in \cite{DebnathD14,DebnathD16-,DebnathD16}. To date, the fair PSI in \cite{DebnathD16} has better complexity and performance compared to previous fair PSIs. It mainly uses  ElGamal encryption, verifiable encryption, and zero-knowledge proofs. The protocol's computation and communication cost is $O(c)$. However, its overall overhead is still high, as it relies on asymmetric key primitives, e.g.,  zero-knowledge proofs. So far, there exists no fair \emph{multi-party} PSI in the literature. Our \withFai is the first  fair multi-party PSI (that also relies on symmetric key primitives).%, which is also efficient.  

 \textit{Delegated} PSIs use  cloud computing  for computation and/or storage, while preserving the privacy of  the computation inputs and outputs from the cloud. They can be divided further into protocols that support \textit{one-off} and \textit{repeated} delegation of PSI computation. The former like \cite{kamarascaling,kerschbaum12,c18} cannot reuse their outsourced encrypted data and require clients to re-encode their data locally for each  computation. The most efficient such protocol is \cite{kamarascaling}, which has been designed for the two-party setting and its computation and communication complexity is $O(c)$.  In contrast, those protocols that support repeated PSI delegation let clients outsource the storage of their encrypted data to the cloud only once, and then execute an unlimited number of computations on the outsourced data. 

The protocol in \cite{eopsi} is the first PSI that efficiently support repeated delegation in the semi-honest model. It relies on the polynomial representation of sets, pseudorandom function, and hash table. Its overall communication and computation complexities are $O(h\cdot d^{\st 2})$ and $O(h\cdot d)$ respectively, where $h$ is the total number of bins in the hash table, $d$ is a bin's capacity (often $d=100$), and $h\cdot d$ is linear with $c$.  
Recently, a multi-party PSI that supports repeated delegation and efficient \emph{updates} has been proposed in \cite{AbadiDMT22}. It allows a party to efficiently update its outsourced set securely. It is also in the semi-honest model and uses a pseudorandom function, hash table, and Bloom filter. The protocol imposes $O(h\cdot d^{\st 2}\cdot m)$ and $O(h\cdot d\cdot m)$  computation and communication costs respectively, during the PSI computation. It also imposes $O(d^{\st 2})$  computation and communication overheads, during the update phase.  Its runtime, during the PSI computation, is up to two times faster than the  PSI in \cite{eopsi}. It is yet to be seen how a fair delegated (two/multi-party) PSI can be designed.

% !TEX root =main.tex

\section{Notations and Preliminaries}

\subsection{Notations}

Table \ref{table:notation-table} summarises the main notations used in this paper. 

% !TEX root =main.tex

\begin{table*}[!htb]
\begin{scriptsize}
\begin{center}
\footnotesize{
\caption{ \small{Notation Table}.}\label{commu-breakdown-party} 
\renewcommand{\arraystretch}{.7}
\scalebox{.837}{
% 1st table
\begin{tabular}{|c|c|c|c|c|c|c|c|c|c|c|c|c|c|} 

\hline 
\multicolumn{1}{|c|}{\cellcolor{yellow!10}\scriptsize{Setting}}&\cellcolor{yellow!10} \scriptsize{Symbol}&\cellcolor{yellow!10} \scriptsize{Description}\\
\hline 

%%%%%%%%%. Generic %%%%%

\cellcolor{yellow!10}&\cellcolor{gray!20}\scriptsize$p$&\cellcolor{gray!20}\scriptsize \text{Large prime number}\\   
\cellcolor{yellow!10}&\cellcolor{white!20}\scriptsize{$\mathbb{F}_{\st p}$}&\cellcolor{white!20}\scriptsize \text{A finite field of prime order $p$}\\   
\cellcolor{yellow!10}&\cellcolor{gray!20}\scriptsize$\cl$&\cellcolor{gray!20}\scriptsize \text{Set of all clients, $\{ A_{\st 1},...,   A_{\st m},  D\}$ }\\   
\cellcolor{yellow!10}&\cellcolor{white!20}\scriptsize$D$&\cellcolor{white!20}\scriptsize \text{Dealer client}\\   
\cellcolor{yellow!10}&\cellcolor{gray!20}\scriptsize$A_{\st m}$&\cellcolor{gray!20}\scriptsize{Buyer client}\\ 
\cellcolor{yellow!10}&\cellcolor{white!20}\scriptsize$m$&\cellcolor{white!20}\scriptsize \text{Total number of clients (excluding $D$)}\\   
\cellcolor{yellow!10}&\cellcolor{gray!20}\scriptsize$\mathtt{H}$&\cellcolor{gray!20}\scriptsize \text{Hash function}\\ 
\cellcolor{yellow!10}&\cellcolor{white!20}\scriptsize$|S_{\st\cap}|$&\cellcolor{white!20}\scriptsize{Intersection size}\\ 
\cellcolor{yellow!10}&\cellcolor{gray!20}\scriptsize$\Smin$&\cellcolor{gray!20}\scriptsize{Smallest set's size}\\
\cellcolor{yellow!10}&\cellcolor{white!20}\scriptsize$\Smax$&\cellcolor{white!20}\scriptsize{Largest set's size}\\
\cellcolor{yellow!10}&\cellcolor{gray!20}\scriptsize$|$&\cellcolor{gray!20}\scriptsize{Divisible}\\
\cellcolor{yellow!10}&\cellcolor{white!20}\scriptsize$\setminus$&\cellcolor{white!20}\scriptsize{Set subtraction}\\
\cellcolor{yellow!10}&\cellcolor{gray!20}\scriptsize$c$&\cellcolor{gray!20}\scriptsize{Set's cardinality}\\ 
 \cellcolor{yellow!10}&\cellcolor{white!20}\scriptsize$h$&\cellcolor{white!20}\scriptsize{Total number of bins in a hash table}\\ 
 \cellcolor{yellow!10}&\cellcolor{gray!20}\scriptsize$d$&\cellcolor{gray!20}\scriptsize{A bin's capacity}\\ 
\cellcolor{yellow!10}&\cellcolor{white!20}\scriptsize$\lambda$ &\cellcolor{white!20}\scriptsize Security parameter  \\  

\cellcolor{yellow!10}&\cellcolor{gray!20}\scriptsize \ole&\cellcolor{gray!20}\scriptsize{Oblivious Linear Evaluation}\\ 
\cellcolor{yellow!10}&\cellcolor{white!20}\scriptsize$\ole^{\st +}$&\cellcolor{white!20}\scriptsize{Advanced \ole}\\ 

\cellcolor{yellow!10}&\cellcolor{gray!20}\scriptsize$\comcom$&\cellcolor{gray!20}\scriptsize \text{Commitment algorithm of commitment}\\ 

\cellcolor{yellow!10}&\cellcolor{white!20}\scriptsize$\comver$&\cellcolor{white!20}\scriptsize \text{Verification algorithm of commitment}\\ 

\cellcolor{yellow!10}&\cellcolor{gray!20}\scriptsize$\mkgen$&\cellcolor{gray!20}\scriptsize \text{Tree construction algorithm of Merkle tree}\\ 

\cellcolor{yellow!10}&\cellcolor{white!20}\scriptsize$\mkprove$&\cellcolor{white!20}\scriptsize \text{Proof generation algorithm of Merkle tree}\\ 

\cellcolor{yellow!10}&\cellcolor{gray!20}\scriptsize$\mkver$&\cellcolor{gray!20}\scriptsize \text{Verification algorithm of Merkle tree}\\ 

\cellcolor{yellow!10}&\cellcolor{white!20}\scriptsize{\ct}&\cellcolor{white!20}\scriptsize \text{Coin tossing protocol}\\  

 \cellcolor{yellow!10}&\cellcolor{gray!20}\scriptsize{\vopr}&\cellcolor{gray!20}\scriptsize \text{Verifiable Oblivious Poly. Randomization}\\
 
 \cellcolor{yellow!10} &\cellcolor{white!20}\scriptsize{\zspa}&\cellcolor{white!20}\scriptsize \text{ Zero-sum Pseudorandom Values Agreement}\\
  
  \cellcolor{yellow!10}   &\cellcolor{gray!20}\scriptsize{\zspaa}&\cellcolor{gray!20}\scriptsize \text{\zspa with an External Auditor}\\

   \cellcolor{yellow!10}  &\cellcolor{white!20}\scriptsize{\p}&\cellcolor{white!20}\scriptsize \text{Multi-party PSI with Fair Compensation}\\
     
    \cellcolor{yellow!10}      &\cellcolor{gray!20}\scriptsize{\ep}&\cellcolor{gray!20}\scriptsize \text{Multi-party PSI with Fair Compensation and Reward}\\

  \cellcolor{yellow!10}   &\cellcolor{white!20}\scriptsize{\fpsi}&\cellcolor{white!20}\scriptsize \text{Protocol that realises \p}\\

  \cellcolor{yellow!10}   &\cellcolor{gray!20}\scriptsize{\epsi}&\cellcolor{gray!20}\scriptsize \text{
          Protocol that realises \ep}\\

\cellcolor{yellow!10}&\cellcolor{white!20}\scriptsize$\prf$ &\cellcolor{white!20}\scriptsize  Pseudorandom function \\

  \cellcolor{yellow!10}&\cellcolor{gray!20}\scriptsize$\prp$ &\cellcolor{gray!20}\scriptsize  Pseudorandom permutation \\ 

  \cellcolor{yellow!10}   &\cellcolor{white!20}\scriptsize{$gcd$}&\cellcolor{white!20}\scriptsize \text{Greatest common divisor}\\

\cellcolor{yellow!10}\multirow{-34}{*}{\rotatebox[origin=c]{90}{\cellcolor{yellow!10}\scriptsize{ {Generic}}}}
  \cellcolor{yellow!10}   &\cellcolor{gray!20}\scriptsize{$\negl$}&\cellcolor{gray!20}\scriptsize \text{Negligible function}\\
     \hline

%%%%%%%%%%%%
\end{tabular}
}
\scalebox{.825}{
\begin{tabular}{|c|c|c|c|c|c|c|c|c|c|c|c|c|c|} 
\hline 
\multicolumn{1}{|c|}{\cellcolor{yellow!10}\scriptsize{Setting}}&\cellcolor{yellow!10} \scriptsize{Symbol}&\cellcolor{yellow!10} \scriptsize{Description}\\
\hline 
     
\cellcolor{yellow!10}&\cellcolor{white!20}\scriptsize$\SCpc$&\cellcolor{white!20}\scriptsize \text{Prisoner's Contract}\\   

\cellcolor{yellow!10}&\cellcolor{gray!20}\scriptsize$\SCcc$&\cellcolor{gray!20}\scriptsize \text{Colluder’s Contract}\\   

\cellcolor{yellow!10}&\cellcolor{white!20}\scriptsize$\SCtc$&\cellcolor{white!20}\scriptsize \text{Traitor's Contract}\\   
\cellcolor{yellow!10}&\scriptsize  \cellcolor{gray!20}\scriptsize$\cc$&\cellcolor{gray!20}\scriptsize \text{Server’s cost for computing a task}\\   
\cellcolor{yellow!10}&\cellcolor{white!20}\scriptsize$\chc$&\cellcolor{white!20}\scriptsize \text{Auditor's cost for resolving disputes
}\\   
\cellcolor{yellow!10}&\scriptsize  \cellcolor{gray!20}\scriptsize$\dc$&\cellcolor{gray!20}\scriptsize \text{Deposit a server pays to get the job}\\  
\cellcolor{yellow!10}&\cellcolor{white!20}\scriptsize$\wc$&\cellcolor{white!20}\scriptsize \text{Amount a server receives for completing the task}\\  
\multirow{-8}{*}{\rotatebox[origin=c]{90}{\cellcolor{yellow!10}\scriptsize{ {Counter}}}}
\multirow{-8}{*}{\rotatebox[origin=c]{90}{\scriptsize{ {Collusion}}}}
\multirow{-8}{*}{\rotatebox[origin=c]{90}{\scriptsize{ {Contracts}}}}
&\cellcolor{gray!20}\scriptsize$(pk, sk)$&\cellcolor{gray!20}\scriptsize \text{\scf's auditor's public-private key pair}\\  
\hline 
%%%% End of Counter Collusion Contacts %%%%%

%%%%%%%.      F-PSI      %%%%%%

%\qinit: Initiation predicate

\cellcolor{yellow!10}&\cellcolor{white!20}\scriptsize$\scf$&\cellcolor{white!20}\scriptsize {\fpsi's smart contract}\\   
\cellcolor{yellow!10}&\cellcolor{gray!20}\scriptsize$\bm\omega, \bm\omega',\bm\rho  $&\cellcolor{gray!20}\scriptsize {Random poly. of degree} $d$\\   
\cellcolor{yellow!10}&\cellcolor{white!20}\scriptsize$\bm\gamma, \bm\delta$&\cellcolor{white!20}\scriptsize {Random poly. of degree} $d+1$\\
\cellcolor{yellow!10}&\cellcolor{gray!20}\scriptsize$\bm\nu^{\st{{(C)}}}$&\cellcolor{gray!20}\scriptsize {Blinded poly. sent by each $C$ to \scf}\\ 
\cellcolor{yellow!10}&\cellcolor{white!20}\scriptsize$\bm\phi$&\cellcolor{white!20}\scriptsize {Blinded poly. encoding the intersection}\\   
\cellcolor{yellow!10}&\cellcolor{gray!20}\scriptsize$\bm\chi$&\cellcolor{gray!20}\scriptsize {Poly. sent to \scf to identify misbehaving parties}\\ 
\cellcolor{yellow!10}&\cellcolor{white!20}\scriptsize$\bar L$&\cellcolor{white!20}\scriptsize {List of identified misbehaving parties}\\ 
%

%%%%%%%%%%%%%
\cellcolor{yellow!10}&\cellcolor{gray!20}\scriptsize&\cellcolor{gray!20}\scriptsize {A portion of a party's deposit into \scf}\\   

\cellcolor{yellow!10}&\multirow{-2}{*}{\cellcolor{gray!20}\scriptsize$\yc$}&\cellcolor{gray!20}\scriptsize{transferred to honest clients if it misbehaves}\\ 
%%%%%%

\cellcolor{yellow!10}&\scriptsize$mk$&\scriptsize{Master key of \prf}\\ 
\cellcolor{yellow!10}&\cellcolor{gray!20}\scriptsize$\qinit$&\cellcolor{gray!20}\scriptsize{Initiation predicate}\\ 
\cellcolor{yellow!10}&\scriptsize$\qdel$&\scriptsize{Delivery predicate}\\ 
\cellcolor{yellow!10}&\cellcolor{gray!20}\scriptsize$\qUnFAbt$&\cellcolor{gray!20}\scriptsize{UnFair-Abort predicate}\\ 

\multirow{-14}{*}{\rotatebox[origin=c]{90}{\cellcolor{yellow!10}\scriptsize{ {\withFai (\fpsi)}}}}
\cellcolor{yellow!10}&\scriptsize$\qFAbt$&\scriptsize{Fair-Abort predicate}\\ 

\hline 
%%%%%   End of F-PSI  %%%%

%%%%%%%    E-PSI    %%%%%%

\cellcolor{yellow!10}&\cellcolor{gray!20}\scriptsize$\SCe$&\cellcolor{gray!20}\scriptsize {\epsi's smart contract} \\   
\cellcolor{yellow!10}&\cellcolor{white!20}\scriptsize$\dc'$&\cellcolor{white!20}\scriptsize {Extractor's deposit} \\
\cellcolor{yellow!10}&\cellcolor{gray!20}\scriptsize$\yc'$&\cellcolor{gray!20}\scriptsize {Each client's deposit into \scf}\\   
\cellcolor{yellow!10}&\cellcolor{white!20}\scriptsize$\lc$&\cellcolor{white!20}\scriptsize {Reward a client earns for an intersection element}\\   
\cellcolor{yellow!10}&\cellcolor{gray!20}\scriptsize$\rc$&\cellcolor{gray!20}\scriptsize {Extractor's cost for extracting an intersection element}\\  
\cellcolor{yellow!10}&\cellcolor{white!20}\scriptsize$\fc$&\cellcolor{white!20}\scriptsize {Shorthand for $\lc(m-1)$}\\ 
\cellcolor{yellow!10}&\cellcolor{gray!20}&\cellcolor{gray!20}\scriptsize{Price a buyer pays for an intersection element}\\ 
\cellcolor{yellow!10}&\multirow{-2}{*}{\cellcolor{gray!20}\scriptsize$\vc$}&\cellcolor{gray!20}\scriptsize{$\vc=m\cdot \lc+2 \rc$}\\ 
\cellcolor{yellow!10}&\scriptsize$mk'$&\scriptsize{Another master key of \prf}\\ 

\cellcolor{yellow!10}&\cellcolor{gray!20}\scriptsize$ct_{\st mk}$&\cellcolor{gray!20}\scriptsize {Encryption of $mk$ under $pk$}\\   
\cellcolor{yellow!10}&\scriptsize$\qdelwr$&\scriptsize{Delivery-with-Reward predicate}\\ 

\multirow{-12}{*}{\rotatebox[origin=c]{90}{\cellcolor{yellow!10}\scriptsize{ {\withRew (\epsi)}}}}
\cellcolor{yellow!10}&\cellcolor{gray!20}\scriptsize$\qUnFAbtwr$&\cellcolor{gray!20}\scriptsize{UnFair-Abort-with-Reward predicate}\\

\hline

%%%%%%%%%%%

\end{tabular}\label{table:notation-table}}}
\end{center}
\end{scriptsize}
\end{table*}

% !TEX root =main.tex

\subsection{Security Model}\label{sec::sec-model}

In this paper, we use the simulation-based paradigm of secure computation \cite{DBLP:books/cu/Goldreich2004} to define and prove the proposed protocols. Since both types of (static) active and passive adversaries are involved in our protocols, we will provide formal definitions for both types. In this work, we consider a static adversary, we assume there is an authenticated private (off-chain) channel between the clients and we consider a standard public blockchain, e.g., Ethereum.

 \subsubsection{Two-party Computation.} A two-party protocol $\Gamma$ problem is captured by specifying a random process that maps pairs of inputs to pairs of outputs, one for each party. Such process is referred to as a functionality denoted by  $f:\{0,1\}^{\st *}\times\{0,1\}^{\st *}\rightarrow\{0,1\}^{\st *}\times\{0,1\}^{\st *}$, where $f:=(f_{\st 1},f_{\st 2})$. For every input pair $(x,y)$, the output pair is a random variable $(f_{\st 1} (x,y), f_{\st 2} (x,y))$, such that the party with input $x$ wishes to obtain $f_{\st 1} (x,y)$ while the party with input $y$ wishes to receive $f_{\st 2} (x,y)$. When $f$ is deterministic, then $f_{\st 1} =f_{\st 2}$. In the setting where $f$ is asymmetric and only one party (say the first one) receives the result, $f$ is defined as $f:=(f_{\st 1}(x,y), \bot)$.

 \subsubsection{Security in the Presence of Passive Adversaries.}  In the passive adversarial model, the party corrupted by such an adversary correctly follows the protocol specification. Nonetheless, the adversary obtains the internal state of the corrupted party, including the transcript of all the messages received, and tries to use this to learn information that should remain private. Loosely speaking, a protocol is secure if whatever can be computed by a party in the protocol can be computed using its input and output only. In the simulation-based model, it is required that a party’s view in a protocol's 
 execution can be simulated given only its input and output. This implies that the parties learn nothing from the protocol's execution. More formally, party $i$’s view (during the execution of $\Gamma$) on input pair  $(x, y)$ is denoted by $\mathsf{View}_{\st i}^{\st \Gamma}(x,y)$ and equals $(w, r^{\st i}, m_{\st 1}^{\st i}, ..., m_{\st t}^{\st i})$, where $w\in\{x,y\}$ is the input of $i^{\st th}$ party, $r_{\st i}$ is the outcome of this party's internal random coin tosses, and $m_{\st j}^{\st i}$ represents the $j^{\st th}$ message this party receives.  The output of the $i^{\st th}$ party during the execution of $\Gamma$ on $(x, y)$ is denoted by $\mathsf{Output}_{\st i}^{\st \Gamma}(x,y)$ and can be generated from its own view of the execution.  The joint output of both parties is denoted by $\mathsf{Output}^{\st \Gamma}(x,y):=(\mathsf{Output}_{\st 1}^{\st \Gamma}(x,y), \mathsf{Output}_{\st 2}^{\st \Gamma}(x,y))$.

\begin{definition}
Let $f$ be the deterministic functionality defined above. Protocol $\Gamma$ security computes $f$ in the presence of a static  passive adversary if there exist polynomial-time algorithms $(\mathsf {Sim}_{\st 1}, \mathsf {Sim}_{\st 2})$ such that:
\end{definition}

  \begin{equation*}
  \{\mathsf {Sim}_{\st 1}(x,f_{\st 1}(x,y))\}_{\st x,y}\stackrel{c}{\equiv} \{\mathsf{View}_{\st 1}^{\st \Gamma}(x,y) \}_{\st x,y}\\
  \end{equation*}
  \begin{equation*}
    \{\mathsf{Sim}_{\st 2}(y, f_{\st 2}(x,y))\}_{\st x,y}\stackrel{c}{\equiv} \{\mathsf{View}_{\st 2}^{\st \Gamma}(x,y) \}_{\st x,y}
  \end{equation*}

 \subsubsection{Security in the Presence of Active Adversaries.}  In this adversarial model, the corrupted party may
arbitrarily deviate from the protocol specification, to learn the private inputs of the other parties or to influence the outcome of the computation. In this case,  the adversary may not use the input provided. Therefore, beyond the possibility that a corrupted party may learn more than it should, correctness is also required. This means that a corrupted party must not be able to cause the output to be incorrectly distributed. Moreover, we require independence of inputs meaning that a corrupted party cannot make its input depend on the other party’s input. To capture the threats,
the security of a protocol is analyzed by comparing what an adversary can do in the real protocol to what it can do in an ideal scenario that is secure by definition. This is formalized by considering an ideal computation involving an incorruptible Trusted Third Party (TTP) to whom the parties send their inputs and receive the output of the ideal functionality. Below, we describe the executions in the ideal and real models. 
 
First, we describe the execution in the ideal model. Let $P_{\st 1}$ and $P_{\st 2}$ be the parties participating in the
protocol, $i\in \{0, 1\}$ be the index of the corrupted party, and $\mathcal A$ be a non-uniform
probabilistic polynomial-time adversary. Also, let $z$ be an auxiliary input given to $\mathcal A$ while  $x$ and $y$ be the input of party $P_{\st 1}$ and $P_{\st 2}$  respectively.  The honest party, $P_{\st j}$, sends its received input to TTP.  The corrupted party $P_{\st i}$ may either abort (by replacing the input with a special abort message $\Lambda_{\st i}$),  send its received input or send some other input of the same length to TTP. This decision is made by the adversary and may depend on the input value of $P_{\st i}$ and $z$. If TTP receives $\Lambda_{\st i}$, then it sends $\Lambda_{\st i}$ to the honest party and the ideal execution terminates.  Upon obtaining an input pair $(x, y)$, TTP computes $f_{\st 1}(x, y)$ and $f_{\st 2}(x, y)$. It first sends $f_{\st i}(x, y)$ to  $P_{\st i}$ which replies with ``continue'' or $\Lambda_{\st i}$. In the former case, TTP sends  $f_{\st j}(x, y)$ to  $P_{\st j}$ and in the latter it sends $\Lambda_{\st i}$ to  $P_{\st j}$. The honest party always outputs the message that it obtained from TTP. A malicious party may output an arbitrary function of its initial inputs and the message it has obtained from TTP.  The ideal execution of $f$ on inputs $(x, y)$ and $z$ is denoted by $\mathsf{Ideal}^{\st f}_{\st\mathcal{A}(z), i}(x,y)$ and is defined as the output pair of the honest party and $\mathcal{A}$ from the above ideal execution.  In the real model, the real two-party protocol $\Gamma$ is executed
without the involvement of TTP. In this setting, $\mathcal{A}$ sends all messages on
behalf of the corrupted party and may follow an arbitrary strategy.
The honest party follows the instructions of $\Gamma$. The real execution of $\Gamma$ is denoted by $\mathsf{Real}^{\st \Gamma}_{\st\mathcal{A}(z), i}(x,y)$, it is defined as the joint output of the parties engaging in the real execution of $\Gamma$ (on the inputs), in the presence of $\mathcal{A}$.

 Next, we define security. At a high level, the definition states that a secure protocol in the real model emulates the ideal model. This is formulated by stating that adversaries in the ideal model can simulate executions of the protocol in the real model. 
 
\begin{definition}\label{def::MPC-active-adv}
Let $f$ be the two-party functionality defined above and $\Gamma$ be a two-party protocol that computes $f$.   Protocol $\Gamma$ securely computes $f$ with abort in the presence of static active adversaries if for every non-uniform probabilistic polynomial time adversary $\mathcal{A}$ for the real model, there exists a non-uniform probabilistic polynomial-time adversary (or simulator) $\mathsf{Sim}$ for the ideal model, such that for every $i\in \{0,1\}$, it holds that: 

\begin{equation*}
\{\mathsf {Ideal}^{\st f}_{\st \mathsf{Sim}(z), i}(x,y)\}_{\st x,y,z}\stackrel{c}{\equiv} \{\mathsf{Real}_{\st \mathcal{A}(z), i}^{\st \Gamma}(x,y) \}_{\st x,y,z}
\end{equation*}
\end{definition}

\subsection{Smart Contracts}

Cryptocurrencies, such as Bitcoin \cite{bitcoin} and Ethereum \cite{ethereum}, beyond offering a decentralised currency,  support computations on transactions. In this setting, often a certain computation logic is encoded in a computer program, called a \emph{``smart contract''}. To date, Ethereum is the most predominant cryptocurrency framework that enables users to define arbitrary smart contracts. In this framework,  contract code is stored on the blockchain and executed by all parties (i.e., miners) maintaining the cryptocurrency,  when the program inputs are provided by transactions. The program execution's correctness is guaranteed by the security of the underlying blockchain components. To prevent a denial-of-service attack, the framework requires a transaction creator to pay a  fee, called \emph{``gas''}, depending on the complexity of the contract running on it. 

\subsection{Counter Collusion Smart Contracts}\label{Counter-Collusion-Smart-Contracts}

In order to let a party, e.g., a client, efficiently delegate a computation to a  couple of potentially colluding third parties, e.g., servers, Dong   \et \cite{dong2017betrayal} proposed two main smart contracts; namely, ``Prisoner's Contract'' ($\SCpc$) and ``Traitor's Contract'' (\SCtc).  
The Prisoner's contract is signed by the client and the servers. This contract tries to incentivize correct computation by using the following idea. It requires each server to pay a deposit before the computation is delegated. It is equipped with an external auditor that is invoked to detect a misbehaving server only when the servers provide non-equal results.

If a server behaves honestly, then it can withdraw its deposit. Nevertheless, if a cheating server is detected by the auditor, then (a portion) of its deposit is transferred to the client. If one of the servers is honest and the other one cheats, then the honest server receives a reward taken from the cheating server's deposit. However, the dilemma, created by \SCpc between the two servers, can be addressed if they can make an enforceable promise, say via a ``Colluder's Contract'' (\SCcc),  in which one party, called ``ringleader'', would pay its counterparty a bribe if both follow the collusion and provide an incorrect computation to \SCpc. To counter \SCcc, Dong   \et proposed \SCtc, which incentivises a colluding server to betray the other server and report the collusion without being penalised by \SCpc. In this work, we slightly adjust and use these contracts. We have stated the related parameters of these tree contracts in Table \ref{table:notation-table}. We refer readers to Appendix \ref{appendix::Counter-Collusion-Contracts} for the full description of the parameters and contracts.

%\begin{itemize}
%\item[$\bullet$] $\bc$: the bribe paid by the ringleader of the collusion to the other
%server in the collusion agreement, in the Colluder’s contract.
%%
%\item[$\bullet$] $\cc$: a server’s cost for computing the task.
%%
%\item[$\bullet$] $\chc$: the fee paid to to invoke an auditor for recomputing a task and resolving
%disputes.
%%
%\item[$\bullet$] $\dc$: the deposit a server needs to pay to be eligible for getting the job.
%%
%\item[$\bullet$] $\tc$: the deposit the colluding parties need to pay in the collusion agreement, in the Colluder’s contract.
%%
%\item[$\bullet$] $\wc$: the amount that a server receives for completing the task.
%%
%\item[$\bullet$] $\wc \geq \cc$: the server would not accept underpaid jobs.
%%
%\item[$\bullet$] $\chc > 2\wc$: If it does not hold, then there would be no need to use the servers and the auditor would do the computation.
%%
%\item [$\bullet$] $(pk,sk)$: an asymmetric-key encryption's public-private key pair belonging to the auditor. 
%\end{itemize}
%\noindent The following relations need to hold when setting the contracts
%in order for the desirable equilibria to hold:
%%
%(i) $\dc>\cc+\chc$, (ii) $\bc<\cc$, and (iii) $\tc<\wc-\cc + 2\dc - \chc -\bc$.
%

%

\subsection{Pseudorandom Function and Permutation}

Informally, a pseudorandom function is a deterministic function that takes a key of length $\lambda$ and an input; and outputs a value  indistinguishable from that of  a truly random function.  In this paper, we use pseudorandom functions:   $\mathtt {PRF}: \{0,1\}^{\st \lambda}\times \{0,1\}^{\st *} \rightarrow  \mathbb{F}_{\st p}$, where $\log_{\st 2}(p)=\lambda$ is the security parameter. In practice, a pseudorandom function can be obtained from an efficient block cipher \cite{DBLP:books/crc/KatzLindell2007}.

The definition of a pseudorandom permutation, $\mathtt {PRP}: \{0,1\}^{\st \lambda}\times \{0,1\}^{\st *} \rightarrow  \mathbb{F}_{\st p}$, is very similar to that of a pseudorandom function, with a difference; namely, it is required the keyed function $\PRP(k,.)$ to be indistinguishable from a uniform permutation, instead of a uniform function. %In cryptographic schemes that involve $\PRP$, sometimes honest parties may require to compute the inverse of pseudorandom permutation, i.e., $\mathtt {PRP}^{\st -1}(k, .)$, as well. In this case, it would require that $\PRP(k,.)$ be indistinguishable from a uniform permutation even if the distinguisher is additionally given oracle access to the inverse of the permutation. 

%\subsection{Random Extraction Beacon}
%\subsection{Commitment Scheme}
% !TEX root =main.tex

\subsection{Commitment Scheme}\label{subsec:commit}

 A commitment scheme involves a  \emph{sender} and a \emph{receiver}. It also  involves  two phases; namely, \emph{commit} and  \emph{open}. In the commit phase, the sender  commits to a message: $x$ as $\comcom(x,r)=com$, that involves a secret value: $r\stackrel{\st\$}\leftarrow \{0,1\}^{\st\lambda}$. At the end of the commit phase,  the commitment ${com}$ is sent to the receiver. In the open phase, the sender sends the opening $\hat{x}:=(x, r)$ to the receiver who verifies its correctness: $\comver({com},\hat{x})\stackrel{\st ?}=1$ and accepts if the output is $1$.  A commitment scheme must satisfy two properties: (a) \textit{hiding}: it is infeasible for an adversary (i.e., the receiver) to learn any information about the committed  message $x$, until the commitment ${com}$ is opened, and (b) \textit{binding}: it is infeasible for an adversary (i.e., the sender) to open a commitment ${com}$ to different values $\hat{x}':=(x',r')$ than that was  used in the commit phase, i.e., infeasible to find  $\hat{x}'$, \textit{s.t.} $\comver({com},\hat{x})=\comver({com},\hat{x}')=1$, where $\hat{x}\neq \hat{x}'$.  There exist efficient  commitment schemes both in (a) the standard model, e.g., Pedersen scheme \cite{Pedersen91}, and (b)  the random oracle model using the well-known hash-based scheme such that committing  is : $\mathtt{H}(x||r)={com}$ and $\comver({com},\hat{x})$ requires checking: $\mathtt{H}(x||r)\stackrel{\st ?}={com}$, where $\mathtt{H}:\{0,1\}^{\st *}\rightarrow \{0,1\}^{\st\lambda}$ is a collision-resistant hash function, i.e., the probability to find $x$ and $x'$ such that $\mathtt{H}(x)=\mathtt{H}(x')$ is negligible in the security parameter $\lambda$.
\subsection{Hash Tables}
A hash table is an array of   bins each of which can hold a set of elements. It is accompanied by a hash function. To insert an element, we first compute the element's hash,  and then store the element in the bin whose index is the element's hash. In this paper, we set the table's parameters appropriately to ensure the number of elements in each bin does not exceed a predefined capacity. Given the maximum number of elements $c$ and the bin's maximum size $d$, we can determine the number of bins, $h$, by analysing hash tables under the balls into the bins model  \cite{DBLP:conf/stoc/BerenbrinkCSV00}. In Appendix \ref{Preliminary-Hash-Table}, we explain how the hash table parameters are set.

%\subsection{Merkel Tree}
% !TEX root =main.tex

\subsection{Merkle Tree}\label{sec::merkle-tree}

A Merkle tree is a data structure that supports a compact commitment of a set of values/blocks.  As a result, it includes two parties, prover $\mathcal{P}$ and verifier $\mathcal{V}$. The  Merkle tree scheme includes three algorithms $(\mkgen, \mkprove,$ $\mkver)$, defined as follows: 

\begin{itemize}
\item[$\bullet$] The algorithm that constructs a Merkle tree, $\mkgen$, is run by $\mathcal{V}$. It takes  blocks, $u:=u_{\st 1},...,u_{\st n}$, as input. Then, it groups the blocks  in pairs. Next,   a collision-resistant hash function, $\mathtt{H}(.)$, is used to hash each pair. After that, the hash values are grouped in pairs and each pair is further hashed, and this process is repeated until only a single hash value, called ``root'', remains. This yields a  tree with the leaves corresponding to the input blocks and the root corresponding to the last remaining hash value. $\mathcal{V}$ sends the root to $\mathcal{P}$.
\item[$\bullet$] The proving algorithm, $\mkprove$, is run by $\mathcal{P}$. It takes a block index, $i$, and a tree as inputs. It outputs a vector proof, of  $\log_{\st 2}(n)$ elements. The proof asserts the membership of $i$-th block in the tree, and consists of all the sibling nodes on a path from the $i$-th block to the root of the Merkle tree (including $i$-th block). The proof is given to $\mathcal{V}$.
\item[$\bullet$] The verification algorithm, $\mkver$, is run by $\mathcal{V}$. It takes as an input $i$-th block, a proof, and the tree's root. It checks if the $i$-th block corresponds to the root. If the verification passes, it outputs $1$; otherwise, it outputs $0$.

\end{itemize}

The Merkle tree-based scheme has two properties: \emph{correctness} and \emph{security}. Informally, the correctness requires that if both parties run the algorithms correctly, then a proof is always accepted by  $\mathcal{V}$. The security requires that a computationally bounded malicious $\mathcal{P}$ cannot convince  $\mathcal{V}$ into accepting an incorrect proof, e.g., proof for a non-member block. The security relies on the assumption that it is computationally infeasible to find the hash function's collision. Usually, for the sake of simplicity, it is assumed that the number of blocks, $n$, is a power of $2$. The height of the tree, constructed on $m$ blocks, is $\log_{\st 2}(n)$. 
\subsection{Polynomial Representation of Sets}\label{sec::poly-rep}

The idea of using a polynomial to represent a set's elements was proposed by Freedman  \et in \cite{DBLP:conf/eurocrypt/FreedmanNP04}. Since then,   the idea has been widely used,  e.g., in \cite{DBLP:conf/fc/AbadiTD16,Feather2020,GhoshS19,DBLP:conf/crypto/KissnerS05}. In this representation, set elements $S=\{s_{\st 1},...,s_{\st d}\}$ are defined over  $\mathbb{F}_{\st p}$ and  set $S$ is represented as a polynomial of   form: $\mathbf{p}(x)=\prod\limits ^{\st {d}}_{\st i=1}(x-s_{\st i})$, where $\mathbf{p}(x) \in \mathbb{F}_{\st p}[X]$ and $\mathbb{F}_{\st p}[X]$ is a polynomial ring.  Often a   polynomial,  $\mathbf{p}(x)$, of degree $d$ is  represented in the ``coefficient form'' as follows:  $\mathbf{p}(x)=a_{\st 0}+a_{\st 1}\cdot x+...+ a_{\st d}\cdot x^{\st d}$. The form $\prod\limits ^{\st {d}}_{\st i=1}(x-s_{\st i})$ is a special case of the coefficient form. As shown in \cite{BonehGHWW13,DBLP:conf/crypto/KissnerS05}, for two sets $S^{\st (A)}$ and $S^{\st (B)}$ represented by polynomials $\mathbf{p}_{\st A}$ and $\mathbf{p}_{\st B}$ respectively, their product, which is polynomial $\mathbf{p}_{\st A}\cdot \mathbf{p}_{\st  B}$,  represents the set union, while their greatest common divisor, $gcd($$\mathbf{p}_{\st A}$$,\mathbf{p}_{\st B})$, represents the set intersection. For two polynomials $\mathbf{p}_{\st A}$ and $\mathbf{p}_{\st B}$ of degree $d$, and two random polynomials $\bm\gamma_{\st A}$ and  $\bm\gamma_{\st B}$ of degree $d$, it is proven in~\cite{BonehGHWW13,DBLP:conf/crypto/KissnerS05} that: $\bm\theta=\bm\gamma_{\st A}\cdot \mathbf{p}_{\st A}+\bm\gamma_{\st B}\cdot\mathbf{p}_{\st B}=\bm\mu\cdot gcd(\mathbf{p}_{\st A},\mathbf{p}_{\st B})$, where $\bm\mu$ is a uniformly random polynomial, and polynomial $\bm\theta$ contains only information about the elements in  $S^{\st (A)}\cap S^{\st (B)}$, and contains no information about other elements in $S^{\st (A)}$ or $S^{\st (B)}$.  

Given a polynomial $\bm\theta$ that encodes sets intersection, one can find the set elements in the intersection via one of the following approaches. First, via polynomial evaluation: the party who already has one of the original input sets, say  $\mathbf{p}_{\st A}$,  evaluates $\bm\theta$ at every element $s_{\st i}$ of $\mathbf{p}_{\st A}$ and considers $s_{\st i}$ in the intersection if $\mathbf{p}_{\st A}(s_{\st i})=0$. Second,   via polynomial root extraction:   the party who does not have one of the original input sets, extracts the roots of $\bm\theta$,  which contain  the roots of (i) random polynomial  $\bm\mu$ and (ii) the polynomial that represents the intersection, i.e., $gcd(\mathbf{p}_{\st A},\mathbf{p}_{\st B})$. In this approach, to distinguish errors (i.e., roots of $\bm\mu$) from the intersection, PSIs in \cite{eopsi,DBLP:conf/crypto/KissnerS05} use the \emph{``hash-based padding technique''}. In this technique, every element $u_{\st i}$ in the set universe $\mathcal{U}$, becomes $s_{\st i}=u_{\st i}||\mathtt{H}(u_{\st i})$, where $\mathtt{H}$ is a cryptographic hash function with a sufficiently large output size. Given a field's arbitrary element, $s \in \mathbb{F}_p$ and $\mathtt{H}$'s output size $|\mathtt{H}(.)|$, we can parse $s$ into $x_{\st 1}$ and $x_{\st 2}$, such that $s=x_{\st 1}||x_{\st 2}$ and  $|x_{\st 2}|=|\mathtt{H}(.)|$. In a  PSI that uses polynomial representation and this padding technique, after we extract each root of  $\bm\theta$, say $s$, we parse it into $(x_{\st 1}, x_{\st 2})$ and check $x_{\st 2}\stackrel{?}=\mathtt{H}(x_{\st 1})$.  If the equation holds, then we consider $s$ as an element of the intersection.

%\TZ{What is meant by ``$\bm\theta$ contains only information about $S^{\st (A)}\cap S^{\st (B)}$"?}--> addressed.. 

%Polynomials can also be represented in the  ``point-value form''. In particular, a polynomial $\mathbf{p}(x)$ of degree $d$ can be represented as a set of $m$ ($m>d$) point-value pairs $\{(x_{\st 1},y_{\st 1}),...,$ $(x_{\st m},y_{\st m})\}$ such that all $x_{\st i}$ are distinct  non-zero points and $y_{\st i}=\mathbf{p}(x_{\st i})$ for all $i$, $1\le i\le m$. If  $x_{\st i}$  are fixed, then we can represent polynomials as a vector $\vv{\bm{y}}=[y_{\st 1}, ..., y_{\st m}]$. Polynomials in point-value form have  been used previously in PSIs~\cite{eopsi,opsi15,DBLP:conf/fc/AbadiTD16,Feather2020,GhoshS19,KolesnikovMPRT17}. A polynomial
%in this form can be converted into coefficient form via polynomial interpolation, e.g., using Lagrange interpolation~\cite{aho19}. Moreover,  one can add or multiply two polynomials,  in point-value form, by adding or multiplying their corresponding y-coordinates. In this case, the  polynomial interpolated from the result would be the two polynomials' addition or product. Often PSIs  that use this representation  assume that all $x_{\st i}$ are picked from $\mathbb{F} \setminus \mathcal{U}$.

\subsection{Horner's Method}
Horner's method \cite{DBLP:journals/ibmrd/Dorn62} allows for efficiently evaluating polynomials at a given point, e.g., $x_{\st 0}$. Specifically, given a polynomial of the form: $\tau(x)= a_{\st 0}+a_{\st 1}\cdot x+a_{\st 2}\cdot x^{\st 2}+...+a_{\st d}\cdot x^{\st d}$ and a point: $x_{\st 0}$, one can efficiently evaluate the polynomial at $x_{\st 0}$ iteratively, in the following fashion: $\tau(x_{\st 0})=a_{\st 0}+x_{\st 0}(a_{\st 1} + x_{\st 0}(a_{\st 2}+...+x_{\st 0}(a_{\st d-1}+x_{\st 0}\cdot a_{\st d})...)))$. Evaluating  a polynomial of degree $d$ naively requires  $d$ additions and $\frac{(d^{\st 2}+d)}{2}$ multiplications. However, using Horner's method the evaluation requires only $d$ additions and $d$ multiplications. We use this method throughout the paper.

\subsection{Oblivious Linear Function Evaluation}\label{sec::OLE-plus}

Oblivious Linear function Evaluation (\ole) is a two-party protocol that involves a sender and receiver. In \ole,  the sender  has two  inputs  $a, b\in \mathbb{F}_{\st p}$ and the receiver has a single input, $c \in \mathbb{F}_{p}$.  The protocol allows the receiver to learn only $s = a\cdot c + b \in \mathbb{F}_{\st p}$, while the sender learns nothing. Ghosh \textit{et al.} \cite{GhoshNN17} proposed an efficient \ole that has $O(1)$ overhead and involves mainly symmetric key operations. Later, in \cite{GhoshN19} an enhanced \ole, called $\ole^{\st +}$ was proposed. The latter ensures that the receiver cannot learn anything about the sender's inputs,  even if it sets its input to $0$. In this paper, we use $\ole^{\st +}$. We refer readers to Appendix \ref{apndx:F-OLE-plus}, for its construction.  %In this case, each party picks a random string, 

\subsection{Coin-Tossing Protocol}

A Coin-Tossing protocol, \ct, allows two mutually distrustful parties, say $A$ and $B$, to jointly generate a single random bit. Formally, \ct computes the functionality $\fct(in_{\st A}, in_{\st B})\rightarrow (out_{\st A}, out_{\st B})$, which takes $in_{\st A}$ and  $in_{\st B}$ as inputs of $A$ and $B$ respectively and outputs $out_{\st A}$ to $A$ and $out_{\st B}$ to $B$, where $out_{\st A}=out_{\st B}$. A basic security requirement of a \ct is that the resulting bit is (computationally) indistinguishable from a truly random bit. 

Blum proposed a simple \ct in \cite{Blum82} that works as follows. Party $A$ picks a random bit $in_{\st A}\stackrel{\st \$}\leftarrow\{0,1\}$, commits to it and sends the commitment to $B$ which sends its choice of random input, $in_{\st B}\stackrel{\st \$}\leftarrow\{0,1\}$, to $A$. Then, $A$ sends the opening of the commitment (including $in_{\st A}$) to $B$, which checks whether the commitment matches its opening. If so, each party computes the final random bit as $in_{\st A}\oplus in_{\st B}$.  

There have also been \emph{fair} coin-tossing protocols, e.g., in \cite{MoranNS09}, that ensure either both parties learn the result or nobody does. These protocols can be generalised to \emph{multi-party} coin-tossing protocols to generate a \emph{random string} (rather than a single bit), e.g., see \cite{BeimelOO10,KiayiasRDO17}.
The overall computation and communication complexities of (fair) multi-party coin-tossing protocols are often linear with the number of participants. In this paper, any secure multi-party \ct that generates a random string can be used. For the sake of simplicity, we let a multi-party \fct take $m$ inputs and output a single value, i.e., $\fct(in_{\st 1}, ..., in_{\st m})\rightarrow out$. 

%,  as an aborting party can be excluded from the next run of the protocol and the aborting party cannot learn partsets' intersection 

%\input{Priliminary}

% !TEX root =U-PSI.tex

%\section{Security Definition}
%
%In this section, we provide the security definition of our protocol. There are two kinds of party involved in the protocol. Namely, (1) a set of clients $\{\resizeT {\textit A}_{\resizeS {\textit  1}},..., \resizeT {\textit A}_{\resizeS {\textit  m}}\}$ potentially malicious (i.e. active adversaries) and all may collude with each other, and (2) a non-colluding dealer: client $\resizeT {\textit D}$, potentially semi-honest (i.e. a passive adversary). In this work, we consider static adversary,  we assume there is an authenticated private (off-chain) channel between the clients and we consider a standard public blockchain, e.g. Ethereum.
% !TEX root =main.tex

\section{Definition of Multi-party PSI with Fair Compensation}\label{sec::F-PSI-model}%\label{Fair-PSI-Protocol}

In this section, we present the notion of multi-party PSI with Fair Compensation  (\p) which allows either all clients to get the result or the honest parties to be financially compensated if the protocol aborts in an unfair manner, where only dishonest parties learn the result.

In a  $\mathcal{PSI}^{\st \mathcal{FC}}$, three types of parties are involved; namely, (1) a set of clients $\{A_{\st 1},...,A_{\st m}\}$ potentially \emph{malicious} (i.e., active adversaries) and all but one may collude with each other, (2) a non-colluding dealer, $D$, potentially semi-honest (i.e., a passive adversary) and has an input set, and (3) an auditor \aud potentially semi-honest, where all parties except \aud have input set. For simplicity, we assume that given an address one can determine whether it belongs to \aud.

The basic functionality that any multi-party PSI computes can be defined as $f^{\st\text{PSI}} (S_{\st 1},..., S_{\st m+1})\rightarrow\underbrace{(S_{\st\cap},..., S_{\st\cap})}_{\st m+1}$, where $S_{\st\cap}= S_{\st 1} \cap S_{\st 2}, ...,\cap\  S_{\st m+1}$.  To formally define a \p, we equip the above PSI functionality with four predicates,  $Q:=(\qinit, \qdel, \qUnFAbt, \qFAbt)$, which ensure that certain financial conditions are met. 
   %  that are invoked after the functionality $f^{\st \text{PSI}}$ is executed. 
   We borrow three of these predicates (i.e., $\qinit, \qdel, \qUnFAbt$) from the ``fair and robust multi-party computation'' initially proposed in \cite{KiayiasZZ16}; nevertheless, we will (i) introduce an additional predicate  \qFAbt and (ii) provide more formal accurate definitions of these predicates. 
   
Predicate \qinit specifies under which condition a protocol that realises \p should start executing, i.e., when all set owners have enough deposit. Predicate \qdel determines in which situation parties receive their output, i.e., when honest parties receive their deposit back. Predicate \qUnFAbt specifies under which condition the simulator can force parties to abort if the adversary learns the output,  i.e., when an honest party receives its deposit back plus a predefined amount of compensation. Predicate \qFAbt specifies under which condition the simulator can force parties to abort if the adversary receives no output, i.e., when honest parties receive their deposits back. We observed that the latter predicate should have been defined in the generic framework in \cite{KiayiasZZ16} too; as the framework should have also captured the cases where an adversary may abort without learning any output after the onset of the protocol.  Intuitively, by requiring any protocol that realises \p to implement a wrapped version of $f^{\st\text{PSI}}$ that includes $Q$, we will ensure that an honest set owner only aborts in an unfair manner if \qUnFAbt returns  $1$, it only aborts in a fair manner if \qFAbt returns  $1$, and outputs a valid value if \qdel returns $1$. Now, we formally define each of these predicates.

 \begin{definition}
  [\qinit: Initiation predicate] Let $\mathcal{G}$ be a stable ledger, $adr_{\st sc}$ be smart contract $sc$'s address, $Adr$ be a set of $m+1$ distinct addresses, and $\xc$ be a fixed amount of coins. Then, predicate $\qinit(\mathcal{G}, adr_{\st sc}, m+1, Adr, \xc)$ returns $1$ if every address in $Adr$ has at least $\xc$ coins in $sc$; otherwise, it returns $0$. 
 \end{definition}

    \begin{definition}  [\qdel:
    Delivery predicate] Let $pram:=(\mathcal{G}, adr_{\st sc}, \xc)$ be the parameters defined above, and   $adr_{\st i}\in Adr$ be the address of an honest party. 
    %
%    Let also $G$ be a compensation function that takes as input  two parameters $(deps, m')$, where $deps$ is the amount of coins  that all $m+1$ parties  deposit; it returns the amount of compensation each honest party must receive, i.e., $G(deps, m')\rightarrow c'$. 
    %
    Then, predicate $\qdel(pram, adr_{\st i})$ returns $1$ if $adr_{\st i}$ has sent $\xc$ amount to $sc$ and received  $\xc$ amount from it; thus,  its balance in $sc$ is $0$. Otherwise, it returns $0$. 
  \end{definition}

   \begin{definition}  [\qUnFAbt: UnFair-Abort predicate]
 Let $pram:=(\mathcal{G}, adr_{\st sc}, \xc)$ be the parameters defined above, and $Adr'\subset Adr$ be a set containing honest parties' addresses, $m' = |Adr'|$,  and   $adr_{\st i}\in Adr'$. Let also $G$ be a compensation function that takes as input  three parameters $(\depsc, adr_{\st i}, m')$, where $\depsc$ is the amount of coins  that all $m+1$ parties  deposit. It returns the amount of compensation each honest party must receive, i.e., $G(\depsc, ard_{\st i}, m')\rightarrow \xci$. Then, predicate $\qUnFAbt$ is defined as $\qUnFAbt(pram, G, \depsc, m', adr_{\st i})\rightarrow (a,b)$, where $a=1$ if $adr_{\st i}$ is an honest party's address and $adr_{\st i}$ has sent $\xc$ amount to $sc$ and received  $\xc+\xci$  from it, and $b=1$ if $adr_{\st i}$ is \aud's address and $adr_{\st i}$ received $\xci$  from $sc$. Otherwise, $a=b=0$. 
  \end{definition}

\begin{definition}  [\qFAbt: Fair-Abort predicate]
 Let $pram:=(\mathcal{G}, adr_{\st sc}, \xc)$ be the parameters defined above, and $Adr'\subset Adr$ be a set containing honest parties' addresses, $m' = |Adr'|$,     $adr_{\st i}\in Adr'$, and  $adr_{\st j}$ be \aud's address. Let $G$ be the compensation function, defined above and let $G(deps, ard_{\st j}, m')\rightarrow \xc_{\st j}$ be the compensation that the auditor must receive.  Then, predicate $\qFAbt(pram, G, \depsc, m', adr_{\st i}, adr_{\st j})$ returns $1$, if $adr_{\st i}$ (s.t. $adr_{\st i}\neq adr_{\st j}$) has sent $\xc$ amount to $sc$ and received  $\xc$  from it, and $adr_{\st j}$ received $\xc_{\st j}$  from $sc$. Otherwise, it returns $0$. 
 \end{definition}

 Next, we present a formal definition of \p. %Note that we have upgraded the simulation-based definition of secure computation (i.e., Definition \ref{def::MPC-active-adv}) to define the security requirements of \p, by incorporating the above predicates into the definition. 
 
\begin{definition}[\p]\label{def::PSI-Q-fair}
Let $f^{\st \text{PSI}}$ be the multi-party PSI functionality defined above. We say  protocol $\Gamma$ realises  $f^{\st \text{PSI}}$ with $Q$-fairness in the presence of $m-1$ static active-adversary clients (i.e., $A_{\st j}$s) or a static passive dealer $D$ or passive auditor $Aud$, if for every non-uniform probabilistic polynomial time adversary $\mathcal{A}$ for the real model, there exists a non-uniform probabilistic polynomial-time adversary (or simulator) $\mathsf{Sim}$ for the ideal model, such that for every $I\in \{A_{\st 1},...,A_{\st m}, D, Aud\}$, it holds that: 

\begin{equation*}
\{\mathsf {Ideal}^{\st \mathcal{W}(f^{\st \text{PSI}},Q)}_{\st \mathsf{Sim}(z), I}(S_{\st 1},..., S_{\st m+1})\}_{\st S_{\st 1},..., S_{\st m+1},z}\stackrel{c}{\equiv} \{\mathsf{Real}_{\st \mathcal{A}(z), I}^{\st \Gamma}(S_{\st 1},..., S_{\st m+1}) \}_{\st S_{\st 1},..., S_{\st m+1},z}
\end{equation*}
where  $z$ is an auxiliary input given to $\mathcal{A}$ and  $\mathcal{W}(f^{\st \text{PSI}},Q)$ is a functionality that wraps $f^{\st \text{PSI}}$ with predicates $Q:=(\qinit, \qdel, \qUnFAbt, \qFAbt)$. 
  \end{definition}

\section{Other Subroutines Used in \withFai}\label{sec::subroutines}
In this section, we present three subroutines and a primitive that we developed and are used in the instantiation of \p, i.e., \withFai. 

\subsection{Verifiable Oblivious Polynomial Randomisation (\vopr)}\label{sec::vopr}

%In this section, we present ``Verifiable Oblivious Polynomial Randomisation'' (VOPR) protocol. 

In the \vopr, two parties are involved, (i) a sender which is potentially a passive adversary and (ii) a receiver that is potentially an active adversary. The protocol allows the receiver with input polynomial $\bm\beta$ (of degree $e'$) and the sender with input random polynomials $\bm\psi$ (of degree $e$) and  $\bm{\alpha}$ (of degree $e+e'$)   to compute: $\bm\theta=\bm\psi\cdot \bm\beta+\bm\alpha$, such that (a) the receiver learns only $\bm\theta$ and nothing about the sender's input even if it sets $\bm \beta=0$, (b) the sender learns nothing, and (c) the receiver's misbehaviour is detected in the protocol. Thus, the functionality that  \vopr computes is defined as $f^{\st {\vopr}}( (\bm\psi, \bm{\alpha}), \bm\beta)\rightarrow(\bot, \bm\psi\cdot \bm\beta+\bm\alpha)$. 
We will use {\vopr} in \withFai for two main reasons:  (a) to let a party re-randomise its counterparty's polynomial (representing its set) and (b) to impose a MAC-like structure to the randomised polynomial; such a structure will allow a verifier to detect if \vopr's output has been modified. 

Now, we outline how we design \vopr without using any (expensive) zero-knowledge proofs.\footnote{Previously, Ghosh \textit{et al.}  \cite{GhoshN19} designed a protocol called Oblivious Polynomial Addition (OPA) to meet similar security requirements that we laid out above. But, as shown in \cite{AbadiMZ21}, OPA  is susceptible to several serious attacks. } In the setup phase, both parties represent their input polynomials in the regular coefficient form; therefore, the sender's polynomials are defined as $\bm\psi=\sum\limits^{\st e}_{\st i=0}g_{\st i}\cdot x^{\st i}$ and  $\bm\alpha=\sum\limits^{\st e+e'}_{\st j=0}a_{\st j}\cdot x^{\st j}$ and the receiver's polynomial is defined as $\bm\beta=\sum\limits^{\st e'}_{\st i=0}b_{\st i}\cdot x^{\st i}$, where $b_{\st i}\neq 0$. However, the sender computes each coefficient $a_{\st j}$ (of polynomial $\bm \alpha$) as follows,  $a_{\st j}=\sum\limits^{\substack{\st k=e'\\ \st t=e}}_{\st t,k=0} a_{\st t,k}$,  where  $t+k=j$ and each $a_{\st t,k}$ is a random value. For instance, if $e=4$ and $e'=3$, then $a_{\st 3}=a_{\st \st 0,3}+a_{\st 3,0}+a_{\st 1,2}+a_{\st 2,1}$. Shortly, we explain why polynomial $\bm\alpha$ is constructed this way.

In the computation phase,  to compute polynomial $\bm\theta$, the two parties interactively multiply and add the related coefficients in a secure way using $\ole^{\st +}$ (presented in Section \ref{sec::OLE-plus}). Specifically,
%
%For simplicity, let $i=0$. 
%
for every $j$  (where $0\leq j\leq e'$) the sender sends $g_{\st i}$ and $a_{\st i,j}$ to an instance of  $\ole^{\st +}$, while the receiver sends $b_{\st j}$ to the same instance,  which returns $c_{\st i,j}=g_{\st i}\cdot b_{\st j}+ a_{\st i,j}$ to the receiver. This process is repeated for every $i$, where $0 \leq i \leq e$. Then, the receiver uses $c_{\st i,j}$ values to construct the resulting polynomial, $\bm\theta$.

The reason that the sender imposes the above structure to (the coefficients of)  $\bm\alpha$ in the setup, is to let the parties securely compute $\bm\theta$ via  $\ole^{\st +}$. Specifically, by imposing this structure (1) the sender  can blind each product $g_{\st i}\cdot b_{\st j}$  with  random value $a_{\st i,j}$ which is a component of $\bm\alpha$'s coefficient and (2) the receiver can construct a result polynomial of the form $\bm\theta=\bm\psi\cdot \bm\beta+\bm\alpha$.

Now, we outline how the verification works. To check the result's correctness, the sender picks and sends a random value $z$ to the receiver which computes  $\bm\theta(z)$ and $\bm\beta(z)$ and sends these two values  to the sender. The sender computes  $\bm\psi(z)$ and $\bm\alpha(z)$ and then checks if equation  $\bm\theta({ z})=\bm\psi({ z})\cdot \bm\beta({ z})+\bm\alpha({ z})$ holds. It accepts the result if the check passes.   

Figure \ref{fig:VOPR} describes \vopr in detail. Note, \vopr requires that the sender to insert non-zero coefficients, i.e., $b_{\st i}\neq 0$ for all $i,0 \leq i \leq e'$. If a   sender inserts a zero-coefficient, then it will learn only a random value (due to  $\ole^{\st +}$), accordingly it cannot pass \vopr's verification phase. However, such a requirement will not affect Justitia's correctness, as we will discuss in Section \ref{Fair-PSI-Protocol} and Appendix \ref{sec::error-prob}.

% !TEX root =main.tex

%%%%%%%%
\begin{figure}[ht]%[!htbp]
\setlength{\fboxsep}{1pt}
\begin{center}
    \begin{tcolorbox}[enhanced,width=5.5in, 
    drop fuzzy shadow southwest,
    colframe=black,colback=white]
%%%%%%%%

\begin{enumerate}
\small{
\item[$\bullet$] \textit{Input.}
\begin{enumerate}
\item[$\bullet$]  \textit{Public Parameters}: upper bound on input polynomials' degree: $e$ and $e'$. %, where $e\geq e'$.
%\item[$\bullet$]  \textit{Sender}: picks an upper bound on input polynomials's degree:  $d$ and $d'$. It sends them to the receiver.
\item[$\bullet$]  \textit{Sender Input}:  random polynomials: $\bm\psi=\sum\limits^{\st e}_{\st i=0}g_{\st i}\cdot x^{\st i}$ and  $\bm\alpha=\sum\limits^{\st e+e'}_{\st j=0}a_{\st j}\cdot x^{\st j}$, where $g_{\st i}\stackrel{\st \$}\leftarrow \mathbb{F}_p$.  Each $a_{\st j}$ has the  form: $a_{\st j}=\sum\limits^{\substack{\st k=e'\\ \st t=e}}_{\st t,k=0} a_{\st t,k}$,  such that $t+k=j$ and $a_{\st t,k}\stackrel{\st \$}\leftarrow \mathbb{F}_p$.

\item[$\bullet$] \textit{Receiver Input}:  polynomial $\bm\beta=\bm\beta_{\st 1}\cdot \bm\beta_{\st 2}=\sum\limits^{\st e'}_{\st i=0}b_{\st i}\cdot x^{\st i}$, where $\bm\beta_{\st 1}$ is a random polynomial of degree $1$, $\bm\beta_{\st 2}$ is an arbitrary polynomial of degree $e'-1$, and $b_{\st i}\neq 0$.

\end{enumerate}
\item[$\bullet$] \textit{Output.} The receiver gets $\bm\theta=\bm\psi\cdot \bm\beta+\bm\alpha$.
\item \textbf{Computation:}

\begin{enumerate} 

\item Sender and receiver together for every $j$, $0\leq j\leq e'$,  invoke $e+1$ instances of $\ole^{\st +}$. In particular, $\forall j, 0\leq j\leq e'$: sender sends $g_{\st i}$ and $a_{\st i,j}$ while the receiver sends $b_{\st j}$ to $\ole^{\st +}$ that returns: $c_{\st i,j}=g_{\st i}\cdot b_{\st j}+ a_{\st i,j}$ to the receiver ($\forall i, 0\leq i\leq e$).

 \item The receiver sums component-wise values $c_{\st i,j}$  that results in polynomial:
 \vspace{-2mm}
  $$\bm\theta=\bm\psi\cdot \bm\beta+\bm\alpha=\sum\limits^{\substack{\st i=e\\ \st j=e'}}_{\st i,j=0}c_{\st i, j}\cdot x^{\st i+j}$$ 
  \vspace{-2mm}
  %

% \item The receiver sums component-wise values $c_{\st i,j}$  that results polynomial $\bm\theta=\bm\psi\cdot \bm\beta+\bm\alpha=\sum\limits^{\st e+e'}_{\st j=0}c_{\st j}\cdot x^{\st j}$, where  each $c_{\st j}$ has   form: $c_{\st j}=\sum\limits^{\substack{\st k=e'\\ \st t=e}}_{\st  t,k=0} c_{\st t,k}$, such that $ t+k=j$.
%\item Sender: $\forall j, 1\leq j\leq 2d+1$, computes $a_{\st j}=a(x_{\st j})$ and $r_{\st j}=r(x_{\st j})$. Then, it  inserts $(a_{\st j}, r_{\st j})$ into  $\mathcal{F}_{\st OLE^{\st +}}$
%\item\label{computing-receiver} Receiver:  $\forall j, 1 \leq j\leq 2d+1$, computes $b_{\st j}=b(x_{\st j})$. Then, it  inserts $b_{\st j}$ into  $\mathcal{F}_{\st OLE^{\st +}}$ and receives $s_{\st j}=a_{\st j}+b_{\st j}\cdot r_{\st j}$. It interpolates a polynomial $s(x)$ using pairs $s_{\st j},x_{\st j}$. 
\end{enumerate}
\item \label{Verification} \textbf{Verification:}
\begin{enumerate}

\item \label{picking-random-x}Sender: picks a random (non-zero) value  $z$ and sends it to the receiver.

\item\label{receiver-OLE-invocation} Receiver: sends $\theta_{\st z}=\bm\theta(z)$ and $\beta_{\st z}=\bm\beta(z)$ to the sender.

\item\label{receiver-OLE-invocation} Sender:  computes $\psi_{\st z}=\bm\psi(z)$ and $\alpha_{\st z}=\bm\alpha(z)$ and checks   if equation  $\theta_{\st z}=\psi{\st z}\cdot \beta_{\st z}+\alpha_{\st z}$ holds. If the equation holds, it concludes that the computation was performed correctly. Otherwise, it aborts. 
\end{enumerate}
}
 \end{enumerate}
 \end{tcolorbox}
\end{center}
\caption{Verifiable Oblivious Polynomial Randomization ({\vopr}) } 
\label{fig:VOPR}
\end{figure}
 %%%%%%%

\begin{theorem}\label{theorem::VOPR}
Let $f^{\st \vopr}$ be the functionality defined above. If the enhanced \ole (i.e., $\ole^{\st +}$) is secure against malicious (or active) adversaries, then the  Verifiable Oblivious Polynomial Randomisation (\vopr), presented in Figure \ref{fig:VOPR}, securely computes $f^{\st \vopr}$ in the presence of (i) a semi-honest sender and honest receiver or (ii) a malicious receiver and honest sender. 
\end{theorem}

% !TEX root =main.tex

\begin{proof}

Before proving Theorem \ref{theorem::VOPR}, we present Lemma \ref{theorem::evaluation-of-random-poly} and Theorem \ref{theorem:coef-poly-prod} that will be used in the proof of Theorem \ref{theorem::VOPR}. 
% !TEX root =main.tex
Informally, Lemma \ref{theorem::evaluation-of-random-poly} states that the evaluation of a random polynomial at a fixed value results in a uniformly random value. %It will be used in {\vopr}'s proof to show that during the verification (in {\vopr}) a malicious party cannot learn anything about its counter party's input.
%\noindent\textbf{Remark 3.} In terms of assumption on the collusion between dealers and assistant clients,  parties, the main difference between using OPE in  and  is that if we will have use \cite{} then even if all but one dealers collude with each other, with all but one assistant clients,  with the rest of authorizer clients and client $B$, they cannot learn honest parties' set elements. However, we use the OPE in \cite{} the assumption is that no dealers collude with any assistant clients. 

%\begin{lemma}\label{theorem::evaluation-of-random-poly}
%Let $(x_{\st i}, y_{\st i})$ be arbitrary elements of a finite field $\mathbb{F}_{\st p}$, where $p$ is a security parameter and sufficiently large prime number.  The probability that the evaluation of a random polynomial $\bm\mu(x)$ at $x_{\st i}$ equals $ y_{\st i}$ is negligible in the security parameter. More formally, for $x_{\st i}, y_{\st i} \in \mathbb{F}_{\st p}$ and $\bm\mu(x)\stackrel{\st\$}\leftarrow \mathbb{F}_{\st p}[X]$, $1\leq deg(\mu)\leq d$, we have: $Pr[\bm\mu(x_{\st i})=y_{\st i}]=\epsilon(p)$. 
%\end{lemma}

\begin{lemma}\label{theorem::evaluation-of-random-poly}
Let $x_{\st i}$ be an element of a finite field $\mathbb{F}_{\st p}$, picked uniformly at random and $\bm\mu(x)$ be a random polynomial of constant degree $d$ (where $d=const(p)$) and defined over $\mathbb{F}_{\st p}[X]$. 
 Then, the evaluation of $\bm\mu(x)$ at $x_{\st i}$ is distributed uniformly at random over the elements of the  field, i.e., $Pr[\bm\mu(x_{\st i})=y]=\frac{1}{p}$, where $y\in \mathbb{F}_{\st p}$.

%equals $ y_{\st i}$ is negligible in the security parameter. More formally, for $x_{\st i}, y_{\st i} \in \mathbb{F}_{\st p}$ and $\bm\mu(x)\stackrel{\st\$}\leftarrow \mathbb{F}_{\st p}[X]$, $1\leq deg(\mu)\leq d$, we have: $Pr[\bm\mu(x_{\st i})=y_{\st i}]=\epsilon(p)$. 
\end{lemma}

\begin{proof} Let $\bm\mu(x)=a_{\st 0}+\sum\limits^{\st d}_{\st j=1} a_{\st j}x^{\st j}$, where the  coefficients  are distributed uniformly at random over the field.

Then, for any choice of $x$ and random coefficients   $a_{\st 1},...,a_{\st d}$, it holds that:

$$Pr[\bm\mu(x)=y]=Pr[\sum\limits_{\st i=0}^{\st d} a_{\st j}\cdot x^{\st j} = y] = Pr[a_{\st 0} = y-\sum\limits_{\st j=1}^{\st d} a_{\st j}\cdot x^{\st j}] = \frac{1}{p}$$

  $\forall y\in \mathbb{F}_{\st p}$, as $a_{\st 0}$ has been picked uniformly at random from $\mathbb{F}_{\st p}$. 
\hfill\(\Box\)
\end{proof}

Informally, Theorem \ref{theorem:coef-poly-prod} states that the product of two arbitrary polynomials (in coefficient form) is a polynomial whose roots are the union of the two original polynomials.  Below, we formally state it. The theorem has been taken from \cite{AbadiMZ21}.

\begin{theorem}\label{theorem:coef-poly-prod}
Let $\mathbf{p}$ and   $\mathbf{q}$ be two arbitrary non-constant polynomials of degree $d$ and $d'$ respectively, such that  $\mathbf{p} , \mathbf{q}   \in \mathbb{F}_{\st p}[X]$ and they are in coefficient form. Then, the product of the two polynomials is a polynomial whose roots include precisely the two polynomials' roots. 
\end{theorem}

%\begin{proof}
%Let $P=\{p_{\st 1},...,p_{\st t}\}$ and $Q=\{q_{\st 1},...,q_{\st t'}\}$ be the roots of polynomials $\mathbf{p}$ and   $\mathbf{q}$  respectively.  By the Polynomial Remainder Theorem,  polynomials $\mathbf{p}$ and $\mathbf{q}$  can be written as $\mathbf{p}(x)=\mathbf{g}(x)\cdot\prod\limits_{\st i=1}^{\st t}(x-p_{\st i})$ and $\mathbf{q}(x)=\mathbf{g}'(x)\cdot\prod\limits_{\st i=1}^{\st t'}(x-q_{\st i})$ respectively, where $\mathbf{g}(x)$ has degree $d-t$ and $\mathbf{g}'(x)$ has degree $d'-t'$. Let the product of the two polynomials be $\mathbf{r}(x)=\mathbf{p}(x)\cdot \mathbf{q}(x)$. For every $p_{\st i}\in P$, it holds  that $\mathbf{r}(p_{\st i})=0$. Because (a) there exists no non-constant polynomial in $\mathbb{F}[X]$ that has a multiplicative inverse (so it could cancel out factor $(x-p_{\st i})$ of $\mathbf{p}(x)$) and (b) $p_{\st i}$ is a root of $\mathbf{p}(x)$. The same argument  can be used to show for every $q_{\st i}\in Q$, it holds  that $\mathbf{r}(q_{\st i})=0$. Thus, $\mathbf{r}(x)$ preserves  roots of  both  $\mathbf{p}$ and $\mathbf{q}$. Moreover, $\mathbf{r}$ does not have any other roots (than $P$ and $Q$). In particular, if $\mathbf{r}(\alpha)=0$, then $\mathbf{p}(\alpha)\cdot \mathbf{q}(\alpha)=0$. Since there is no non-trivial divisors of zero in $\mathbb{F}[X]$  (as it is an integral domain), it must hold that either $\mathbf{p}(\alpha)=0$ or $\mathbf{q}(\alpha)=0$. Hence, $\alpha\in P$ or $\alpha\in Q$.  %\hfill\(\Box\)
%\end{proof}\hfill\(\Box\)

We refer readers to Appendix \ref{sec::proof-of-poly-union} for the proof of Theorem \ref{theorem:coef-poly-prod}. Next, we  prove the main theorem, i.e., Theorem \ref{theorem::VOPR}, by considering the case where each party is corrupt, in turn.

% In each case, we invoke the simulator with the corresponding party’s input and output. 

\

\noindent\textbf{Case 1: Corrupt sender.} In the real execution, the sender's view is defined as follows:

$$ \mathsf{View}_{\st S}^{\st \vopr} \Big((\bm\psi, \bm{\alpha}), \bm\beta\Big)=\{\bm\psi, \bm{\alpha}, r_{\st S},  \bm\beta(z), \bm\theta(z), \mathsf{View}^{\st \ole^{\st +}}_{\st S}, \bot \}$$
where $r_{\st S}$ is the outcome of internal random coins of the sender and $\mathsf{View}^{\st \ole^{\st +}}_{\st S}$ refers to the sender's real-model view during the execution of  $\ole^{\st +}$. The simulator $\mathsf{Sim}^{\st \vopr}_{\st S}$, which receives $\bm\psi$ and $\bm \alpha$, works as follows. 
\begin{enumerate}
\item generates an empty view. It appends to the view polynomials $(\bm\psi$, $\bm{\alpha})$ and coins $r'_{\st S}$ chosen uniformly at random. 
\item computes polynomial $\bm\beta=\bm\beta_{\st 1} \cdot \bm\beta_{\st 2}$, where $\bm\beta_{\st 1}$ is a random polynomial of degree $1$ and $\bm\beta_{\st 2}$ is an arbitrary polynomial of degree $e'-1$. Next,  it constructs polynomial $\bm\theta$ as follows: $\bm\theta=\bm\psi\cdot \bm\beta+\bm \alpha$.

\item picks value $z\stackrel{\st\$}\leftarrow \mathbb{F}_{p}$. Then, it evaluates polynomials $\bm\beta$ and $\bm\theta$  at point $z$. This results in values $\bm\beta_{\st z}$ and $\bm\theta_{\st z}$ respectively. It appends these two values to the view. 
\item extracts the sender-side simulation of $\ole^{\st +}$ from  $\ole^{\st +}$'s simulator. Let $\mathsf{Sim}^{\st \ole^{\st +}}_{\st S}$ be this simulation. Note, the latter simulation is guaranteed to exist, as $\ole^{\st +}$ has been proven secure (in \cite{GhoshN19}). It appends $\mathsf{Sim}^{\st \ole^{\st +}}_{\st S}$ and $\bot$ to its view. 
\end{enumerate}

Now, we are ready to show that the two views are computationally indistinguishable. The sender's inputs are identical in both models, so they have identical distributions. Since the real-model semi-honest adversary samples its randomness according to the protocol's description, the random coins in both models have identical distributions.  Next, we explain why values $\bm\beta(z)$ in the real model and $\bm\beta_{\st z}$ in the ideal model are (computationally) indistinguishable. In the real model, $\bm\beta(z)$ is the evaluation of polynomial $\bm\beta=\bm\beta_{\st 1}\cdot \bm\beta_{\st 2}$ at random point $z$, where $\bm\beta_{\st 1}$ is a random polynomial. We know that $\bm\beta(z)=\bm\beta_{\st 1}(z)\cdot \bm\beta_{\st 2}(z)$, for any (non-zero) $z$.  Moreover, by Lemma \ref{theorem::evaluation-of-random-poly}, we know that $\bm\beta_{\st 1}(z)$ is a uniformly random value. Therefore, $\bm\beta(z)=\bm\beta_{\st 1}(z)\cdot \bm\beta_{\st 2}(z)$ is a uniformly random value as well. In the ideal world, polynomial $\bm\beta$ has the same structure as $\bm\beta$ has (i.e., $\bm\beta=\bm\beta_{\st 1}\cdot \bm\beta_{\st 2}$, where $\bm\beta_{\st 1}$ is a random polynomial). That means $\bm\beta_{\st z}$ is a uniformly random value too. Thus,  $\bm\beta(z)$ and $\bm\beta_{\st z}$ are computationally indistinguishable. Next, we turn our attention to values $\bm\theta(z)$ in the real model and $\bm\theta_{\st z}$ in the ideal model. We know that $\bm\theta(z)$ is a function of $\bm\beta_{\st 1}(z)$, as polynomial $\bm\theta$ has been defined as $\bm\theta=\bm\psi\cdot \bm(\bm\beta_{\st 1}\cdot \bm\beta_{\st 2})+\bm\alpha$. Similarly, $\bm\theta_{\st z}$ is a function of  $\bm\beta_{\st z}$. As we have already discussed,  $\bm\beta(z)$ and $\bm\beta_{\st z}$ are computationally indistinguishable, so are their functions $\bm\theta(z)$ and $\bm\theta_{\st z}$. Moreover, as  $\ole^{\st +}$ has been proven secure,  $\mathsf{View}^{\st \ole^{\st +}}_{\st S}$ and  $\mathsf{Sim}^{\st \ole^{\st +}}_{\st S}$ are computationally indistinguishable. It is also clear that $\bot$ is identical in both models. We conclude that the two views are computationally indistinguishable.

\

\noindent\textbf{Case 2: Corrupt receiver.}  Let $\mathsf{Sim}^{\st\vopr}_{\st R}$ be the simulator, in this case, which uses a subroutine adversary, $\mathcal{A}_{\st R}$. $\mathsf{Sim}^{\st \vopr}_{\st R}$ works as follows. 
\begin{enumerate}
\item simulates ${\ole^{\st +}}$ and receives $\mathcal{A}_{\st R}$'s input coefficients $b_{\st j}$ for all $j$, $0\leq j \leq e'$, as we are in $f_{\st \ole^{\st +}}$-hybrid model.
\item reconstructs polynomial $\bm \beta$, given the above coefficients. 
\item simulates the honest sender's inputs as follows. 
%%%%%%%%%
It picks two random polynomials: ${\bm\psi}=\sum\limits^{\st e}_{\st i=0}{g}_{\st i}\cdot x^{\st i}$ and  ${\bm\alpha}=\sum\limits^{\st e+e'}_{\st j=0}{a}_{\st j}\cdot x^{\st j}$, such that ${g}_{\st i}\stackrel{\st \$}\leftarrow \mathbb{F}_{\st p}$ and  every $a_{\st j}$ has the  form: $a_{\st j}=\sum\limits^{\substack{\st k=e'\\ \st t=e}}_{\st t,k=0} a_{\st t,k}$,  where $t+k=j$ and $a_{\st t,k}\stackrel{\st \$}\leftarrow \mathbb{F}_{\st p}$. 
%%%%%%%
\item sends to ${\ole^{\st +}}$'s functionality values $g_{\st i}$ and $a_{\st i,j}$ and receives   $c_{\st i,j}$ from this functionality (for all $i,j$).
\item sends all ${c}_{\st i,j}$ to TTP and receives polynomial ${\bm\theta}$. 
\item picks a random value $ z$ from $\mathbb{F}_{\st p}$. Then, it computes $ \psi_{\st  z} = {\bm\psi}(  z)$ and $\alpha_{\st  z}= {\bm\alpha}(z)$. 
\item sends $ z$ and all ${c}_{\st i,j}$ to $\mathcal{A}_{\st R}$ which sends back $\theta_{ z}$ and $\beta_{\st z}$ to the simulator. 
\item sends  $ \psi_{\st  z}$ and $\alpha_{\st z}$ to $\mathcal{A}_{\st R}$. 
\item checks if the following relation hold:
\begin{equation}\label{equ::beta}
 \beta_{\st \bar z}={\bm\beta}( z) \hspace{6mm} \wedge \hspace{6mm} \theta_{\st  z}={\bm\theta}( z) \hspace{6mm}\wedge\hspace{6mm} {\bm\theta}( z)=  \psi_{\st  z}\cdot \beta_{\st  z}+\alpha_{\st  z}
 \end{equation}
 %
% \begin{equation}\label{equ::theta}
% \bar\theta_{\st \bar z}=\bar{\bm\theta}(\bar z) 
% \end{equation}
% %
%\begin{equation}\label{equ::theta-equals=psi}
%\bar{\bm\theta}(\bar z)= \bar \psi_{\st \bar z}\cdot \bar\beta_{\st \bar z}+\bar\alpha_{\st  \bar z}
%  \end{equation}
%%
% \begin{equation}\label{equ::theta-evl-equals=psi}
% \bar\theta_{\st \bar z}= \bar \psi_{\st \bar z}\cdot \bar\beta_{\st \bar z}+\bar\alpha_{\st  \bar z}
 %\end{equation}
%

%%%%%%%%%%%%%%%

%
% $\bar{\bm\theta}(\bar z) =\bar\theta_{\st \bar z} = \bar \psi_{\st \bar z}\cdot \bar\beta_{\st \bar z}+\bar\alpha_{\st  \bar z}$. 
 %
 If Relation \ref{equ::beta} does not hold, it aborts (i.e., sends abort signal $\Lambda$ to the sender) and still proceeds to the next step. 
 
 \item outputs whatever $\mathcal{A}_{\st R}$ outputs. 
\end{enumerate}

We first focus on the adversary's output. Both values of $z$ in the real and ideal models have been picked uniformly at random from $\mathbb{F}_{\st p}$; therefore, they have identical distributions. In the real model, values  $\psi_{\st z}$ and $\alpha_{\st z}$ are the result of the evaluations of two random polynomials at (random) point $z$. In the ideal model, values $\psi_{\st  z}$ and $\alpha_{\st z}$ are also the result of the evaluations of two random polynomials (i.e., ${\bm\psi}$ and ${\bm\alpha}$) at point $ z$.  By Lemma \ref{theorem::evaluation-of-random-poly}, we know that the evaluation of a random polynomial at an arbitrary value  
 yields a uniformly random value in $\mathbb{F}_{\st p}$. Therefore, the distribution of pair $(\psi_{\st z}, \alpha_{\st z})$ in the real model is identical to that of pair $( \psi_{\st  z}, \alpha_{\st  z})$ in the ideal model. Moreover, the final result (i.e., values ${c}_{\st i,j}$) in the real model has the same distribution as the final result (i.e., values ${c}_{\st i,j}$)  in the ideal model, as they are the outputs of the ideal calls to $f_{\st \ole^{\st +}}$, as we are in the $f_{\st \ole^{\st +}}$-hybrid model. 

Next, we turn our attention to the sender's output. We will show that the output distributions of the honest sender in the ideal and real models are statistically close. 
%
%Note that the messages distributions the sender receive from the ideal call to $f_{\st \text{OLE}^{\st +}}$ are identical in both models. 
%
Our focus will be on the probability that it aborts in each model, as it does not receive any other output. In the ideal model, $\mathsf{Sim}^{\st \vopr}_{\st R}$ is given the honestly generated result polynomial ${\bm \theta}$ (computed by TTP) and the adversary's input polynomial ${\bm \beta}$. $\mathsf{Sim}^{\st \vopr}_{\st R}$ aborts with a probability of 1 if Relation \ref{equ::beta} does not hold. However, in the real model, the honest sender (in addition to its inputs) is given only $\beta_{\st  z}$ and $\theta_{\st  z}$ and is not given polynomials ${\bm\beta}$ and  ${\bm\theta}$; it wants to check if the following equation holds, $\theta_{\st z} =  \psi_{\st z}\cdot  \beta_{\st z}+ \alpha_{\st z}$. Note, polynomial $\bm\theta=\bm\psi\cdot \bm\beta+\bm\alpha$ (resulted from ${c}_{\st i,j}$) is well-structured, as it satisfies the following three conditions, regardless of the adversary's input $\bm\beta$ to $\ole^{\st +}$, (i) $deg(\bm\theta)=Max \Big(deg(\bm\beta)+deg(\bm\psi), deg(\bm\alpha) \Big)$, as $\mathbb{F}_{\st p}[X]$ is an integral domain and ($\bm\psi,\bm\alpha$) are random polynomials, (ii)  the roots of the product polynomial $\bm\nu=\bm\psi\cdot \bm\beta$ contains exactly both polynomials' roots, by Theorem \ref{theorem:coef-poly-prod}, and (iii)  the roots of $\bm\nu+\bm\alpha$ is the intersection of the roots of $\bm\nu$ and $\bm\alpha$, as shown in \cite{DBLP:conf/crypto/KissnerS05}. Furthermore, polynomial $\bm \theta$ reveals no information (about  $\bm \psi$ and $\bm\alpha$ except their degrees) to the adversary and the pair $(\psi_{\st z}, \alpha_{\st z})$ is given to the adversary after it sends the pair $(\theta_{\st z}, \beta_{\st z})$ to the sender. 
There are exactly four cases where pair $(\theta_{\st z}, \beta_{\st z})$ can be constructed by the real-model adversary. Below, we state each case and the probability that the adversary is detected in that case during the verification, i.e., $\theta_{\st z}\stackrel{\st ?}=\psi{\st z}\cdot \beta_{\st z}+\alpha_{\st z}$. 
\begin{enumerate}
\item $\theta_{\st z}= \bm\theta(z) \wedge  \beta_{\st  z}={\bm\beta}( z)$. This is a trivial non-interesting case, as the adversary has behaved honestly, so it can always pass the verification. %, i.e., $\theta_{\st z}\stackrel{\st ?}=\psi{\st z}\cdot \beta_{\st z}+\alpha_{\st z}$.
\item $\theta_{\st z}\neq \bm\theta(z) \wedge  \beta_{\st  z}={\bm\beta}( z)$. In this case, the adversary is detected with a probability of 1. 
\item $\theta_{\st z}= \bm\theta(z) \wedge  \beta_{\st  z}\neq {\bm\beta}( z)$.  In this case, the adversary is also detected with a probability of 1.
\item $\theta_{\st z}\neq \bm\theta(z) \wedge  \beta_{\st  z} \neq {\bm\beta}( z)$. In this case, the adversary is detected with an overwhelming probability, i.e., $1-\frac{1}{2^{\st 2\lambda}}$. 
\end{enumerate}

As we illustrated above, in the real model, the lowest probability that the honest sender would abort in case of adversarial behaviour is $1-\frac{1}{2^{\st 2\lambda}}$. Thus, the honest sender's output distributions in the ideal and real models are statistically close, i.e., $1$ \text{vs} $1-\frac{1}{2^{\st 2\lambda}}$. 

We conclude that the distribution of the joint outputs of the honest sender and adversary in the real and ideal models are computationally indistinguishable. 
  \hfill\(\Box\)\end{proof}

% !TEX root =main.tex

\subsection{Zero-sum Pseudorandom Values Agreement Protocol (\zspa)}

The \zspa  allows $m$ parties (the majority of which is potentially malicious) to efficiently agree on (a set of vectors, where each $i$-th vector has) $m$ pseudorandom values such that their sum equals zero. At a high level, the parties first sign a smart contract, register their accounts/addresses in it, and then run a  coin-tossing protocol \ct to agree on a key: $k$.  Next, one of the parties generates $m-1$ pseudorandom values $z_{\scriptscriptstyle i, j}$ (where $1\leq j\leq m-1$) using key $k$ and $\mathtt{PRF}$. It sets the last value as the additive inverse of the sum of the values generated, i.e. $z_{\scriptscriptstyle i, m}=-\sum\limits^{\scriptscriptstyle m-1}_{\scriptscriptstyle j=1}z_{\scriptscriptstyle i, j}$ (similar to the standard XOR-based secret sharing \cite{Schneier0078909}). 
%
%Next, it commits to each value, where it uses $k_{\scriptscriptstyle 2}$ to generate the randomness of each commitment. 
%
Then, it constructs a Merkel tree on top of the pseudorandom values and stores only the tree's root $g$ and the key's hash value $q$ in the smart contract.  Then, each party (using the key) locally checks if the values (on the contract) have been constructed correctly; if so, then it sends a (signed) ``approved" message to the contract which only accepts messages from registered parties. Hence, the functionality that \zspa computes is defined as $f^{\st \zspa}\underbrace{(\bot,..., \bot)}_{\st m}\rightarrow \underbrace{((k, g, q),..., (k, g,q))}_{\st m}$, where $g$ is the Markle tree's root built on the pseudorandom values $z_{\st i, j}$, $q$ is the hash value of the key used to generate the pseudorandom values, and $m\geq 2$. Figure \ref{fig:ZSPA} presents \zspa in detail.

Briefly, \zspa will be used in \withFai to allow clients $\{A_{\st 1},...,A_{\st m}\}$ to provably agree on a set of pseudorandom values, where each set represents a pseudorandom polynomial (as the elements of the set are considered the polynomial's coefficients). Due to \zspa's property, the sum of these polynomials is zero.  Each of these polynomials will be used by a client to blind/encrypt the messages it sends to the smart contract, to protect the privacy of the plaintext message (from \aud, D, and the public). To compute the sum of the plaintext messages, one can easily sum the blinded messages, which removes the blinding polynomials. 

% !TEX root =main.tex

\begin{figure}[ht]%[!htbp]
\setlength{\fboxsep}{1pt}
\begin{center}
    \begin{tcolorbox}[enhanced,width=5.5in, 
    drop fuzzy shadow southwest,
    colframe=black,colback=white]

\small{

\begin{enumerate}
\item[$\bullet$]    {Parties.} A set of clients $\{    A_{\st 1},...,  A_{\st m}\}$.
\item[$\bullet$]    {Input.}  $m$: the total number of participants, $adr$: an address of a deployed smart contract which maintains the clients' addresses and accepts messages only from them, and $b$: the total number of vectors. Let $b'=b-1$. 
\item[$\bullet$]   {Output.}  $k$: a secret key that generates $b$ vectors $[z_{\scriptscriptstyle 0,1},...,z_{\scriptscriptstyle 0,m}],...,[z_{\scriptscriptstyle b',1},...,z_{\scriptscriptstyle b',m}]$ of pseudorandom values, $q$: hash of the key,  $g$: a Merkle tree's root, and a vector of (signed) messages.

%, such that the sum of each vector's elements equals zero: $\sum\limits^{\scriptscriptstyle m}_{\scriptscriptstyle j=1}z_{\scriptscriptstyle i,j}=0$. 

\item {\textbf{Coin-tossing.} $\ct (in_{\st 1},..., in_{\st m})\rightarrow k$}. 

All participants run a coin-tossing protocol to agree on $\mathtt{PRF}$'s key, $k$.
\item\label{ZSPA:val-gen}  {\textbf{Encoding.} $\mathtt{Encode}(k, m)\rightarrow (g,q)$}.

 One of the parties takes the following steps:  
\begin{enumerate}

\item for every $i$ (where $0\leq i \leq b'$), generates $m$ pseudorandom values as follows. 
 $$\forall j, 1\leq j \leq m-1: z_{\scriptscriptstyle i,j}=\mathtt{PRF}(k,i||j), \hspace{5mm} z_{\scriptscriptstyle i,m}=-\sum\limits^{\scriptscriptstyle m-1}_{\scriptscriptstyle j=1}z_{\scriptscriptstyle i,j}$$
\item   constructs a Merkel tree on top of all pseudorandom values,  $\mkgen(z_{\scriptscriptstyle 0,1},...,z_{\scriptscriptstyle b',m})\rightarrow g$. 

\item sends the Merkel tree's root: $g$,   and the key's hash: $q=\mathtt {H}(k)$ to $adr$. 

\end{enumerate}

\item\label{ZSPA:verify}{\textbf{Verification.} $\mathtt{Verify}(k, g, q, m, flag)\rightarrow (a, s)$}. 
Each party $A_{\st t}$ checks if, all $z_{\scriptscriptstyle i,j}$ values, the root $g$, and key's hash $q$ have been correctly generated, by retaking  step \ref{ZSPA:val-gen}. 
\begin{itemize}
\item If the checks pass, then it sets $a=1$ and
\begin{itemize}
\item if $flag=A$, where $A\in \{    A_{\st 1},...,  A_{\st m}\}$, then  it sets $s$ to a (signed) ``approved'' message, and sends $s$ to $adr$
\item  if $flag= A$, where $A\notin \{    A_{\st 1},...,  A_{\st m}\}$, then it sets $s = \bot$. 
\end{itemize}
%it sets $a=1$,  sets $s$ to a signed ``approved'' message, and sends $s$ to $adr$. 
\item If the checks do not pass, it aborts by returning $a=0$ and $s = \bot$. 
\end{itemize}
%\item\label{ZSPA:verify}\textbf{Sending-approval.} $\mathtt{Approve}(a)\rightarrow s$.
%Each party  sets $s$ to a signed ``approved'' message if $a=1$, and sends $s$ to $adr$. Otherwise, it sets $s = \bot$ and aborts if $a=0$. 
%
 \end{enumerate}
}
 \end{tcolorbox}
\end{center}
\caption{Zero-sum Pseudorandom Values Agreement (\zspa). The use of $flag$ allows an external party (e.g., an auditor) to locally run the verification without having to send any messages to the smart contract.} 
\label{fig:ZSPA}
\end{figure}

\begin{theorem}\label{theorem::ZSPA-comp-correctness}
Let $f^{\st \zspa}$ be the functionality defined above. If \ct is secure against a malicious adversary and the correctness of $\mathtt{PRF}$, $\mathtt{H}$, and Merkle tree holds, then \zspa,  in Figure \ref{fig:ZSPA}, securely computes $f^{\st \zspa}$ in the presence of $m-1 $ malicious  adversaries. 
\end{theorem}

\begin{proof}
For the sake of simplicity, we will assume the sender, which generates the result, sends the result directly to the rest of the parties, i.e., receivers, instead of sending it to a smart contract. We first consider the case in which the sender is corrupt. 

\

\noindent\textbf{Case 1: Corrupt sender.}  Let $\mathsf{Sim}^{\st \zspa}_{\st S}$ be the simulator using a subroutine adversary, $\mathcal{A}_{\st S}$. $\mathsf{Sim}^{\st \zspa}_{\st S}$ works as follows. 
\begin{enumerate}
\item simulates  \ct  and receives the output value $k$ from $f_{\st \ct}$, as we are in $f_{\st \ct}$-hybrid model.
\item sends $k$ to TTP and receives back from it $m$ pairs, where each pair is of the form $( g,  q)$. 
\item sends $ k$ to $\mathcal{A}_{\st S}$ and receives back from it $m$ pairs  where each pair is of the form $( g',  q')$. 
\item checks whether the following equations hold (for each pair): $ g= g' \hspace{2mm} \wedge  \hspace{2mm}  q= q'$. If the two equations do not hold, then it aborts (i.e., sends abort signal $\Lambda$ to the receiver) and proceeds to the next step.
\item outputs whatever $\mathcal{A}_{\st S}$ outputs.
 \end{enumerate}
 
 We first focus on the adversary’s output. In the real model, the only messages that the adversary receives are those messages it receives as the result of the ideal call to $f_{\st \ct}$. These messages have identical distribution to the distribution of the messages in the ideal model, as the \ct is secure. Now, we move on to the receiver’s output. We will show that the output distributions of the honest receiver in the ideal and real models are computationally indistinguishable. In the real model,  each element of pair $(g, p)$ is the output of a deterministic function on the output of $f_{\st \ct}$. We know the output of $f_{\st \ct}$ in the real and ideal models have an identical distribution, and so do the evaluations of deterministic functions (i.e., Merkle tree, $\mathtt{H}$, and $\mathtt{PRF}$) on them, as long as these three functions' correctness holds. Therefore, each pair $(g,q)$ in the real model has an identical distribution to pair $(g,  q)$ in the ideal model.  For the same reasons, the honest receiver in the real model aborts with the same probability as  $\mathsf{Sim}^{\st \zspa}_{\st S}$ does in the ideal model.  We conclude that the distributions of the joint outputs of the adversary and honest receiver in the real and ideal models are  (computationally) indistinguishable. 

\

\noindent\textbf{Case 2: Corrupt receiver.}   Let $\mathsf{Sim}^{\st \zspa}_{\st R}$ be the simulator that uses subroutine adversary $\mathcal{A}_{\st R}$. $\mathsf{Sim}^{\st \zspa}_{\st R}$ works as follows. 

\begin{enumerate}
\item simulates   \ct  and receives the output value $ k$ from $f_{\st \ct}$.
\item sends $ k$ to TTP and receives back $m$ pairs of the form $( g,  q)$ from TTP. 
\item sends $( k,  g,  q)$ to $\mathcal{A}_{\st R}$ and outputs whatever  $\mathcal{A}_{\st R}$ outputs. 
 \end{enumerate}

In the real model, the adversary receives two sets of messages, the first set includes the transcripts (including $ k$) it receives when it makes an ideal call to $f_{\st \ct}$ and the second set includes pair $(g, q)$. As we already discussed above (because we are in the  $f_{\st \ct}$-hybrid model) the distributions of the messages it receives from $f_{\st \ct}$ in the real and ideal models are identical. Moreover, the distribution of $f_{\st \ct}$'s output (i.e., $\bar k$ and $k$) in both models is identical; therefore, the honest sender's output distribution in both models is identical. As we already discussed,  the evaluations of deterministic functions (i.e., Merkle tree, $\mathtt{H}$, and $\mathtt{PRF}$) on $f_{\st \ct}$'s outputs have an identical distribution. Therefore, each pair $(g, q)$ in the real model has an identical distribution to the pair $(g, q)$ in the ideal model.  Hence, the distribution of the joint outputs of the adversary and honest receiver in the real and ideal models is indistinguishable.
  \hfill\(\Box\)\end{proof}

In addition to the security guarantee (i.e., computation's correctness against malicious sender or receiver) stated by Theorem \ref{theorem::ZSPA-comp-correctness}, \zspa offers  (a) privacy against the public, and (b)  non-refutability. Informally, privacy here means that given the state of the contract (i.e., $g$ and  $q$), an external party cannot learn any information about any of the pseudorandom values,  $z_{\scriptscriptstyle j}$; while non-refutability means that if a party sends ``approved" then in future cannot deny the knowledge of the values whose representation is stored in the contract. %Furthermore, indistinguishability means that every $z_{\scriptscriptstyle j}$ ($1\leq j \leq m$) should be indistinguishable from a truly random value. 

\begin{theorem}
If  $\mathtt{H}$ is preimage resistance, $\mathtt{PRF}$ is secure, the signature scheme used in the smart contract is secure (i.e., existentially unforgeable under chosen message attacks), and the blockchain is secure (i.e., offers persistence and liveness properties \cite{GarayKL15}) then \zspa offers (i) privacy against the public and (ii) non-refutability. 
\end{theorem}

\begin{proof}
First, we focus on privacy. Since key $k$, for $\mathtt{PRF}$, has been picked uniformly at random and $\mathtt{H}$ is preimage resistance, the probability that given $g$ the adversary can find $k$ is negligible in the security parameter, i.e., $\negl(\lambda)$. Furthermore, because $\mathtt{PRF}$ is secure (i.e., its outputs are indistinguishable from random values) and  $\mathtt{H}$ is preimage resistance, given the Merkle tree's root $g$, the probability that the adversary can find a leaf node, which is the output of $\mathtt{PRF}$, is $\negl(\lambda)$ too. 

Now we move on the non-refutability. Due to the persistency property of the blockchain, once a transaction/message goes more than $v$ blocks deep into the blockchain of one honest player (where $v$ is a security parameter), then it will be included in every honest player's blockchain with overwhelming probability, and it will be assigned a permanent
position in the blockchain (so it will not be modified with an overwhelming probability). Also, due to the liveness property,   all transactions originating from honest parties will eventually end up at a depth of more than $v$ blocks in an honest player's blockchain; therefore, the adversary cannot
perform a selective denial of service attack against honest account holders.  Moreover, due to the security of the digital signature (i.e., existentially unforgeable under chosen message attacks), one cannot deny sending the messages it sent to the blockchain and smart contract. 
\hfill\(\Box\)
\end{proof}

\subsection{\zspa's Extension: \zspa with an External Auditor (\zspaa)}

In this section, we present an extension of \zspa, called \zspaa which lets a (trusted) third-party auditor, \aud, help identify misbehaving clients in the \zspa and generate a vector of random polynomials with a specific structure. Specifically, the functionality that \zspaa computes is defined as 
$f^{\st \zspaa}(\underbrace{\bot,..., \bot}_{\st m}, b, \bm\zeta)\rightarrow (\underbrace{(k, g, q),..., (k, g,q)}_{\st m},$ $\bm\mu^{\st (1)},..., \bm\mu^{\st (m)})$, where $g$ is the Markle tree's root built on the pseudorandom values $z_{\st i, j}$, $q$ is the hash value of the key used to generate the pseudorandom values, $m\geq 2$, $b$ is constant value, $\bm\zeta$ is a random polynomial of degree $1$, each $\bm\mu^{\st (j)} $ has the form $\bm\mu^{\st (j)} = \bm\zeta\cdot \bm\xi^{\st (j)}-\bm\tau^{\st (j)}$,  $\bm\xi^{\st (j)}$ is a random polynomial of degree $b-2$, and $\bm\tau^{\st (j)}=\sum\limits^{\st b-1}_{\st i=0}z_{\st i,j}\cdot x^{\st i}$.

Informally, \zspaa requires that misbehaving parties are always detected, except with a negligible probability. \aud of this protocol will be invoked by \withFai when \withFai's smart contract detects that a combination of the messages sent by the clients is not well-formed. Later, in \withFai's proof, we will show that even a \emph{semi-honest} \aud who observes all messages that clients send to \withFai's smart contracts, cannot learn anything about their set elements. We present \zspaa in Figure \ref{fig:arbiter}.

% !TEX root =main.tex

\begin{figure}[ht]%[!htbp]
\setlength{\fboxsep}{1pt}
\begin{center}
    \begin{tcolorbox}[enhanced,width=5.5in, 
    drop fuzzy shadow southwest,
    colframe=black,colback=white]

{\small{

%\underline{$\mathtt{Audit}( \vv{{k}},  q, \bm\zeta, \bar d, g, \vv v)\rightarrow (L, \vv{{\mu}})$}
\begin{enumerate}
%\item[$\bullet$] Parties: clients: $\{  {   A}_{    {    1}},...,   {   A}_{    {    m}}\}$, the dealer and  an Arbiter.

\item[$\bullet$]    {Parties.} A set of clients $\{ A_{\st 1},...,  A_{\st m}\}$ and an external auditor, \aud. 

\item[$\bullet$]    {Input.}  $m$: the total number of participants (excluding the auditor), $\bm\zeta$: a random polynomial of degree $1$, $b$: the total number of vectors, and $adr$: a deployed smart contract's address. Let $b'=b-1$.

%\item[$\bullet$]   {Input.} $\vv{{k}}=[k_{\st 1},..., k_{\st m}]$,    $q$: a  hash value, $\bm\zeta$: a random polynoimal of degree $1$, $\bar d$: a polynoimal's degree,   $g$: a root of Merkle tree, and $\vv v$: binary vector of size $m$. 

\item[$\bullet$]  {Output of  each} $  A_{\st j}$.   $k$: a secret key that generates $b$ vectors $[z_{\scriptscriptstyle 0,1},...,z_{\scriptscriptstyle 0,m}],...,[z_{\scriptscriptstyle b',1},...,z_{\scriptscriptstyle b', m}]$ of pseudorandom values, $h$: hash of the key,  $g$: a Merkle tree's root, and a vector of signed messages.

\item[$\bullet$]    {Output of \aud.} $L$: a list of misbehaving parties' indices, and  $\vv{{\mu}}$: a vector of random polynomials.
\item\label{ZSPA::ZSPA-invocation} {\textbf{\zspa invocation.}  $\zspa(\bot,..., \bot)\rightarrow \Big((k, g, q),..., (k, g,q )\Big)$}. 

All parties in $\{A_{\st 1},...,  A_{\st m}\}$ call the same instance of \zspa, which results in  $(k, g, q), ..., (k, g, q)$. 

\item\label{ZSPA-A::Auditor-computation}  {\textbf{Auditor computation.} $\mathtt{Audit}(m, \vv{{k}},  q, \bm\zeta, b, g)\rightarrow (L, \vv{{\mu}})$}. 

\aud\ takes the below steps. Note,  each $k_{\st j}\in \vv{{k}}$ is given by  $  A_{\st j}$. An honest party's input, $k_{\st j}$,  equals $k$, where $1\leq j \leq m$.

\begin{enumerate}
\item runs the checks in the verification phase (i.e., Phase \ref{ZSPA:verify}) of \zspa for every $j$, i.e., $\mathtt{Verify}(k_{\st j}, g, q, m, auditor)\rightarrow (a_{\st j}, s)$.
\item appends $j$ to $L$, if any checks fails, i.e., if $a_{\st j}=0$. In this case, it skips the next two steps for the current $j$.

%
%
%\item  Checks whether equation $\mathtt{H}(k_{\st j})=q$ holds  for every $j$, $1\leq j \leq m$.   
%%
%\begin{itemize}
%%
%\item[$\bullet$] if any $j$-th check fails,  it adds $j$ to $L$.
%%
%\item[$\bullet$]  if $L$ contains all $j\in[1,m]$, it returns $L$ and aborts. 
%%
%\end{itemize}
%%
%\item\label{zero-sum-arbiter-verification} Verifies the Merkle tree's root, $g$, by checking if the tree (corresponding to  $g$) has been correctly constructed on the correct leaf nodes. In particular, it takes the following steps. 
%
%\begin{enumerate}
%
%\item regenerates the tree's leaf nodes (similar to step \ref{ZSPA:val-gen} in Fig. \ref{fig:ZSPA}) as follows. Let $k$ be a key that passed the above check.  For every $i$ (where $0\leq i \leq \bar d$), it recomputes $m$ pseudorandom values: 
%%
%$$\forall j, 1\leq j \leq m-1: z_{\st i,j}=\mathtt{PRF}(k,i||j), \hspace{4mm} z_{\st i,m}=-\sum\limits^{\st m-1}_{\st j=1}z_{\st i,j}$$
%%
%\item   constructs a Merkel tree on top of all pseudorandom values generated in the previous step, i.e., $\mathtt{MT.genTree}(z_{\st 0,1},...,z_{\st \bar d,m})\rightarrow g'$. 
%%
%\item checks if $g=g'$. If the equation does not hold, then it adds to $L$ every index $j$ whose value in $\vv v$ is $1$, i.e., $\vv v[j]=1$; in this case, it returns $L$ and aborts.
%%
%\end{enumerate}
%

\item\label{ZSPA-A::gen-z} For every $i$ (where $0\leq i \leq b'$), it recomputes $m$ pseudorandom values: 
$\forall j, 1\leq j \leq m-1: z_{\st i,j}=\mathtt{PRF}(k,i||j), \hspace{4mm} z_{\st i,m}=-\sum\limits^{\st m-1}_{\st j=1}z_{\st i,j}$.
 \item generates polynomial $\bm\mu^{\st (j)}$ as follows: 
   $\bm\mu^{\st (j)} = \bm\zeta\cdot \bm\xi^{\st (j)}-\bm\tau^{\st (j)}$, 
    where $\bm\xi^{\st (j)}$ is a random polynomial of degree $b'-1$ and $\bm\tau^{\st (j)}=\sum\limits^{\st b'}_{\st i=0}z_{\st i,j}\cdot x^{\st i}$. By the end of this step, a vector $\vv{{\mu}}$ containing at most $m$ polynomials is generated. 
 \item returns   list $L$ and $\vv{{\mu}}$.
 
\end{enumerate}
 \end{enumerate}
}}
 \end{tcolorbox}
\end{center}
\caption{\zspa with an external auditor (\zspaa)} 
\label{fig:arbiter}
\end{figure}

\begin{theorem}\label{theorem::ZSPA-A}
If \zspa is secure, $\mathtt{H}$ is second-preimage resistant, and the correctness of $\mathtt{PRF}$, $\mathtt{H}$, and Merkle tree holds,  then \zspaa securely computes $f^{\st \zspaa}$ in the presence of $m-1 $ malicious adversaries.% or (ii) a semi-honest auditor. 
\end{theorem}

%As we stated previously, the ZSPA-A protocol will be invoked as a subroutine in the fair PSI protocol. As part of proving Theorem \ref{theorem::ZSPA-A}, we would like to show that the semi-honest auditor's view can be simulated (so it cannot learn the parties' set elements), even if it has access to those transcripts of the fair PSI protocol sent to the smart contract; because such an approach offers a stronger security guarantee than proving the ZSPA-A protocol in isolation.  Therefore, we will present the proof of Theorem \ref{theorem::ZSPA-A} after we present the fair PSI protocol. 

% !TEX root =main.tex

%\subsection{Proof of the ZSPA-A Protocol}\label{subsec::ZSPA-A-proof}

%Below, we prove Theorem \ref{theorem::ZSPA-A}, i.e., the security of the ZSPA-A.  Note, to allow this proof to be self-contained, we repeat some content from the proof of Theorem \ref{theorem::ZSPA-comp-correctness}. 

\begin{proof}
First, we consider the case where a sender, who (may collude with $m-2$ senders and) generates pairs $(g,q)$, is corrupt. 

\

\noindent\textbf{Case 1: Corrupt sender.}  Let $\mathsf{Sim}^{\st \zspaa}_{\st S}$ be the simulator using a subroutine adversary, $\mathcal{A}_{\st S}$. Below, we explain how $\mathsf{Sim}^{\st \zspaa}_{\st S}$ works.

%%%%%%%%%%%%%%%%%%%
\begin{enumerate}
\item simulates  \ct  and receives the output value $ k$ from $f_{\st \ct}$.
\item sends $ k$ to TTP and receives back from it $m$ pairs, where each pair is of the form $( g,  q)$. 
\item sends $ k$ to $\mathcal{A}_{\st S}$ and receives back from it $m$ pairs  where each pair is of the form $( g',  q')$. 

\item constructs an empty vector $ L$. $\mathsf{Sim}^{\st \zspaa}_{\st S}$ checks whether the following equations hold for each $j$-th pair: $ g= g' \hspace{2mm} \wedge  \hspace{2mm}  q= q'$. If these two equations do not hold, it sends an abort message $\Lambda$ to other receiver clients, appends the index of the pair (i.e., $j$) to $ L$, and proceeds to the next step for the valid pairs. In the case where there are no valid pairs, it moves on to step \ref{adversary-outputs}. 
\item picks a random polynomial ${\bm \zeta}$ of degree $1$. Moreover, for every $j\notin  L$, $\mathsf{Sim}^{\st \zspaa}_{\st S}$ picks a random polynomial ${\bm\xi}^{\st (j)}$ of degree $b'-1$, where $1\leq j \leq m$.

\item computes $m$ pseudorandom values  for every $i,j'$, where $0\leq i \leq b'$ and $j'\notin  L$ as follows.   
$$\forall j', 1\leq j' \leq m-1: z_{\st i,j}=\mathtt{PRF}( k,i||j')\hspace{5mm}  \text{and }\hspace{5mm} z_{\st i,m}=-\sum\limits^{\st m-1}_{\st j'=1}z_{\st i,j}$$ 
 \item generates polynomial ${\bm\mu}^{\st (j)}$, for every $j \notin L$, as follows:
   ${\bm\mu}^{\st (j)} = {\bm\zeta}\cdot {\bm\xi}^{\st (j)}-\bm\tau^{\st (j)}$, where ${\bm\tau}^{\st (j)}=\sum\limits^{\st b'}_{\st i=0}z_{\st i,j}\cdot x^{\st i}$.
\item sends the above ${\bm \zeta}$,  ${\bm\xi}^{\st (j)}$, and ${\bm\mu}^{\st (j)}$ to all parties (i.e., $\mathcal{A}_{\st S}$ and the receivers), for every $j\notin L$. 
\item\label{adversary-outputs} outputs whatever $\mathcal{A}_{\st S}$ outputs.
 \end{enumerate}
%%%%%%%%%%%%%%%%%

 Now, we focus on the adversary’s output. In the real model, the messages that the adversary receives include those messages it receives as the result of the ideal call to $f_{\st \ct}$ and (${\bm \zeta}, {\bm\xi}^{\st (j)}, \bm\mu^{\st (j)}$), where $j \notin L$ and $1\leq j\leq m$. Those messages yielded from the ideal calls have identical distribution to the distribution of the messages in the ideal model, as \ct is secure. The distribution of each $\bm\mu^{\st (j)}$ depends on the distribution of its components; namely, ${\bm \zeta}, {\bm\xi}^{\st (j)}$, and $\bm \tau^{\st j}$. As we are in the $f_{\st \ct}$-hybrid model, the distributions of $\bm \tau^{\st (j)}$ in the real model and ${\bm \tau}^{\st (j)}$ in the ideal model are identical, as they were derived from the output of $f_{\st \ct}$. Furthermore, in the real model, each polynomial  ${\bm \zeta}$ and ${\bm\xi}^{\st (j)}$ has been picked uniformly at random and they are independent of the clients' and the adversary's inputs. The same arguments hold for  ($ {\bm \zeta}, {\bm\xi}^{\st (j)}, {\bm\mu}^{\st (j)}$) in the ideal model. Therefore, (${\bm \zeta}, {\bm\xi}^{\st (j)}, {\bm\mu}^{\st (j)}$) in the real model and ($ {\bm \zeta}, {\bm\xi}^{\st (j)}, {\bm\mu}^{\st (j)}$) in the ideal model have identical distributions.

 Next, we turn our attention to the receiver’s output. We will show that the output distributions of an honest receiver and the auditor in the ideal and real models are computationally indistinguishable. In the real model,  each element of the pair $(g, p)$ is the output of a deterministic function on the output of $f_{\st \ct}$. We know the outputs of $f_{\st \ct}$ in the real and ideal models have an identical distribution, and so do the evaluations of deterministic functions (namely Merkle tree, $\mathtt{H}$, and $\mathtt{PRF}$) on them. Therefore, each pair $(g,q)$ in the real model has an identical distribution to the pair $( g, q)$ in the ideal model.  For the same reasons, the honest receiver in the real model aborts with the same probability as  $\mathsf{Sim}^{\st \zspaa}_{\st S}$ does in the ideal model.  The same argument holds for the arbiter's output, as it performs the same checks that an honest receiver does.  Thus, the distribution of the joint outputs of the adversary, honest receiver, and honest in the real and ideal models is computationally indistinguishable. 

\

\noindent\textbf{Case 2: Corrupt receiver.}   Let $\mathsf{Sim}^{\st \zspaa}_{\st R}$ be the simulator that uses subroutine adversary $\mathcal{A}_{\st R}$. Below, we explain how $\mathsf{Sim}^{\st \zspaa}_{\st R}$ works.

\begin{enumerate}
\item simulates  \zspa and receives the $m$ output pairs of the form $( k,  g,  q) $ from $f^{\st \zspa}$.
\item sends $( k,  g,  q) $ to $\mathcal{A}_{\st R}$ and receives $m$ keys, $ k'_{\st j}$, where $1\leq j \leq m$. 
\item\label{ZSPA-A-Case-2-generate-z} generates an empty vector $ L$. Next, for every $j$, $\mathsf{Sim}^{\st \zspaa}_{\st R}$ computes $q'_{\st j}$ as $\mathtt {H}( k'_{\st j})=q_{\st j}$.  It generates $g_{\st j}$ as follows.

\begin{enumerate}

\item\label{gen-pr-vals} for every $i$ (where $0\leq i \leq b'$), generates $m$ pseudorandom values as below. 
 $$\forall j, 1\leq j' \leq m-1: z_{\scriptscriptstyle i,j}=\mathtt{PRF}( k'_{\st j},i||j'), \hspace{5mm} z_{\scriptscriptstyle i,m}=-\sum\limits^{\scriptscriptstyle m-1}_{\scriptscriptstyle j=1}z_{\scriptscriptstyle i,j}$$
\item   constructs a Merkel tree on top of all pseudorandom values,  $\mathtt{MT.genTree}(z_{\scriptscriptstyle 0,1},...,z_{\scriptscriptstyle b', m})\rightarrow  g'_{\st j}$. 
\end{enumerate}
\item checks if the following equations hold for each $j$-th pair: 
$( k=    k'_{\st j}) \ \wedge\  ( g=   g'_{\st j})  \ \wedge\ ( q=  q'_{\st j})$.
%
%\item generates an empty vector $\bar L$ and then checks if the following equations hold for each $j$-th pair: 
%%
%$\bar k=   \bar k' \ \wedge\  \bar g=  \bar g'  \ \wedge\ \bar q= \bar q'$.
%
\item If these equations do not hold for $j$-th value, it appends $j$ to $ L$ and proceeds to the next step for the valid value. In the case where there is no valid value, it moves on to step \ref{adversary-outputs--}.

\item picks a random polynomial ${\bm \zeta}$ of degree $1$. Also, for every $j\notin  L$, it picks a random polynomial ${\bm\xi}^{\st (j)}$ of degree $b'-1$, where $1\leq j \leq m$. 
%

%\item computes $m$ pseudorandom values  for every $i,j'$, where $0\leq i \leq \bar d$ and $j'\notin \bar L$ as follows.   
%%
%$\forall j', 1\leq j' \leq m-1: z_{\st i,j}=\mathtt{PRF}(\bar k,i||j')\hspace{3mm}  \text{and }\hspace{3mm} z_{\st i,m}=-\sum\limits^{\st m-1}_{\st j'=1}z_{\st i,j}$.
%
 \item generates polynomial ${\bm\mu}^{\st (j)}$, for every $j \notin L$, as follows:
   ${\bm\mu}^{\st (j)} = {\bm\zeta}\cdot {\bm\xi}^{\st (j)}-\bm\tau^{\st (j)}$, where ${\bm\tau}^{\st (j)}=\sum\limits^{\st b'}_{\st i=0}z_{\st i,j}\cdot x^{\st i}$, and values $z_{\st i,j}$ were generated in step \ref{gen-pr-vals}.
\item sends the above ${\bm \zeta}$,  ${\bm\xi}^{\st (j)}$, and ${\bm\mu}^{\st (j)}$ to $\mathcal{A}_{\st R}$, for every $j\notin L$ and $1\leq j \leq m$. 
\item\label{adversary-outputs--} outputs whatever  $\mathcal{A}_{\st R}$ outputs. 
 \end{enumerate}

In the real model, the adversary receives two sets of messages, the first set includes the transcripts (including $ k,  g,  q$) it receives when it makes an ideal call to $f^{\st \zspa}$ and the second set includes pairs $({\bm \zeta},  {\bm\xi}^{\st (j)}, {\bm\mu}^{\st (j)})$, for every $j\notin  L$ and $1\leq j \leq m$. Since we are in the  $f^{\st \zspa}$-hybrid model and (based on our assumption) there is at least one honest party participated in  $\zspa$ (i.e., there are at most $m-1$ malicious participants of \zspa), the distribution of the messages it receives from $f^{\st \zspa}$ in the real and ideal models is identical. Furthermore, as we discussed in Case 1, (${\bm \zeta}, {\bm\xi}^{\st (j)}, {\bm\mu}^{\st (j)}$) in the real model and ($ {\bm \zeta}, {\bm\xi}^{\st (j)}, {\bm\mu}^{\st (j)}$) in the ideal model have identical distribution. The honest sender's output distribution in both models is identical, as the distribution of $f_{\st \ct}$'s output (i.e., $ k$) in both models is identical.

Now we show that the probability that the auditor aborts in the ideal and real models are statistically close. In the ideal model, $\mathsf{Sim}^{\st \zspaa}_{\st R}$ is given the ideal functionality's output that includes key $ k$. Therefore, it can check whether the key that $\mathcal{A}_{\st R}$ sends to it equals $ k$, i.e., $ k\stackrel{\st ?}=    k'_{\st j}$. Thus, it aborts with the probability $1$. However, in the real model, an honest auditor is not given the output of \ct (say key $k$) and it can only check whether the key is consistent with the hash value $q$ and the Merkle tree's root $g$ stored on the blockchain. This means the adversary can distinguish the two models if in the real model it sends a key $\ddot k$, such that $\ddot k\neq k$ and still passes the checks. Specifically, it sends the invalid key $\ddot k$ that can generate  valid pair $(g, q)$, as follows: $\mathtt{H}(\ddot k)=q$ and $\mkgen(z'_{\scriptscriptstyle 0,1}, ..., z'_{\scriptscriptstyle b', m})\rightarrow  g$, where each $z'_{\scriptscriptstyle i,j}$ is derived from $\ddot k$ using the same technique described in step \ref{ZSPA-A-Case-2-generate-z} above. Nevertheless, this means that the adversary breaks the second preimage resistance property of  $\mathtt{H}$; however, $\mathtt{H}$ is the second-preimage resistance and the probability that the adversary succeeds in finding the second preimage is negligible in the security parameter, i.e., $\negl(\lambda)$ where $\lambda$ is the hash function's security parameter. Therefore, in the real model, the auditor aborts if an invalid key is provided with a probability $1-\negl(\lambda)$ which is statically close to the probability that $\mathsf{Sim}^{\st \zspaa}_{\st R}$ aborts in the same situation in the ideal model, i.e., $1-\negl(\lambda)$ vs $1$. 
Hence, the distribution of the joint outputs of the adversary, honest sender, and honest auditor in the real and ideal models is indistinguishable.
  \hfill\(\Box\)\end{proof}

\subsection{Unforgeable Polynomials}

In this section, we introduce the notion of ``unforgeable polynomials''. Informally, an unforgeable polynomial has a secret factor. To ensure that an unforgeable polynomial has not been tampered with, a verifier can check whether the polynomial is divisible by the secret factor.

To turn an arbitrary polynomial $\bm\pi$ of degree $d$ into an unforgeable polynomial $\bm\theta$, one can (i) pick three secret random polynomials $(\bm\zeta, \bm\omega, \bm \gamma)$ and (ii) compute $\bm\theta=\bm\zeta\cdot \bm\omega\cdot\bm \pi + \bm \gamma \bmod p$, where  $deg(\bm\zeta)= 1, deg(\bm\omega)=d,$ and   $deg(\bm\gamma)= 2d+1$. 

To verify whether $\bm\theta$ has been tampered with, a verifier (given $\bm\theta, \bm \gamma$, and $\bm\zeta$) can check if $\bm\theta-\bm \gamma$ is divisible by $\bm\zeta$. Informally, the security of an unforgeable polynomial states that an adversary (who does not know the three secret random polynomials) cannot tamper with an unforgeable polynomial without being detected, except with a negligible probability, in the security parameter. Below, we formally state it. 

\begin{theorem}[Unforgeable Polynomial]\label{proof::unforgeable-poly}
%Let polynomials $\zeta$ and $\gamma$ be two secret uniformly random polynomials (i.e., $\zeta, \gamma\stackrel{\st\$}\leftarrow \mathbb F_{\st p}[x]$),   $GCD(\zeta, \gamma)=1$, polynomial $\pi$ be an arbitrary polynomial,   $deg(\zeta)= 1, deg(\gamma)= d+1$,  $deg(\pi)=d$, and $p$ be a $\lambda$-bit prime number. Also, let polynomial $\theta$ be defined as  $\theta=\zeta\cdot \pi+ \gamma \bmod p$. Given $(\theta,\pi)$, the probability that a PPT adversary (which does not know $\zeta$ and $\gamma$) can forge $\theta$ to an arbitrary polynomial $\theta'$ such that  $\theta'\neq \theta$, $deg(\theta')\leq poly(\lambda)$, and $\zeta$ divides $\theta'-\gamma$ is negligible in the security parameter, i.e.,
%
Let polynomials $\bm\zeta$, $\bm\omega$, and $\bm\gamma$ be three secret uniformly random polynomials (i.e., $\bm\zeta,\bm\omega, \bm\gamma\stackrel{\st\$}\leftarrow \mathbb F_{\st p}[x]$), polynomial $\bm\pi$ be an arbitrary polynomial,   $deg(\bm\zeta)= 1, deg(\bm\omega)=d,  deg(\bm\gamma)= 2d+1$,  $deg(\bm\pi)=d$, and $p$ be a $\lambda$-bit prime number. Also, let polynomial $\bm\theta$ be defined as  $\bm\theta=\bm\zeta\cdot \bm\omega\cdot\bm \pi+\bm \gamma \bmod p$. Given $(\bm\theta,\bm\pi)$, the probability that an adversary (which does not know $\bm\zeta, \bm\omega$, and $\bm\gamma$) can forge $\bm\theta$ to an arbitrary polynomial $\bm\delta$ such that  $\bm\delta\neq \bm\theta$, $deg(\bm\delta)= const(\lambda)$, and $\bm\zeta$ divides $\bm\delta-\bm\gamma$ is negligible in the security parameter $\lambda$, i.e.,

$$Pr[ \bm\zeta \ | \ (\bm\delta-\bm\gamma) ]\leq \negl(\lambda)$$

\end{theorem}

\begin{proof}

Let $\bm\tau=\bm\delta-\bm\gamma$ and $\bm\zeta=a\cdot x+b$. Since $\bm\gamma$ is a random polynomial of degree $2d+1$ and unknown to the adversary, given $(\bm\theta, \bm\pi)$,  the adversary cannot learn anything about the factor $\bm\zeta$; as from its point of view every polynomial of degree $1$ in $\mathbb{F}_{\st p}[X]$ is equally likely to be $\bm\zeta$. Moreover,  polynomial $\bm\tau$ has at most $Max\big(deg(\bm\delta), 2d+1\big)$ irreducible non-constant factors.  For $\bm\zeta $ to divide $\bm\tau$,  one of the factors of $\bm\tau$ must be equal to $\bm\zeta$. We  also know that $\bm\zeta$ has been picked uniformly at random (i.e., $a,b
\stackrel{\st \$}\leftarrow \mathbb F_{\st p}$). Thus, the probability that $\bm\zeta $ divides $\bm\tau$ is negligible in the security parameter, $\lambda$. Specifically,

$$Pr[ \bm\zeta \ | \ (\bm\delta-\bm\gamma)]\leq \frac{Max\big(deg(\bm\delta), 2d+1\big)} {2^{\st \lambda}}=\negl(\lambda)$$
\end{proof} \hfill\(\Box\)

%$Max\big(deg(\theta'), d+1\big)$
 An interesting feature of an unforgeable polynomial is that the verifier can perform the check without needing to know the original polynomial $\bm\pi$. Another appealing feature of the unforgeable polynomial is that it supports \emph{linear combination} and accordingly \emph{batch verification}. Specifically, to turn $n$ arbitrary polynomials $[\bm\pi_{\st 1},..., \bm\pi_{\st n}]$ into unforgeable polynomials, one can construct  $\bm\theta_{\st i}=\bm\zeta\cdot \bm\omega_{\st i}\cdot \bm\pi_{\st i}+ \bm\gamma_{\st i} \bmod p$, where $\forall i, 1\leq i\leq n$.

To check whether all polynomials $[\bm\theta_{\st 1},..., \bm\theta_{\st n}]$ are intact, a verifier can (i) compute their sum $\bm \chi=\sum\limits_{\st i=1}^{\st n}\bm\theta_{\st i}$ and (ii) check whether $\bm \chi- \sum\limits_{\st i=1}^{\st n}\bm\gamma_{\st i} $ is divisible by $\bm \zeta$.  Informally, the security of an unforgeable polynomial states that an adversary (who does not know the three secret random polynomials for each $\bm\theta_{\st i}$) cannot tamper with any subset of the unforgeable polynomials without being detected, except with a negligible probability. We formally state it, below.

%%%%%%%%%%%%%%%%%%%%%%%%%%%%%%%
\begin{theorem}[Unforgeable Polynomials' Linear Combination]\label{Unforgeable-Polynomials-Linear-Combination}
 Let polynomial $\bm\zeta$ be a secret polynomial picked uniformly at random; also, let   $\vv{\bm\omega}=[\bm\omega_{\st 1},..., \bm\omega_{\st n}]$ and $\vv{\bm\gamma}=[\bm\gamma_{\st 1},..., \bm\gamma_{\st n}]$ be two vectors of secret uniformly random polynomials (i.e., ${\bm\zeta}, \bm\omega_{\st i}, \bm\gamma_{\st i} \stackrel{\st\$}\leftarrow \mathbb F_{\st p}[x]$),  $\vv{\bm\pi}=[\bm\pi_{\st 1},..., \bm\pi_{\st n}]$ be a vector of arbitrary polynomials,   $deg(\bm\zeta)= 1, deg(\bm\omega_{\st i})=d,  deg(\bm\gamma_{\st i})= 2d+1$,  $deg(\bm\pi_{\st i})=d$,  $p$ be a $\lambda$-bit prime number, and $1\leq i \leq n$. Moreover, let polynomial $\bm\theta_{\st i}$ be defined as  $\bm\theta_{\st i}=\bm\zeta\cdot \bm\omega_{\st i}\cdot \bm\pi_{\st i}+ \bm\gamma_{\st i} \bmod p$, and $\vv{\bm\theta} = [\bm\theta_{\st 1},..., \bm\theta_{\st n}]$.  Given $(\vv{\bm\theta}, \vv{\bm\pi})$, the probability that an adversary (which does not know $\bm\zeta, \vv{\bm\omega}$, and $\vv{\bm\gamma}$) can forge $t$ polynomials, without loss of generality, say $\bm\theta_{\st 1},..., \bm\theta_{\st t} \in \vv{\bm\theta}$ to arbitrary polynomials $\bm\delta_{\st 1},..., \bm\delta_{\st t}$ such that   $\sum\limits_{\st j=1}^{\st t}\bm\delta_{\st j}\neq \sum\limits_{\st j=1}^{\st t}\bm\theta_{\st j}$, $deg(\bm\delta_{\st j})= const(\lambda)$, and $\bm\zeta$ divides $(\sum\limits_{\st j=1}^{\st t}\bm\delta_{\st j} + \sum\limits_{\st j=t+1}^{\st n}\bm\theta_{\st j} - \sum\limits_{\st j=1}^{\st n}\bm\gamma_{\st j} )$ is negligible in the security parameter $\lambda$, i.e.,  

$$Pr[ \bm\zeta \ | \ (\sum\limits_{\st j=1}^{\st t}\bm\delta_{\st j} + \sum\limits_{\st j=t+1}^{\st n}\bm\theta_{\st j} - \sum\limits_{\st j=1}^{\st n}\bm\gamma_{\st j} ) ]\leq \negl(\lambda)$$

\end{theorem}

%%%%
\begin{proof}  
This proof is a generalisation of that of Theorem \ref{proof::unforgeable-poly}.  
Let $\bm\tau_{\st j}=\bm\delta_{\st j}-\bm\gamma_{\st j}$ and $\bm\zeta=a\cdot x+b$. Since  every $\bm\gamma_{\st j}$ is a random polynomial of degree $2d+1$ and unknown to the adversary, given $(\vv{\bm\theta}, \vv{\bm\pi})$,  the adversary cannot learn anything about the factor $\bm\zeta$. Each polynomial $\bm\tau_{\st j}$ has at most $Max\big(deg(\bm\delta_{\st j}), 2d+1\big)$ irreducible non-constant factors. 
%
%In order for $\bm\zeta$ to divide polynomial $\sum\limits_{\st j=1}^{\st t}\bm\delta_{\st j} + \sum\limits_{\st j=t+1}^{\st n}\bm\theta_{\st j} - \sum\limits_{\st j=1}^{\st n}\bm\gamma_{\st j}$  one of the factors of every $\bm\tau_{\st j}$ needs to equal $\bm\zeta$, where $1 \leq j \leq t$. 
%
We  know that $\bm\zeta$ has been picked uniformly at random (i.e., $a,b
\stackrel{\st \$}\leftarrow \mathbb F_{\st p}$). Therefore, the probability that $\bm\zeta$ divides $\sum\limits_{\st j=1}^{\st t}\bm\delta_{\st j} + \sum\limits_{\st j=t+1}^{\st n}\bm\theta_{\st j} - \sum\limits_{\st j=1}^{\st n}\bm\gamma_{\st j}$ equals the probability that $\bm\zeta$ equals to one of the factors of  every $\bm\tau_{\st j}$, that is negligible in the security parameter. Concretely,

$$Pr[ \bm\zeta \ | \ (\sum\limits_{\st j=1}^{\st t}\bm\delta_{\st j} + \sum\limits_{\st j=t+1}^{\st n}\bm\theta_{\st j} - \sum\limits_{\st j=1}^{\st n}\bm\gamma_{\st j} ) ]\leq  \frac{\prod \limits^{\st t}_{\st j=1}Max\big(deg(\bm\delta_{\st j}), 2d+1\big)} {2^{\st \lambda t}}=\negl(\lambda)$$
\hfill\(\Box\)
\end{proof}

It is not hard to see that, Theorem \ref{Unforgeable-Polynomials-Linear-Combination} is a generalisation of Theorem \ref{proof::unforgeable-poly}. Briefly, in \withFai, we will use unforgeable polynomials (and their linear combinations) to allow a smart contract to efficiently check whether the polynomials that the clients send to it are intact, i.e., they are \vopr's outputs.

\section{\withFai: A Concrete Construction of \p}

\subsection{Main Challenges to Overcome}

 We need to address several key challenges, to design an efficient scheme that realises \p. Below, we outline these challenges.

 \subsubsection{Keeping Overall Complexities Low.}
 
 In general, in multi-party PSIs, each client needs to send messages to the rest of the clients and/or engage in secure computation with them, e.g., in \cite{DBLP:conf/scn/InbarOP18,DBLP:conf/ccs/KolesnikovMPRT17}, which would result in communication and/or computation quadratic with the number of clients. To address this challenge, we  (a) allow one of the clients as a dealer to interact with the rest of the clients,\footnote{This approach has similarity with the non-secure PSIs in \cite{GhoshN19}.} and   (b) we use a smart contract, which acts as a bulletin board to which most messages are sent and also performs lightweight computation on the clients' messages. The combination of these approaches will keep the overall communication and computation linear with the number of clients (and sets' cardinality).

 \subsubsection{Securely Randomising Input Polynomials.}  In multi-party PSIs that rely on the polynomial representation, it is essential that an input polynomial of a client be randomised by another client \cite{AbadiMZ21}. To do that securely and efficiently, we require the dealer and each client together to engage in an instance of \vopr, which we developed in Section \ref{sec::subroutines}. 
 
 \subsubsection{Preserving the Privacy of Outgoing Messages.} Although the use of regular public smart contracts (e.g., Ethereum) will help keep overall complexity low, it introduces another challenge; namely, if clients do not protect the privacy of the messages they send to the smart contracts, then other clients (e.g., dealer) and non-participants of PSI (i.e., the public) can learn the clients' set elements and/or the intersection. Because standard smart contracts do not automatically preserve messages' privacy. To efficiently protect the privacy of each client's messages (sent to the contracts) from the dealer, we require the clients (except the dealer) to engage in \zspaa which lets each of them generate a pseudorandom polynomial with which it can blind its message. To protect the privacy of the intersection from the public, we require all clients to run a coin-tossing protocol to agree on a blinding polynomial, with which the final result that encodes the intersection on the smart contract will be blinded.  
 
 \subsubsection{Ensuring the Correctness of Subroutine Protocols' Outputs.} 
 
 In general, any MPC that must remain secure against an active adversary is equipped with a verification mechanism that ensures an adversary is detected if it deviates from the protocol and affects messages' integrity, during the protocol's execution. This is the case for the subroutine protocols that we use, i.e., \vopr and \zspaa. Nevertheless, this type of verification itself is not always sufficient. Because in certain cases, the output of an MPC protocol may be fed as input to another MPC and we need to ensure that the \emph{actual/intact} output of the first MPC is fed to the second one. This is the case in our PSI's subroutines as well. To address this challenge, we use unforgeable polynomials; specifically, the output of \vopr is an unforgeable polynomial (that encodes the actual output); if the adversary tampers with the \vopr's output and uses it later, then a verifier can detect it. We will have the same integrity guarantee for the output of \zspaa for free. Because (i) \vopr is called before \zspaa, and (ii) if clients use intact outputs of \zspaa, then the final result (i.e., the sum of all clients' messages) will not contain any output of \zspaa, as they would cancel out each other. Thus, by checking the correctness of the final result, one can ensure the correctness of the outputs of \vopr and \zspaa, in one go.

\subsection{Description of \withFai (\fpsi)}\label{Fair-PSI-Protocol}
%This section presents \fpsi, a protocol that realises \p. 

\subsubsection{An overview.} At a high level, \withFai (\fpsi) works as follows. First, each client encodes its set elements into a polynomial. All clients sign a smart contract and deposit a predefined amount of coins into it.  Next,  one of the clients as a dealer, $D$, randomises the rest of the clients' polynomials and imposes a certain structure to their polynomials. The clients also randomise $D$'s polynomials. The randomised polynomials reveal nothing about the clients' original polynomials representing their set elements. Then, all clients send their randomised polynomials to the smart contract.  The contract combines all polynomials and checks whether the resulting polynomial still has the structure imposed by $D$. If the contract concludes that the resulting polynomial does not have the structure, then it invokes an auditor, \aud, to identify misbehaving clients and penalise them. Nevertheless, if the resulting polynomial has the structure, then the contract outputs an encoded polynomial and refunds the clients' deposits. In this case, all clients can use the resulting polynomial (output by the contract) to locally find the intersection. Figure \ref{fig:parties-interactions-in-Jus} outlines the interaction between parties.

\begin{figure}[htp]
    \centering
    \includegraphics[width=14cm]{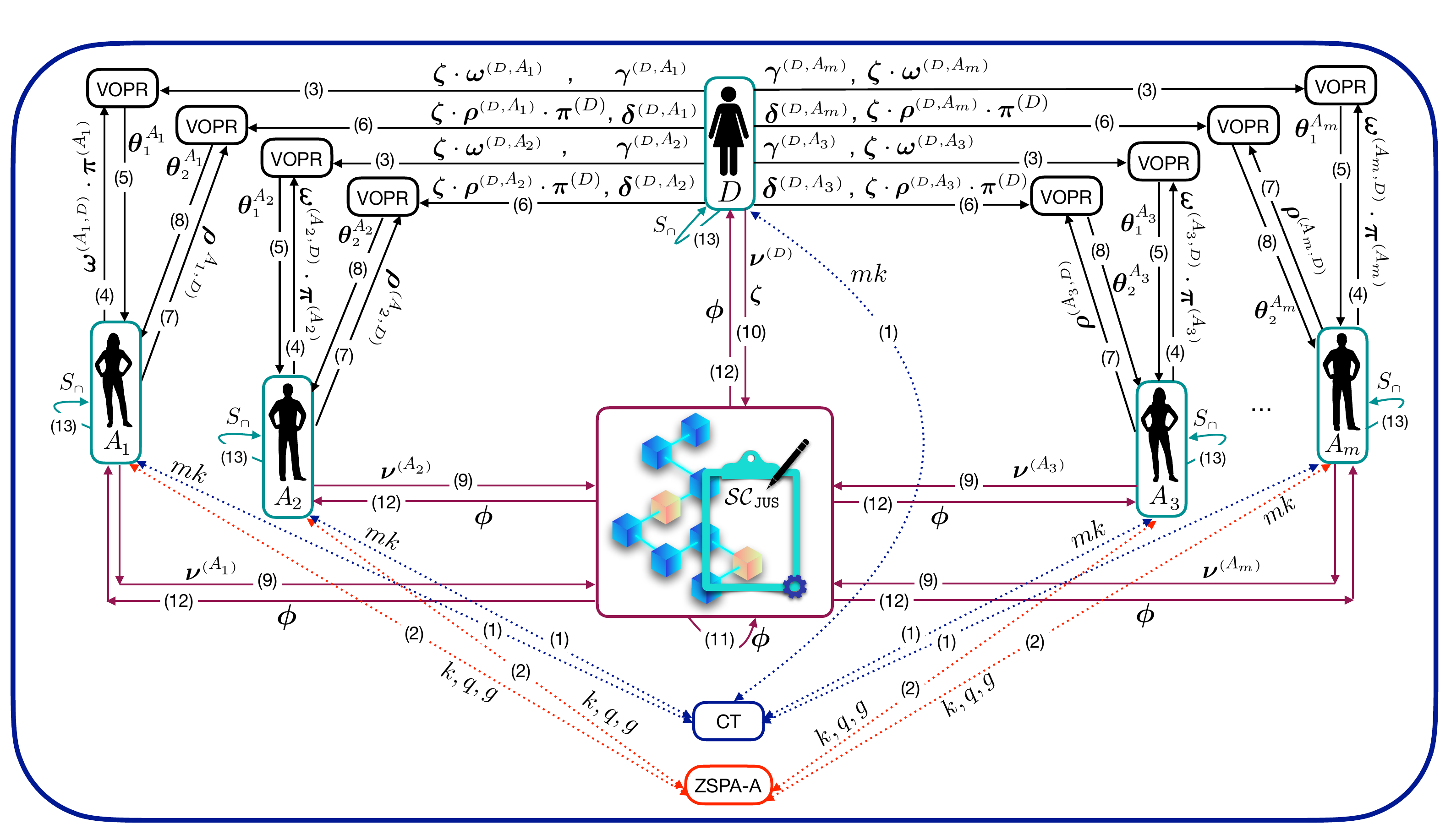}
    \caption{Outline of the interactions between parties in \withFai}\label{fig:parties-interactions-in-Jus}
\end{figure}

%%%%%%%%%%%%

One of the novelties of \fpsi is a lightweight verification mechanism which allows a smart contract to efficiently verify the correctness of the clients' messages without being able to learn/reveal the clients' set elements. To achieve this, $D$ randomises each client's polynomials and constructs unforgeable polynomials on the randomised polynomials (in one go). If any client modifies an unforgeable polynomial that it receives and sends the modified polynomial to the smart contract,  then the smart contract would detect it, by checking whether the sum of all clients' (unforgeable) polynomials is divisible by a certain polynomial of degree $1$.  The verification is lightweight because: (i) it does not use any public key cryptography (often computationally expensive), (ii) it needs to perform only polynomial division, and (iii) it can perform batch verification, i.e., it sums all clients randomised polynomials and then checks the result's correctness.

%%%%%%%%%%%%

%
%One of the novelties of F-PS is a lightweight verification mechanism which allows a smart contract to efficiently verify the correctness of the clients' messages without being able to learn/reveal the clients' set elements. To achieve this, the dealer during randomising other clients' polynomials, imposes a MAC-like structure on the randomised polynomials, such that if a client (who receives its randomised polynomial) tampers with it, then the smart contract would detect it. To do the verification, the smart contract needs to only check whether the sum of all clients' randomised polynomials is divisible by a polynomial of degree $1$.  The verification is lightweight because: (i) it does not rely on any public key cryptography (i.e., zero-knowledge proofs), (ii) it needs to perform only polynomial division, and (iii) it can perform batch verification, i.e., instead of individually checking each client's randomised polynomial, it sums all clients randomised polynomials (related to a hash table's bin) and then checks the result's correctness.

% his own
%outsourced dataset and having any knowledge of the other client’s dataset 
%
%
% mainly stems from our observation (stated  in Theorem \ref{proof::unforgeable-poly}) which leads to an  efficient verification mechanism carried out by the contract. 

Now, we describe \fpsi in more detail. First, all clients sign and deploy a  smart contract, \scf. Each of them put a certain amount of deposit into it. Then, they together run \ct to agree on a key, $mk$, that will be used to generate a set of blinding polynomials to hide the final result from the public. Next, each client locally maps its set elements to a hash table and represents the content of each hash table's bin as a polynomial, $\bm\pi$. After that, for each bin, the following steps are taken.  All clients, except $D$, engage in \zspaa to agree on a set of pseudorandom blinding factors such that their sum is zero.  %The clients will use these polynomials to hide from $D$ the polynomials that they will send to \scf. 

Then, $D$ randomises and constructs an unforgeable polynomial on each client's polynomial, $\bm\pi$. To do that, $D$ and every other client engage in \vopr that returns to the client a polynomial. $D$ and every other client invoke \vopr again to randomise $D$'s polynomial. \vopr returns to the client another unforgeable polynomial. Note that the output of \vopr reveals nothing about any client's original polynomial $\bm\pi$, as this polynomial has been blinded/encrypted with another secret random polynomial by $D$, during the execution of \vopr. Each client sums the two polynomials,  blinds the result (using the output of  \zspaa), and sends it to \scf.

After all of the clients send their input polynomials to \scf, $D$ sends to \scf a \emph{switching polynomial} that will allow \scf to obliviously switch the secret blinding polynomials used by $D$ (during the execution of \vopr) to blind each client's original polynomial $\bm\pi$  to another blinding polynomial that all clients can generate themselves, by using key $mk$.  The switching polynomial is constructed in a way that does not affect the verification of unforgeable polynomials.

Next, $D$ sends to \scf a secret polynomial, $\bm\zeta$, that will allow \scf to check unforgeable polynomials' correctness. Then, \scf\ sums all clients' polynomials and checks if $\bm\zeta$ can divide the sum. \scf\ accepts the clients' inputs if the polynomial divides the sum; otherwise, it invokes \aud to identify misbehaving parties.  In this case, all honest parties' deposit is refunded to them and the deposit of misbehaving parties is distributed among the honest ones as well. If all clients behave honestly,  then each client can locally find the intersection. To do that, it uses $mk$ to locally remove the blinding polynomial from the sum (that the contract generated), then evaluates the unblinded polynomial at each of its set elements and considers an element in the intersection when the evaluation equals zero. 
%The efficiency of the verification in our protocol  mainly stems from our  observation that if an adversary who know only $xx$ modified the polynomial of the form $xx$ then $\zeta$ will not divide result polynomial after unblinding will not divide with a high probability.  

\subsubsection{Detailed Description of \fpsi.} Below, we elaborate on how \fpsi exactly works (see Table \ref{table:notation-table} for description of the main notations used). 

\begin{enumerate}

%\item[$\bullet$] \textbf{Input:} a pseudorandom function: $\mathtt{PRF}$, a hash table's parameters (i.e., the  total number of bins: $h$ and a bin's capacity: $d$), and clients' sets: $S^{\st (I)}$, where $I\in \bar{P}$.

%\item[$\bullet$] \textbf{Output:}  for every bin of the hash table, it outputs a polynomial: $\phi$, whose roots are  encrypted sets elements (of the bin) in the intersection.
\item\label{gen-FPSI-cont} All clients in $\cl=\{ A_{\st 1},...,   A_{\st m},  D\}$ sign a smart contract: \scf and deploy it to a blockchain. All clients get the deployed contract's address. Also, all clients engage in \ct to agree on a secret master key, $mk$.

\item \label{encode-encrypt} Each client in $\cl$  builds a  hash table,  $\mathtt{HT}$, and inserts the set elements into  it, i.e.,  $\forall i: \mathtt{H}( s_{\st i})={indx}$, then $ s_{\st i}\rightarrow \mathtt{HT}_{\st indx}$. It pads every bin with random dummy elements to $d$ elements (if needed). Then,  for every bin, it constructs a polynomial whose roots are the bin's content: $\bm\pi=\prod\limits^{\st d}_{\st i=1} (x-s'_{\st i})$, where $s'_{\st i}$ is either $s_{\st i}$ or a random value. 
\item \label{ZSPA} Every client $ C$ in $\cl\setminus D$, for every bin, agree on $b=3d+3$ vectors of pseudorandom blinding factors: $z_{\st i,j}$, such that the sum of each vector elements is zero, i.e., $\sum\limits^{\st m}_{\st j=1}z_{\st i,j}=0$, where $0\leq i\leq b-1$. To do that, they participate in step \ref{ZSPA::ZSPA-invocation} of \zspaa. By the end of this step, for each bin, they agree on a secret key $k$ (that will be used to generate the zero-sum values) as well as two values stored in $\mathcal{SC}_{\fpsi}$, i.e., $q$: the key's hash value and $g$: a Merkle tree's root. After time $t_{\st 1}$,  $D$ ensures that all other clients have agreed on the vectors (i.e., all provided ``approved''  to the contract). If the check fails, it halts. 
\item\label{F-PSI::each-client-deposit} Each client in $\cl$ deposits $\yc+\chc$ amount to $\mathcal{SC}_{\fpsi}$. After time $t_{\st 2}$, every client ensures that in total $(\yc+\chc)\cdot (m+1)$ amount has been deposited. Otherwise, it halts and the clients' deposit is refunded.

\item\label{JUS::check-non-zero-coeff}  $D$ picks a  random polynomial $\bm\zeta \stackrel{\st\$}\leftarrow \mathbb{F}_{\st p}[X]$ of degree $1$, for each bin.  
It, for each client $C$, allocates to each bin two random polynomials: $\bm\omega^{\st(D,C)}, \bm\rho^{\st (D,C)}\stackrel{\st\$}\leftarrow \mathbb{F}_{\st p}[X]$ of degree $d$, and  two  random polynomials: $\bm\gamma^{\st (D,C)}, \bm\delta^{\st (D,C)} \leftarrow \mathbb{F}_{\st p}[X]$ of degree $3d+1$. Also, each client $C$, for each bin, picks two  random polynomials: $\bm\omega^{\st (C,D)}, \bm\rho^{\st (C,D)}\stackrel{\st\$}\leftarrow \mathbb{F}_{\st p}[X]$ of degree $d$, and ensures polynomials $\bm\omega^{\st (C,D)}\cdot \bm\pi^{\st  {  {(C)}}}$ and  $\bm\rho^{\st (C,D)}$ do not contain zero coefficients. %It also evaluates each polynomial at every element of $\vv{\bm{x}}$ that results in  $\omega^{  {  {D,C}}}_{\st i}$ and $\rho^{  {  {D,C}}}_{\st i}$.

\item\label{e-psi::D-randomises}  $D$ randomises other clients' polynomials. To do so, for every bin, it invokes an instance of {\vopr} (presented in Fig. \ref{fig:VOPR}) with  each client $  C$; where  $D$ sends $\bm\zeta \cdot \bm\omega^{\st  {  {(D,C)}}}$ and $\bm\gamma^{\st  {  {(D,C)}}}$, while client $ C$ sends $\bm\omega^{\st  {  {(C,D)}}}\cdot \bm\pi^{\st  {  {(C)}}}$ to {\vopr}. Each client $C$, for every bin, receives a blind polynomial of the following form: 

$$\bm\theta^{  {  {(C)}}}_{\st 1}=\bm\zeta \cdot \bm\omega^{\st  {  {(D,C)}}}\cdot \bm\omega^{\st  {  {(C,D)}}}\cdot \bm\pi^{\st  {  {(C)}}}+\bm\gamma^{\st  {  {(D,C)}}}$$
 from {\vopr}. If any party aborts, the deposit would be refunded to all parties.

\item\label{e-psi::C-randomises} Each client $    {  C}$ randomises  $ {D}$'s polynomial. To do that, each client $    {  C}$, for each bin,  invokes an instance of {\vopr} with   $ {D}$,    where each client $    {  C}$  sends $\bm\rho^{\st  {  {(C,D)}}}$, while  ${D}$  sends $\bm\zeta\cdot\bm \rho^{\st  {  {(D,C)}}}\cdot \bm\pi^{  {  {(D)}}}$ and $\bm\delta^{\st  {  {(D,C)}}}$ to {\vopr}. Every client   $    {  C}$, for each bin,  receives a blind polynomial of the following form: 

$$\bm\theta^{  {  {(C)}}}_{\st 2}=\bm\zeta \cdot \bm\rho^{\st  {  {(D,C)}}}\cdot \bm\rho^{\st  {  {(C,D)}}}\cdot \bm\pi^{\st  {  {(D)}}}+\bm\delta^{\st  {  {(D,C)}}}$$
 from {\vopr}. If any party aborts, the deposit would be refunded to all parties.

\item\label{blindPoly-C-sends-to-contract} Each client $ C$, for every bin, masks the sum of polynomials $\bm\theta^{\st  {  {(C)}}}_{\st 1}$ and $\bm\theta^{\st  {  {(C)}}}_{\st 2}$  using the blinding factors: $z_{\st i,c}$, generated in step \ref{ZSPA}. Specifically, it computes the following blind polynomial (for every bin):  

$$\bm\nu^{ \st {  {(C)}} }= \bm\theta^{ \st {  {(C)}}}_{\st 1}+\bm\theta^{\st  {  {(C)}}}_{\st 2}+\bm \tau^{\st  {  {(C)}} }$$

where $\bm\tau^{\st  {  {(C)}}}=\sum\limits^{\st 3d+2}_{\st i=0}z_{\st i,c}\cdot x^{\st i}$. Next, it sends  all $\bm\nu^{\st  {  {(C)}} }$ to $\mathcal{SC}_{\fpsi}$. If any party aborts, the deposit would be refunded to all parties.

%\item Client $    {  D}$ ensures all clients have sent their inputs to $\mathcal{SC}_{  {   {F-PSI}}}$. In the case where $m'$ parties do not provide their inputs, client $    {  D}$ aborts. In this case, the rest (including the dealer) get their deposit back. Also,  the deposit of the parties who did not send  inputs would be evenly distributed among the rest. The total amount each party above receives is: $y+\frac{m'\cdot y}{m-m'}$

%%
%\item Client $    {  D}$ and each client $    {  C}$ collaboratively, for each bin, generate a polynomial that will be used to (obliviously) check if  $    {  C}$ misbehaved during the computation of each $\bm\nu^{  {  {(C)}} }$. To do so, for every bin, client $    {  D}$ invokes an instance of $\mathtt{VOPR}$ with  each client $    {  C}$, where  client $    {  D}$ sends: $\bm\zeta$, while client $    {  C}$ sends $\bm\xi^{  {  {(C)}}}$ and $-\bm\tau^{  {  {(C)}}}$   to $\mathtt{VOPR}$, where $\bm\xi^{  {  {(C)}}}$ is a random polynomial of degree $3d+1$. Client $    {  D}$ for each  $    {  C}$'s bin recives the following polynomial: 
%%
%$$\bm\mu^{  {  {(D,C)}}} = \bm\zeta\cdot \bm\xi^{  {  {(C)}}}-\bm\tau^{  {  {(C)}}}$$
%%

\item\label{f-psi::D-gen-random-poly} ${D}$ ensures all clients sent their inputs to $\mathcal{SC}_{\fpsi}$. If the check fails, it halts and the deposit would be refunded to all parties. It allocates a fresh pseudorandom polynomial $\bm\gamma'$ of degree $3d$, to each bin. To do so, it uses $mk$ to derive a key for each bin: $k_{\st  { {indx}}}=\mathtt{PRF}(mk, {    {   {indx}}})$ and then uses the derived key to generate $3d+1$ pseudorandom coefficients $g_{\st  { {j,indx}}}=\mathtt{PRF}(k_{\st  { {indx}}}, j)$ where $ 0\leq j \leq 3d$. Also, for each bin, it allocates a fresh random polynomial:  $\bm\omega'^{\st  {  {(D)}}}$ of degree $d$. 

\item\label{f-psi::D-gen-switching-poly}  $ {D}$,  for every bin, computes a polynomial of the form:  

$$\bm\nu^{\st  {  {(D)}}}=\bm\zeta \cdot  \bm\omega'^{\st  {  {(D)}}}\cdot \bm\pi^{\st  {  {(D)}} }-\sum\limits^{\st  {   A}_{\st  {   m}}}_{\st   {  {C }= }   {   A}_{\st  {  1}}}(\bm\gamma^{\st  {  {(D,C)}}} + \bm\delta^{\st  {  {(D,C)}}}) + \bm\zeta \cdot \bm\gamma'$$ 
It sends to $\mathcal{SC}_{\fpsi}$  polynomials $\bm\nu^{\st  {  {(D)}}}$ and $\bm\zeta$, for each bin.

 \item\label{compute-res-poly}  $\mathcal{SC}_{\fpsi}$ takes the following steps:
 \begin{enumerate}
 \item for every bin, sums all related polynomials  provided by all clients in $\bar{P}$:
 
 \begin{equation*}
\begin{split}
 \bm\phi&= \bm\nu^{\st  {  {(D)}} }+\sum\limits^{\st  {   A}_{\st  {   m}}}_{\st   {  {C }= }   {   A}_{\st  {  1}}}\bm\nu^{\st  {  {(C)}} }\\
 &= \bm\zeta\cdot \bigg(\bm\omega'^{\st  {  {(D)}}}\cdot \bm\pi^{\st  {  {(D)}} } +\sum\limits^{\st  {   A}_{  {   m}}}_{\st  {  {C }= }   {   A}_{\st  {  1}}}(\bm\omega^{\st  {  {(D,C)}}} \cdot \bm\omega^{\st  {  {(C,D)}}}\cdot \bm\pi^{\st  {  (C)}}) +\bm\pi^{\st  {  {(D)}}}\cdot\sum\limits^{\st  {   A}_{ \st {   m}}}_{\st  {C= }   {   A}_{\st  {  1}}}(\bm\rho^{\st  {  {(D,C)}}} \cdot \bm\rho^{\st  {  {(C,D)}}}) + \bm\gamma'\bigg)
  \end{split}
\end{equation*}
% \item\label{F-PSI:detect-misbehaving-party} ensures that, for every bin, $\bm\zeta$ divides $\bm\phi$. Otherwise, it aborts and Arbiter protocol (presented in Fig. \ref{fig:arbiter}) is invoked to find misbehaving parties.
 
  \item\label{F-PSI:detect-misbehaving-party} checks whether, for every bin, $\bm\zeta$ divides $\bm\phi$. If the check passes, it sets $Flag=True$. Otherwise, it sets $Flag=False$. 
  
   %aborts and Arbiter protocol (presented in Fig. \ref{fig:arbiter}) is invoked to find misbehaving parties.

% \item if the check passes (i.e., $Flag=True$), each party gets back its deposit (i.e., $y$ amount).
 \end{enumerate}
 
\item\label{F-PSI::flag-is-true} If the above check passes (i.e., $Flag=True$), then the following steps are taken:

\begin{enumerate}
 \item $\mathcal{SC}_{\fpsi}$ sends back each party's deposit, i.e., $\yc+\chc$ amount.
 
  \item each client (given $\bm\zeta$ and $mk$) finds the elements in the intersection as follows. 
  \begin{enumerate}
  \item derives a bin's pseudorandom polynomial, $\bm\gamma'$, from $mk$. 
  
  \item removes the blinding polynomial from each bin's polynomial: 
  
  $$\bm\phi'=\bm\phi-\bm\zeta\cdot \bm\gamma'$$ 
  
  \item\label{F-PSI::find-intersection} evaluates each bin's unblinded polynomial at every element $s_{\st i}$ belonging to that bin and considers the element in the intersection if the evaluation is zero: i.e., $\bm\phi'(s_{\st i})=0$.
 
 \end{enumerate}

 \end{enumerate}
 
 \item\label{F-PSI::flag-is-false} If the check does not pass (i.e., $Flag=False$), then the following steps are taken.

 \begin{enumerate}

 \item\label{auditor}  \aud asks every client $    {  C}$ to send to it the  $\mathtt{PRF}$'s key (generated in step \ref{ZSPA}), for every bin. It inserts the keys to $\vv k$.  It generates a list $\bar L$ initially empty. Then, for every bin,  \aud takes step \ref{ZSPA-A::Auditor-computation} of \zspaa, i.e., invokes  $\mathtt{Audit}(m, \vv{{k}},  q, \bm\zeta, 3d+3, g)\rightarrow (L, \vv{{\mu}})$.  Every time it invokes $\mathtt{Audit}$, it appends the elements of returned $L$ to $\bar L$.  \aud for each bin sends  $ \vv{{\mu}}$ to $\mathcal{SC}_{\fpsi}$. It also sends  to $\mathcal{SC}_{\fpsi}$ the list $\bar L$ of all misbehaving clients detected so far.

 \item to  help identify further  misbehaving clients, $D$ takes the following steps,  for each bin of client $    {  C}$ whose ID is not in $\bar L$.   
 \begin{enumerate}
 \item\label{gen-unmaking-poly} computes polynomial $\bm\chi^{  {  {(D, C)}}}$ as follows. 
 
 $$\bm\chi^{ \st {  {(D, C)}}}=\bm\zeta\cdot \bm\eta^{ \st {  {(D,C)}}}-(\bm\gamma^{ \st {  {(D,C)}}}+\bm\delta^{ \st {  {(D,C)}}})$$
 
 %+\bm\mu^{  {  {(D, C)}}}
 
  where $\bm\eta^{ \st {  {(D,C)}}}$ is a fresh random polynomial of degree $3d+1$. 
  
  \item\label{send-unmaking-poly} sends  polynomial $\bm\chi^{ \st {  {(D, C)}}}$ to  $\mathcal{SC}_{\fpsi}$.

 \end{enumerate}
  Note, if $\bar L$ contains all clients' IDs, then $D$ does not need to take the above steps \ref{gen-unmaking-poly} and \ref{send-unmaking-poly}. 
 %%%%%%%%%%%%%%%%%%%%%%
 
 \item  $\mathcal{SC}_{\fpsi}$,   takes the following steps to check if the client misbehaved,  for each bin of client $    {  C}$ whose ID is not in $\bar L$.
 
 %for each client $    {  C}$'s bin, takes the following steps to check if the client misbehaved.
 
  \begin{enumerate}
  
 \item computes  polynomial $\bm\iota^{\st  {  {(C)}}}$ as follows: 
 
  \begin{equation*}
\begin{split}
 \bm\iota^{ \st {  {(C)}}}&=\bm\chi^{\st  {  {(D, C)}}}+\bm\nu^{\st  {  {(C)}}} +\bm\mu^{ \st {  {(C)}}} \\ 
 &=\bm\zeta\cdot(\bm\eta^{ \st {  {(D,C)}}} + \bm\omega^{ \st {  {(D,C)}}}\cdot \bm\omega^{ \st {  {(C,D)}}}\cdot \bm\pi^{ \st {  {(C)}}}+\bm\rho^{ \st {  {(D,C)}}}\cdot \bm\rho^{ \st {  {(C,D)}}}\cdot \bm\pi^{ \st {  {(D)}}}+\bm\xi^{ \st {  {(C)}}})
 \end{split}
\end{equation*}

 where $\bm\mu^{ \st {  {(C)}}} \in \vv{\mu}$ generated and sent to $\mathcal{SC}_{\fpsi}$  by \aud in step \ref{auditor}.   
  \item checks if $\bm\zeta$  divides $\bm\iota^{ \st {  {(C)}}}$. If the check fails, it appends the client's ID to  a list $ L'$.
  \end{enumerate}
   If $\bar L$ contains all clients' IDs, then $\mathcal{SC}_{\fpsi}$ does not take the above two steps. 

   \item  $\mathcal{SC}_{\fpsi}$  refunds the honest parties' deposit. Also, it retrieves the total amount of  $\chc$ from the deposit of dishonest clients (i.e., those clients whose IDs are in $\bar L$ or $L'$) and sends it to \aud.  It also splits the remaining deposit of the misbehaving parties among the honest ones. Thus, each honest client  receives $\yc+\chc+\frac{m'\cdot (\yc+\chc)-\chc}{m-m'}$ amount in total, where $m'$ is the total number of misbehaving parties.

%  \item  $\mathcal{SC}_{  {   {F-PSI}}}$  refunds the honest parties' deposit and splits the deposit of the misbehaving parties (i.e., those clients whose IDs are in $\bar L$ or $L'$)  among the honest ones. Thus, each honest party would receive $y+\frac{m'\cdot y}{m-m'}$ amount in total, where $m'$ is the total number of misbehaving parties.
 %%%%%%%%%%%%%%%%%%%%%
  \end{enumerate}

  \end{enumerate}
 
% the result: $cccc$ by locally evaluating the result polynomial: $\phi(x)$, at every  encrypted element, $e^{  {  {(I)}}}_{\st i}$, it has and considering the elements in the intersection if the following equation holds.  $\forall i, 1\leq i\leq d: \phi(e^{  {  {(I)}} }_{\st i})-\zeta(e^{  {  {(I)}}}_{\st i})\cdot \gamma'(e^{  {  {(I)}}}_{\st i})=0$.

%\begin{remark} After the Arbiter detects  misbehaving parties,  in step \ref{F-PSI:detect-misbehaving-party},  it sends their ID's to $\mathcal{SC}_{  {   {F-PSI}}}$ which refunds the honest parties' deposit and splits the misbehaving parties' deposit among the honest ones. Thus, each honest party would receive: $y+\frac{m'\cdot y}{m-m'}$ amount in total, where $m'$ is the total number of misbehaving parties.
% \end{remark}

%\begin{remark}
One may be tempted to replace $\withFai$ with a scheme in which all clients send their encrypted sets to a server (potentially semi-honest and plays \aud's role) which computes the result in a privacy-preserving manner.  We highlight that the main difference is that in this (hypothetical) scheme the server is \emph{always involved};  whereas, in our protocol, \aud remains offline as long as the clients behave honestly and it is invoked only when the contract detects misbehaviours.  

In step \ref{JUS::check-non-zero-coeff}, \fpsi requires each client $C$ to ensure polynomials $\bm\omega^{\st (C,D)}\cdot \bm\pi^{\st  {  {(C)}}}$ and  $\bm\rho^{\st (C,D)}$ do not contain zero coefficients. However,  this check can be removed from this step, if we allow the protocol to output error with a negligible probability. We refer readers to Appendix \ref{sec::error-prob} for further discussion. 

%\end{remark}

 Next, we present a theorem that formally states the security of \fpsi. 
 
 \begin{theorem}\label{theorem::F-PSI-security}
Let polynomials $\bm\zeta$, $\bm\omega$, and $\bm\gamma$ be three secret uniformly random polynomials. If  $\bm\theta=\bm\zeta\cdot \bm\omega\cdot\bm \pi+\bm \gamma \bmod p$ is an unforgeable polynomial (w.r.t. Theorem \ref{proof::unforgeable-poly}), \zspaa, \vopr,  $\mathtt{PRF}$, and smart contracts are secure, then \fpsi securely realises  $f^{\st \text{PSI}}$ with $Q$-fairness (w.r.t. Definition \ref{def::PSI-Q-fair}) in the presence of $m-1$  active-adversary clients (i.e., $A_{\st j}$s) or a passive dealer client, passive auditor, or passive public. 
 \end{theorem}

 % !TEX root =main.tex

\subsection{Proof of \fpsi}\label{sec::F-PSI-proof}

In this section, we prove Theorem \ref{theorem::F-PSI-security}, i.e., the security of \fpsi.

\begin{proof}
We prove Theorem  \ref{theorem::F-PSI-security} by considering the case where each party is corrupt, at a time.

\

\noindent\textbf{Case 1: Corrupt $m-1$ clients in $\{  {  A}_{ \st {   1}}, ...,   {  A}_{ \st {   m}}\}$}.  Let $P'$ be a set of at most $m-1$ corrupt clients, where $P'\subset \{  {  A}_{ \st {   1}}, ...,   {  A}_{ \st {   m}}\}$. Let set $\hat P$ be defined as follows: $\hat P=\{  {  A}_{ \st {   1}}, ...,   {  A}_{ \st {   m}}\}\setminus P'$. Also, let $\mathsf{Sim}^{\st\fpsi}_{\st A}$ be the simulator, which uses a subroutine adversary, $\mathcal{A}$.  Below, we explain how $\mathsf{Sim}^{\st \fpsi}_{\st A}$ (which receives the input sets of honest dealer $D$ and honest client(s) in $\hat P$) works.

\begin{enumerate}
\item constructs and deploys a smart contract. It sends the contract's address to $\mathcal{A}$. 
\item simulates \ct and receives the output value, $ {mk}$, from its functionality, $f_{\st \ct}$.
\item\label{sim::ZSPA-A-invocation} simulates \zspaa for each bin and receives the output value, $( k,  g,  q)$, from its functionality, $f^{\st \zspaa}$.
\item deposits in the contract the total amount of $(\yc+\chc)\cdot (m-|P'|+1)$ coins on behalf of $D$ and honest client(s) in $\hat P$. It sends to $\mathcal{A}$ the amount deposited in the contract. 
\item checks if $\mathcal{A}$ has deposited $(\yc+\chc)\cdot |P'|$ amount. If the check fails, it instructs the ledger to refund the coins that every party deposited and sends message $abort_{\st 1}$ to TTP (and accordingly to all parties); it outputs whatever $\mathcal{A}$ outputs and then halts. 

\item picks a random polynomial ${\bm\zeta}$ of degree $1$, for each bin. Also, $\mathsf{Sim}^{\st \fpsi}_{\st A}$, for each client $  {  C}\in \{  {  A}_{ \st {   1}}, ...,   {  A}_{ \st {   m}}\}$ allocates to each bin two random polynomials: (${\bm\omega}^{ \st {  {(D,C)}}}, {\bm\rho}^{ \st {  {(D,C)}}}$) of degree $d$ and   two random polynomials: (${\bm\gamma}^{ \st {  {(D,C)}}}$, ${\bm\delta}^{ \st {  {(D,C)}}}$) of degree $3d+1$. Moreover, $\mathsf{Sim}^{\st \fpsi}_{\st A}$ for every honest client $C'\in \hat P$, for each bin, picks two random polynomials: (${\bm\omega}^{ \st {  {(C',D)}}}$, ${\bm\rho}^{ \st {  {(C',D)}}}$) of degree $d$. %Note, all polynomials are picked from $\mathbb{F}_{\st p}[X]$. 
\item\label{F-PSI::sim-A-first-VOPR-invocation} simulates \vopr using inputs ${\bm\zeta} \cdot {\bm\omega}^{ \st {  {(D,C)}}}$ and ${\bm\gamma}^{ \st {  {(D,C)}}}$ for each bin. Accordingly, it receives the inputs of clients $C''\in P'$, i.e., ${\bm\omega}^{ \st {  {(C'',D)}}}\cdot {\bm\pi}^{ \st {  {(C'')}}}$, from its functionality $f^{\st \vopr}$, for each bin.  
%
%\item receives from $\mathcal{A}$ the outputs of VOPR for every client $C'' \in P'$. It ensures the output for every $C''$ has been provided. Otherwise, it halts.  
%
\item extracts the roots of polynomial ${\bm\omega}^{ \st {  {(C'',D)}}}\cdot {\bm\pi}^{ \st {  {(C'')}}}$ for each bin and appends those roots that are in the sets universe to a new set $S^{\st(C'')}$. 
\item simulates  \vopr again using inputs ${\bm\zeta}\cdot {\bm \rho}^{ \st {  {(D,C)}}}\cdot {\bm\pi}^{ \st {  {(D)}}}$ and ${\bm\delta}^{ \st {  {(D,C)}}}$, for each bin. %It receives the inputs of clients $C''\in P'$, i.e., $ {\bm\rho}^{ \st {  {(C'',D)}}}$, from $f^{\st \text{VOPR}}$, for each bin. 
\item sends to TTP the input sets of all parties; namely, (i) client $D$'s input set: $S^{\st (D)}$, (ii) honest clients' input sets: $S^{\st (C')}$ for all $C'$ in $\hat P$, and (iii) $\mathcal{A}$'s input sets: $S^{\st(C'')}$, for all $C''$ in $P'$.  For each bin, it receives the intersection set, $S_{\st\cap}$, from TTP. 
\item represents the intersection set for each bin as a polynomial, ${\bm \pi}$, as follows: ${\bm \pi}=\prod\limits^{\st |S_{\st\cap}|}_{\st i=1 }(x-s_{\st i})$, where $s_{\st i}\in S_{\st\cap}$. 

\item constructs polynomials ${\bm\theta}^{ \st {  {(C')}}}_{\st 1}={\bm\zeta} \cdot {\bm\omega}^{ \st {  {(D,C')}}}\cdot {\bm\omega}^{ \st {  {(C',D)}}}\cdot {\bm\pi}+{\bm\gamma}^{ \st {  {(D,C')}}}, {\bm\theta}^{ \st {  {(C')}}}_{\st 2}=  {\bm\zeta} \cdot  {\bm\rho}^{ \st {  {(D,C')}}}\cdot  {\bm\rho}^{ \st {  {(C',D)}}}\cdot  {\bm\pi}+  {\bm\delta}^{ \st {  {(D,C')}}}$, and $ {\bm\nu}^{ \st {  {(C')}} }=  {\bm\theta}^{ \st {  {(C')}}}_{\st 1}+ {\bm\theta}^{ \st {  {(C')}}}_{\st 2}+ {\bm \tau}^{ \st {  {(C')}} }$, for each bin and each honest client $C'\in\hat P$, where $ {\bm\tau}^{ \st {  {(C)}}}=\sum\limits^{\st 3d+2}_{\st i=0}z_{\st i,c}\cdot x^{\st i}$ and each $z_{\st i,c}$ is derived from $  k$. 
\item sends to $\mathcal{A}$ polynomial  $ {\bm\nu}^{ \st {  {(C')}} }$ for each bin and each client $C'$. 
\item\label{F-PSI::sim-A-receive-nu-from-adv} receives $ {\bm\nu}^{ \st {  {(C'')}} }$  from $\mathcal{A}$, for each bin and each client $C''\in P'$. It ensures that the output for every $C''$ has been provided. Otherwise, it halts. 
\item if there is any abort within steps \ref{F-PSI::sim-A-first-VOPR-invocation}--\ref{F-PSI::sim-A-receive-nu-from-adv}, then it sends $abort_{\st 2}$ to TTP and instructs the ledger to refund the coins that every party deposited.  It outputs whatever $\mathcal{A}$ outputs and then halts. 
\item constructs polynomial  $ {\bm\nu}^{ \st {  {(D)}}}= {\bm\zeta} \cdot   {\bm\omega'}^{ \st {  {(D)}}}\cdot  {\bm\pi} - \sum\limits^{ \st {   A}_{ \st {   m}}}_{  \st {  {C }= }  \st {   A}_{ \st {  1}}}( {\bm\gamma}^{ \st {  {(D,C)}}} +  {\bm\delta}^{ \st {  {(D,C)}}}) +  {\bm\zeta} \cdot  {\bm\gamma'}$ for each bin, where $ {\bm\omega'}^{ \st {  {(D)}}}$ is a fresh random polynomial of degree $d$ and $ {\bm\gamma'}$ is a pseudorandom polynomial derived from $  {mk}$.
\item sends to $\mathcal{A}$ polynomials $ {\bm\nu}^{ \st {  {(D)}}}$ and $ {\bm\zeta}$ for each bin. 
\item given each $ {\bm\nu}^{ \st {  {(C'')}} }$, computes polynomial $ {\bm\phi'}$ as follows $ {\bm\phi'} = \sum\limits_{\st \forall C''\in P'} {\bm\nu}^{ \st {  {(C'')}} }- \sum\limits_{\st \forall C''\in P'}( {\bm\gamma}^{ \st {  {(D,C'')}}} +  {\bm\delta}^{ \st {  {(D,C'')}}})$, for every bin. Then, $\mathsf{Sim}^{\st \fpsi}_{\st A}$ checks whether  $ {\bm\zeta}$  divides ${\bm\phi'}$, for every bin. If the check passes, it sets $Flag=True$. Otherwise, it sets $Flag=False$. 
\item if $Flag=True$:

\begin{enumerate}
 \item instructs the ledger to send back each party's deposit, i.e., $\yc+\chc$ amount. It sends a message $deliver$ to TTP.  
 \item outputs whatever $\mathcal{A}$ outputs and then halts. 
 \end{enumerate}
\item if $Flag=False$: 
\begin{enumerate}
 \item receives $|P'|$ keys of the $\mathtt{PRF}$ from $\mathcal{A}$, i.e., $\vv k'=[   k'_{\st 1}, ...,   k'_{\st |P'|}]$, for every bin. %, which should be the output of $f^{\st \text{ZSPA-A}}$ in step \ref{sim::ZSPA-A-invocation} above. 
\item checks whether the following equation holds: $  k'_{\st j}=  k$, for every $  k'_{\st j}\in\vv k'$. Note that $  k$ is the output of $f^{\st \zspaa}$ generated in step \ref{sim::ZSPA-A-invocation}. It constructs an empty list $  L'$ and appends to it the indices (e.g., $j$) of the keys that do not pass the above check. 
\item simulates \zspaa and receives from $f^{\st \zspaa}$ the output that contains a vector of random polynomials, $\vv{\mu}'$, for each valid key. 
\item sends to  $\mathcal{A}$, the list $  L'$ and vector  $\vv{\mu}'$, for every bin. 
\item  for each bin of client $  {  C}$ whose index (or ID) is not in $  L'$ computes polynomial $ {\bm\chi}^{ \st {  {(D, C)}}}$ as follows: 
 $ {\bm\chi}^{ \st {  {(D, C)}}}= {\bm\zeta}\cdot  {\bm\eta}^{ \st {  {(D,C)}}}-( {\bm\gamma}^{ \st {  {(D,C)}}}+ {\bm\delta}^{ \st {  {(D,C)}}})$,   where $ {\bm\eta}^{ \st {  {(D,C)}}}$ is a fresh random polynomial of degree $3d+1$. Note, $C$ includes both honest and corrupt clients, except those clients whose index is in  $  L'$. $\mathsf{Sim}^{\st \fpsi}_{\st A}$ sends every polynomial $ {\bm\chi}^{ \st {  {(D, C)}}}$ to  $\mathcal{A}$. 
 \item given each $ {\bm\nu}^{ \st {  {(C'')}} }$ (by $\mathcal{A}$ in step \ref{F-PSI::sim-A-receive-nu-from-adv}), computes polynomial $ {\bm\phi'}^{ \st {  {(C'')}} }$ as follows: ${\bm\phi'}^{ \st {  {(C'')}} } =  {\bm\nu}^{ \st {  {(C'')}} }-  {\bm\gamma}^{ \st {  {(D,C'')}}} -  {\bm\delta}^{ \st {  {(D,C'')}}}$, for every bin. Then, $\mathsf{Sim}^{\st \fpsi}_{\st A}$ checks whether  $ {\bm\zeta}$  divides $ {\bm\phi'}^{ \st {  {(C'')}} }$, for every bin. It appends the index of those clients that did not pass the above check to a new list, $  L''$. Note that $  L'\cap   L''=\bot$.
\item if  $  L'$ or $  L''$ is not empty, then it instructs the ledger: (i) to refund the coins of those parties whose index is not in $  L'$ and $  L''$, (ii) to retrieve $\chc$ amount from the adversary (i.e., one of the parties whose index is in one of the lists) and send the $\chc$ amount to the auditor, and (iii) to compensate each honest party (whose index is not in the two lists)  $\frac{m'\cdot (\yc+\chc)-\chc}{m-m'}$ amount, where $m'=|  L'|+|  L''|$.  Then, it sends  message $abort_{\st 3}$ to TTP. 
\item outputs whatever $\mathcal{A}$ outputs and halts. 
 \end{enumerate}
\end{enumerate}

Next, we show that the real and ideal models are computationally indistinguishable. We first focus on the adversary’s output. In the real and ideal models, the adversary sees the transcripts of ideal calls to $f_{\st \ct}$ as well as this functionality outputs, i.e., $mk$. Due to the security of \ct (as we are in the $f_{\st \ct}$-hybrid world), the transcripts of $f_{\st \ct}$ in both models have identical distribution, so have the random output of $f_{\st \ct}$, i.e., $mk$. The same holds for (the transcripts and) outputs (i.e., $(  k,   g,   q)$) of $f^{\st \zspaa}$ that the adversary observes in the two models. Also, the deposit amount is identical in both models. Thus, in the case where $abort_{\st 1}$ is disseminated at this point; the adversary's output distribution in both models is identical.

The adversary also observes (the transcripts and) outputs of ideal calls to $f^{\st \vopr}$ in both models, i.e., output ($\bm\theta^{ \st {  {(C'')}}}_{\st 1}=\bm\zeta \cdot \bm\omega^{ \st {  {(D, C'')}}}\cdot \bm\omega^{ \st {  {(C'', D)}}}\cdot \bm\pi^{ \st {  {(C'')}}}+\bm\gamma^{ \st {  {(D,C'')}}}, \bm\theta^{ \st {  {(C'')}}}_{\st 2}=\bm\zeta \cdot \bm\rho^{ \st {  {(D, C'')}}}\cdot \bm\rho^{ \st {  {(C'',D)}}}\cdot \bm\pi^{ \st {  {(D)}}}+\bm\delta^{ \st {  {(D, C'')}}}$) for each corrupted client $C''$. However, due to the security of \vopr, the $\mathcal{A}$'s view, regarding \vopr, in both models have identical distribution.  
In the real model, the adversary observes the polynomial ${\bm\nu}^{ \st {  {(C)}}}$ that each honest client $C$ stores in the smart contract. Nevertheless, this is a blinded polynomial comprising of two uniformly random blinding polynomials (i.e., $\bm\gamma^{ \st {  {(D,C)}}}$ and $\bm\delta^{ \st {  {(D,C)}}}$) unknown to the adversary. In the ideal model, $\mathcal{A}$ is given polynomial  $ {\bm\nu}^{ \st {  {(C')}}}$ for each honest client $C'$. This polynomial has also been blinded via two uniformly random polynomials (i.e., $ {\bm\gamma}^{ \st {  {(D, C')}}}$ and $ {\bm\delta}^{ \st {  {(D, C')}}}$) unknown to $\mathcal{A}$. Thus, ${\bm\nu}^{ \st {  {(C)}}}$ in the real model and $ {\bm\nu}^{ \st {  {(C)}}}$ in the ideal model have identical distributions. As a result, in the case where $abort_{\st 2}$ is disseminated at this point; the adversary's output distribution in both models is identical.

Furthermore, in the real world, the adversary observes polynomials ${\bm\zeta}$ and  $\bm\nu^{ \st {  {(D)}}}$  that  $D$ stores in the smart contract. Nevertheless, ${\bm\zeta}$ is a uniformly random polynomial, also polynomial $\bm\nu^{ \st {  {(D)}}}$ has been blinded; its blinding factors are the additive inverse of  the sum of the random polynomials $\bm\gamma^{ \st {  {(D,C)}}}$ and $\bm\delta^{ \st {  {(D,C)}}}$ unknown to the adversary, for every client $C\in \{  {  A}_{ \st {   1}}, ...,   {  A}_{ \st {   m}}\}$ and  $D$. In the ideal model, $\mathcal{A}$ is given $ {\bm\zeta}$ and $ {\bm\nu}^{ \st {  {(D)}}}$, where the former is a random polynomial and the latter is a blinded polynomial that has been blinded with the additive inverse of the sum of random polynomials $ {\bm\gamma}^{ \st {  {(D,C)}}}$ and $ {\bm\delta}^{ \st {  {(D,C)}}}$ unknown to it,  for all client $C$. Therefore,  (${\bm\zeta}, \bm\nu^{ \st {  {(D)}}}$)  in the real model and  ($ {\bm\zeta},  {\bm\nu}^{ \st {  {(D)}}}$) in the ideal model component-wise have identical distribution.

Also, the sum of less than $m+1$ blinded polynomials ${\bm\nu}^{\st (A_{\st 1})},...,{\bm\nu}^{\st (A_{\st m})}, \bm\nu^{ \st {  {(D)}}}$   in the real model has identical distribution to the sum of less than $m+1$ blinded polynomials $ {\bm\nu}^{\st (A_{\st 1})},...,  {\bm\nu}^{\st (A_{\st m})},  {\bm\nu}^{ \st {  {(D)}}}$ in the ideal model, as such a combination would still be blinded by a set of random blinding polynomials unknown to the adversary. Now we discuss why the two polynomials $\frac{\bm\phi}{\bm\zeta}- \bm\gamma'$ in the real model and $\frac{ {\bm\phi}} { {\bm\zeta}}-  {\bm\gamma'}$ in the ideal model are indistinguishable. Note that we divide and then subtract  polynomials ${\bm\phi}$ because the adversary already knows (and must know) polynomials $(\bm\zeta, \bm\gamma')$. In the real model, polynomial $\frac{\bm\phi}{\bm\zeta}- \bm\gamma'$ has the following form: 
\begin{equation}\label{equ::-corrupt-A-real-world-phi}
 \frac{\bm\phi}{\bm\zeta}- \bm\gamma'=  \bm\omega'^{ \st {  {(D)}}}\cdot \bm\pi^{ \st {  {(D)}} } +\sum\limits^{ \st {   A}_{ \st {   m}}}_{ \st {  {C }= }  \st {   A}_{ \st {  1}}}(\bm\omega^{ \st {  {(D,C)}}} \cdot \bm\omega^{ \st {  {(C,D)}}}\cdot \bm\pi^{ \st {  (C)}}) +\bm\pi^{ \st {  {(D)}}}\cdot\sum\limits^{ \st {   A}_{ \st {   m}}}_{ \st {  {C }= }  \st {   A}_{ \st {  1}}}(\bm\rho^{ \st {  {(D,C)}}} \cdot \bm\rho^{ \st {  {(C,D)}}})=\bm\mu\cdot gcd( \bm\pi^{ \st {  {(D)}} },\bm\pi^{\st (A_1)}, ..., \bm\pi^{\st (A_m)})
\end{equation}
In Equation \ref{equ::-corrupt-A-real-world-phi}, every element of   $[\bm\omega'^{ \st {  {(D)}}},..., \bm\omega^{ \st {  {(D,C)}}}, \bm\rho^{ \st {  {(D,C)}}}]$ is a uniformly random polynomial for every  client $C\in \{  {  A}_{ \st {   1}}, ...,   {  A}_{ \st {   m}}\}$  (including corrupt ones) and client $D$; because it has been picked by (in this case honest) client $D$. Thus,  as shown in Section \ref{sec::poly-rep}, $\frac{\bm\phi}{\bm\zeta}- \bm\gamma'$ has the form $\bm\mu\cdot gcd( \bm\pi^{ \st {  {(D)}} },\bm\pi^{\st (A_1)}, ..., \bm\pi^{\st (A_m)})$, where $\bm\mu$ is a uniformly random polynomial and $gcd( \bm\pi^{ \st {  {(D)}} },\bm\pi^{\st (A_1)}, ..., \bm\pi^{\st (A_m)})$ represents the intersection of the input sets. 

In the ideal model, $\mathcal{A}$ can construct polynomial $ {\bm\phi}$ using its (well-formed) inputs $ {\bm\nu}^{ \st {  {(C'')}} }$ and polynomials $ {\bm\nu}^{ \st {  {(C')}} }$ that the simulator has sent to it, for all $C'\in\hat P$ and all $C''\in P'$. Thus, in the ideal model, polynomial $\frac{ {\bm\phi}}{\bm\zeta}-  {\bm\gamma'}$ has the following form:
\begin{equation}\label{equ::-corrupt-A-ideal-world-phi}
\begin{split}
\frac{ {\bm\phi}}{\bm\zeta}-  {\bm\gamma'}  = &\  { \bm\pi}\cdot\Big(\sum\limits_{\st \forall C'\in \hat P}(\bm\omega^{ \st {  {(D, C')}}} \cdot \bm\omega^{ \st {  {(C', D)}}}) + \sum\limits_{\st \forall C'\in \hat P}(\bm\rho^{ \st {  {(D, C')}}} \cdot \bm\rho^{ \st {  {(C', D)}}})\Big)+\\+ & \Big(\sum\limits_{\st \forall C''\in P'}(\bm\omega^{ \st {  {(D,C'')}}} \cdot \bm\omega^{ \st {  {(C'',D)}}}\cdot \bm\pi^{ \st {  (C'')}}) +\bm\pi^{ \st {  {(D)}}}\cdot\sum\limits_{\st \forall C''\in P'} (\bm\rho^{ \st {  {(D, C'')}}} \cdot \bm\rho^{ \st {  {(C'',D)}}})\Big) \\  =  &\ {\bm\mu}\cdot gcd(  {\bm\pi}, \bm\pi^{ \st {  {(D)}}}, \bm\pi^{ \st {  {(C'')}}})
 \end{split}
\end{equation}
In Equation \ref{equ::-corrupt-A-ideal-world-phi}, every element of the vector $[\bm\omega^{ \st {  {(D,C')}}}, \bm\omega^{ \st {  {(D,C'')}}}, \bm\rho^{ \st {  {(D,C')}}}, \bm\rho^{ \st {  {(D,C'')}}}]$ is a uniformly random polynomial for all $C'\in\hat P$ and all $C''\in P'$, as they have been picked by $\mathsf{Sim}^{\st \fpsi}_{\st A}$. Therefore, $\frac{ {\bm\phi}}{\bm\zeta}-  {\bm\gamma'}$ equals $ {\bm\mu}\cdot gcd(  {\bm\pi}, \bm\pi^{ \st {  {(D)}}}, \bm\pi^{ \st {  {(C'')}}})$, such that $ {\bm\mu}$ is a uniformly random polynomial and $gcd(  {\bm\pi}, \bm\pi^{ \st {  {(D)}}}, \bm\pi^{ \st {  {(C'')}}})$ represents the intersection of the input sets. We know that $gcd( \bm\pi^{ \st {  {(D)}} },\bm\pi^{\st (A_1)}, ..., \bm\pi^{\st (A_m)})=gcd(  {\bm\pi}, \bm\pi^{ \st {  {(D)}}}, \bm\pi^{ \st {  {(C'')}}})$, as $ {\bm\pi}$ includes the intersection of all clients' sets. Also, $ {\bm\mu}$ has identical distribution in the two models, because they are uniformly random polynomials. Thus,  $\frac{\bm\phi}{\bm\zeta}- \bm\gamma'$ in the real model and $\frac{ {\bm\phi}} { {\bm\zeta}}-  {\bm\gamma'}$ in the ideal model are indistinguishable.

Now we focus on the case where $Flag=False$.  In the real model, the adversary observes the output of $\mathtt{Audit}(.)$ which is a list of indices $  L$ and a vector of random polynomials $\vv\mu$ picked by an honest auditor.  In the ideal model, $\mathcal{A}$ is given a  list $  L'$ of indices and a vector of random polynomials $\vv{\mu}'$ picked by the simulator. Thus, the pair ($  L, \vv\mu$) in the real model has identical distribution to the pair ($  L', \vv\mu'$) in the ideal model. Moreover, in the real model, the adversary observes each polynomial $\bm\chi^{ \st {  {(D, C)}}}=\bm\zeta\cdot \bm\eta^{ \st {  {(D,C)}}}-(\bm\gamma^{ \st {  {(D,C)}}}+\bm\delta^{ \st {  {(D,C)}}})$ that $D$ stores in the contract, for each bin and each client $C$ whose index is not in $  L$. This is a blinded polynomial with blinding factor $\bm\eta^{ \st {  {(D,C)}}}$ which itself is a uniformly random polynomial picked by $D$. In the ideal model, $\mathcal{A}$ is given a polynomial of the form $ {\bm\chi}^{ \st {  {(D, C)}}}= {\bm\zeta}\cdot  {\bm\eta}^{ \st {  {(D,C)}}}-( {\bm\gamma}^{ \st {  {(D,C)}}}+ {\bm\delta}^{ \st {  {(D,C)}}})$, for each bin and each client $C$ whose index is not in $  L'$. This is also a blinded polynomial whose blinding factor is  $ {\bm\eta}^{ \st {  {(D,C)}}}$ which itself is a random polynomial picked by the simulator.  Therefore, $\bm\chi^{ \st {  {(D, C)}}}$ in the real model has  identical distribution to $\bm\chi^{ \st {  {(D, C)}}}$ in the ideal model.  In the real model, the adversary  observes polynomial $ \bm\iota^{ \st {  {(C)}}}= \bm\zeta\cdot(\bm\eta^{ \st {  {(D,C)}}} + \bm\omega^{ \st {  {(D,C)}}}\cdot \bm\omega^{ \st {  {(C,D)}}}\cdot \bm\pi^{ \st {  {(C)}}}+\bm\rho^{ \st {  {(D,C)}}}\cdot \bm\rho^{ \st {  {(C,D)}}}\cdot \bm\pi^{ \st {  {(D)}}}+\bm\xi^{ \st {  {(C)}}})$ which is a blinded polynomial whose blinding factor is the sum of the above random polynomials, i.e., $\bm\eta^{ \st {  {(D,C)}}}+\bm\xi^{ \st {  {(C)}}}$. In the ideal model, $\mathcal{A}$ already has  polynomials $ {\bm\chi}^{ \st {  {(D, C)}}},  {\bm\nu}^{ \st {  {(C)}}}$, and $ {\bm\mu}^{ \st {  {(C)}}}$, where $ {\bm\mu}^{ \st {  {(C)}}}\in {\vv \mu'}$; this lets $\mathcal{A}$ compute
 $  {\bm\iota}^{ \st {  {(C)}}} =  {\bm\chi}^{ \st {  {(D, C)}}} +   {\bm\nu}^{ \st {  {(C)}}} +  {\bm\mu}^{ \st {  {(C)}}}= {\bm\zeta}\cdot ( {\bm\eta}^{ \st {  {(D,C)}}}+  {\bm\omega}^{ \st {  {(D,C')}}}\cdot  {\bm\omega}^{ \st {  {(C',D)}}}\cdot  {\bm\pi}+  {\bm\rho}^{ \st {  {(D,C')}}}\cdot  {\bm\rho}^{ \st {  {(C',D)}}}\cdot  {\bm\pi}+ {\bm\xi}^{ \st {  {(C)}}})$,
  where $ {\bm\xi}^{ \st {  {(C)}}}$ is a random blinding polynomial used in $ {\bm\mu}^{ \st {  {(C)}}}$. Nevertheless, $ {\bm\iota}^{ \st {  {(C)}}}$ itself is a blinded polynomial whose blinding factor is the sum of random polynomials, i.e., $ {\bm\eta}^{ \st {  {(D,C)}}}+ {\bm\xi}^{ \st {  {(C)}}}$. Hence, the distribution of polynomial $ \bm\iota^{ \st {  {(C)}}}$ in the real model and $ {\bm\iota}^{ \st {  {(C)}}}$ in the ideal model are identical. Moreover, the integer $\yc+\chc+\frac{m'\cdot (\yc+\chc)-\chc}{m-m'}$ has identical distribution  in both  models. 

Next, we show that the honest party aborts with the same probability in the real and ideal models. Due to the security of \ct, an honest party (during the execution of \ct) aborts with the same probability in both models; in this case, the adversary learns nothing about the parties' input set and the sets' intersection as the parties have not sent out any encoded input set yet.  The same holds for the probability that honest parties abort during the execution of \zspaa.  In this case, an aborting adversary also learns nothing about the parties' input set and the sets' intersection. Since all parties' deposit is public, an honest party can look up and detect if not all parties have deposited a sufficient amount with the same probability in both models. If parties halt because of insufficient deposit, no one learns about the parties' input set and the sets' intersection because the inputs (representation) have not been sent out at this point.

Due to the security of \vopr, honest parties abort with the same probability in both models. In the case of an abort during \vopr execution, the adversary would learn nothing (i) about its counter party' input set due to the security of \vopr, and (ii) about the rest of the honest parties' input sets and the intersection as the other parties' input sets are still blinded by random blinding factors known only to $D$. In the real model, $D$ can check if all parties provided their encoded inputs, by reading from the smart contract.  The simulator also can do the same check to ensure  $\mathcal{A}$ has provided the encoded inputs of all corrupt parties. Therefore, in both models, an honest party with the same probability would detect the violation of such a requirement, i.e., providing all encoded inputs. Even in this case, if an adversary aborts (by not proving its encoded inputs), then it learns nothing about the honest parties' input sets and the intersection for the reason explained above. 

%Now we determine the probability that $Flag=True$ in each model. 

In the real model, the smart contract sums every client $C$'s polynomial $\bm\nu^{ \st {  {(C)}} }$ with each other and with $D$'s polynomial $\bm\nu^{ \st {  {(D)}} }$, which removes the blinding factors that $D$ initially inserted (during the execution of \vopr), and then checks whether the result is divisible by  $\bm \zeta$. Due to (a) Theorem \ref{Unforgeable-Polynomials-Linear-Combination} (i.e., unforgeable polynomials’ linear combination), (b) the fact that the smart contract is given the random polynomial $\bm \zeta$ in plaintext, (c) no party (except honest client $D$) knew anything about $\bm \zeta$ before they send their input to the contract, and (d) the security of the contract (i.e., the adversary cannot influence the correctness of the above verification performed by the contract), the contract can detect if a set of outputs of \vopr has been tampered with, with a probability at least $1-\negl(\lambda)$. In the ideal model, $\mathsf{Sim}^{\st \fpsi}_{\st A}$ also can remove the blinding factors and it knows the random polynomial $ {\bm \zeta}$, unlike the adversary who does not know $ {\bm \zeta}$ when it sends the outputs of \vopr to $\mathsf{Sim}^{\st \fpsi}_{\st A}$. So, $\mathsf{Sim}^{\st \fpsi}_{\st A}$ can detect when $\mathcal{A}$ modifies a set of the outputs of \vopr that were sent to  $\mathsf{Sim}^{\st \fpsi}_{\st A}$ with a probability at least $1-\negl(\lambda)$,  due to Theorem \ref{Unforgeable-Polynomials-Linear-Combination}. Hence, the smart contract in the real model and the simulator in the ideal model would abort with a similar probability. 

Moreover, due to the security of \zspaa, the probability that an invalid key $k_{\st i}\in \vv{k}$ is added to the list $  L$ in the real world is similar to the probability that  $\mathsf{Sim}^{\st \fpsi}_{\st A}$ detects an invalid key $  k'_{\st i}\in \vv{k'}$ in the ideal world. In the real model, when $Flag=False$, the smart contract can identify each ill-structured output of \vopr (i.e.,  $\bm\nu^{ \st {  {(C)}} }$) with a probability of at least $1-\negl(\lambda)$ by checking whether $\bm\zeta$  divides $\bm\iota^{ \st {  {(C)}}}$, due to  (a) Theorem \ref{proof::unforgeable-poly} (i.e., unforgeable polynomial), (b) the fact that the smart contract is given $\bm \zeta$ in plaintext, (c) no party (except honest client $D$) knew anything about $\bm \zeta$ before they send their input to the contract, and (d) the security of the contract. In the ideal model, when $Flag=False$, given each $ {\bm\nu}^{ \st {  {(C'')}} }$, $\mathsf{Sim}^{\st \fpsi}_{\st A}$ can remove its blinding factors from  $ {\bm\nu}^{ \st {  {(C'')}} }$ which results in $ {\bm\phi'}^{ \st {  {(C'')}} }$ and then check if $ {\bm\zeta}$  divides $ {\bm\phi'}^{ \st {  {(C'')}} }$. The simulator can also detect an ill-structured  $ {\bm\nu}^{ \st {  {(C'')}} }$ with a probability of at least $1-\negl(\lambda)$, due to Theorem \ref{proof::unforgeable-poly}, the fact that the simulator is given $ {\bm \zeta}$ in plaintext,  and the adversary is not given any knowledge about $ {\bm \zeta}$ before it sends to the simulator the outputs of \vopr.  Hence, the smart contract in the real model and $\mathsf{Sim}^{\st \fpsi}_{\st A}$ in the ideal model would detect an ill-structured input of an adversary with the same probability. 

%$Q:=(\qinit, \qdel, \qUnFAbt, \qFAbt)$

Now, we analyse the output of the predicates $(\qinit, \qdel, \qUnFAbt, \qFAbt)$ in the real and ideal models. In the real model, all clients proceed to prepare their input set only if the predefined amount of coins has been deposited by the parties; otherwise, they will be refunded and the protocol halts. In the ideal model also the simulator proceeds to prepare its inputs only if a sufficient amount of deposit has been put in the contract; otherwise, it would send message $abort_{\st 1}$ to TTP. Thus, in both models, the parties proceed to prepare their inputs only if $\qinit(.) \rightarrow1$. 
In the real model, if there is an abort after the parties ensure there is enough deposit and before client $D$ provides its encoded input to the contract, then all parties would be able to retrieve their deposit in full; in this case, the aborting adversary would not be able to learn anything about honest parties input sets, because the parties' input sets are still blinded by random blinding polynomials known only to client $D$. In the ideal model, if there is any abort during steps \ref{F-PSI::sim-A-first-VOPR-invocation}--\ref{F-PSI::sim-A-receive-nu-from-adv}, then the simulator sends $abort_{\st 2}$ to TTP and instructs the ledger to refund the coins that every party deposited. Also, in the case of an abort (within the above two points of time), the auditor is not involved. Thus, in both models,  in the case of an abort within the above points of time, we would have $\qFAbt(.)\rightarrow1$. In the real model, if $Flag=True$, then all parties would be able to learn the intersection and the smart contract refunds all parties, i.e., sends each party $\yc+\chc$ amount which is the amount each party initially deposited. 

In the ideal model, when $Flag=True$, then $\mathsf{Sim}^{\st \fpsi}_{\st A}$ can extract the intersection (by summing the output of \vopr provided by all parties and removing the blinding polynomials) and sends back each party's deposit, i.e., $\yc+\chc$ amount. Hence, in both models in the case of $Flag=True$, when all of the parties receive the result, we would have $\qdel(.)\rightarrow 1$. In the real model, when $Flag=False$, only the adversary which might corrupt $m'$ clients would be able to learn the result; in this case, the contract sends (i) $\chc$ amount to the auditor, and (ii) $\frac{m'\cdot (\yc+\chc)-\chc}{m-m'}$ amount as a compensation, to each honest party, in addition to each party's deposit $\yc+\chc$. In the ideal model,  when $Flag=False$, $\mathsf{Sim}^{\st \fpsi}_{\st A}$  sends $abort_{\st 3}$ to TTP and instructs the ledger to distribute the same amount among the auditor (e.g., with address $adr_{\st j}$) and every honest party (e.g., with address $adr_{\st i}$) as the contract does in the real model. Thus, in both models when $Flag=False$, we would have $\qUnFAbt(., ., ., ., adr_{\st i})\rightarrow (a=1, .)$ and  $\qUnFAbt(., ., ., ., adr_{\st j})\rightarrow (., b=1)$. 

We conclude that the distributions of the joint outputs of the honest client $C\in \hat P$, client $D$, \aud, and the adversary in the real and ideal models are computationally indistinguishable.

\

\noindent\textbf{Case 2: Corrupt dealer $D$}.  In the real execution, the dealer's view is defined as follows:

$$ \mathsf{View}_{\st D}^{\st \fpsi} \Big(S^{\st (D)},(S^{\st (1)},..., S^{\st (m)})\Big)=$$ $$ \{S^{\st (D)}, adr_{\st sc}, m\cdot(\yc+\chc), r_{\st D}, \mathsf{View}^{\st \ct}_{\st D}, k, g, q, \mathsf{View}^{\st \vopr}_{\st D}, \bm\nu^{\st (A_1)},..., \bm\nu^{\st (A_m)}, S_{\st \cap}\}$$
where  $\mathsf{View}^{\st \ct}_{\st D}$ and $\mathsf{View}^{\st \vopr}_{\st D}$ refer to the dealer's real-model view during the execution of \ct and \vopr respectively. Also, $r_{\st D}$ is the outcome of internal random coins of client $D$ and $adr_{\st sc}$ is the address of contract $\mathcal{SC}_{\fpsi}$. The simulator $\mathsf{Sim}^{\st \fpsi}_{\st D}$, which receives all parties' input sets, works as follows. 

\begin{enumerate}

\item receives from the subroutine adversary polynomials $ {\bm\zeta}, ( {\bm\gamma}^{\st(A_1)},  {\bm\delta}^{\st (A_1)}),..., ( {\bm\gamma}^{\st(A_m)},  {\bm\delta}^{\st (A_m)})$,  $( {\bm\omega}'^{\st (A_1)},  {\bm\rho}'^{\st (A_1)}),..., $ $( {\bm\omega}'^{\st (A_m)},  {\bm\rho}'^{\st (A_m)})$, where $deg( {\bm\gamma}^{\st(C)})=deg( {\bm\delta}^{\st(C)})=3d+1, deg( {\bm\omega}'^{\st (C)}) =deg( {\bm\rho}'^{\st (C)})= d$, and $deg( {\bm\zeta})=1$, where $C\in  \{  {  A}_{ \st {   1}}, ...,   {  A}_{ \st {   m}}\} $.
\item generates an empty view. It appends to the view, the input set $S^{\st (D)}$. It constructs and deploys a smart contract. Let $  {adr}_{\st sc}$ be the contract's address. It appends $ {adr}_{\st sc}$ to the view.

\item appends to the view integer $m\cdot(\yc+\chc)$   and coins $r'_{\st D}$ chosen uniformly at random. 
\item extracts the simulation of \ct from \ct's simulator for client $D$. Let $\mathsf{Sim}^{\st \ct}_{\st D}$ be the simulation. It appends $\mathsf{Sim}^{\st \ct}_{\st D}$ to the view. 
\item picks a random key, $  k'$, and derives pseudorandom values $z'_{\st i,j}$ from the key (the same way is done in Figure \ref{fig:ZSPA}). It constructs a Merkle tree on top of all values $z'_{\st i, j}$. Let $g'$ be the root of the resulting tree. It appends $  k', g'$, and $q'=\mathtt{H}(  k')$ to the view. 
\item invokes \vopr's functionality twice and extracts the simulation of \vopr from \vopr's simulator for client $D$. Let $\mathsf{Sim}^{\st \vopr}_{\st D}$ be the simulation. It appends $\mathsf{Sim}^{\st \vopr}_{\st D}$ to the view.  
\item given the parties' input sets, computes a polynomial $ {\bm\pi}$ that represents the intersection of the sets. 

\item picks $m$ random polynomials $ {\bm\tau}^{\st (A_1)}, ...,  {\bm\tau}^{\st (A_m)}$ of degree $3d+1$ such that their sum is $0$.  

%\item picks $m$ pair of random polynomials $(\bm\gamma^{\st(1)}, \bm\delta^{\st (1)}),..., (\bm\gamma^{\st(m)}, \bm\delta^{\st (m)})$ where each polynomial is of degree $3d+1$.

\item picks $m$ pairs of random polynomials $( {\bm\omega}^{\st (A_1)},  {\bm\rho}^{\st (A_1)}),..., ( {\bm\omega}^{\st (A_m)},  {\bm\rho}^{\st (A_m)})$, where each polynomial is of degree $d$. Then, $\mathsf{Sim}^{\st \fpsi}_{\st D}$ for each client $C\in  \{  {  A}_{ \st {   1}}, ...,   {  A}_{ \st {   m}}\} $ computes polynomial $ {\bm\nu}^{\st (C)}= {\bm\zeta}\cdot {\bm\pi}\cdot(  {\bm\omega}^{\st (C)}\cdot  {\bm\omega}'^{\st (C)}+  {\bm\rho}^{\st (C)}\cdot  {\bm\rho}'^{\st (C)})+ {\bm\delta}^{\st (C)}+ {\bm\gamma}^{\st(C)}+  {\bm\tau}^{ \st {  {(C)}}}$. 
\item appends $ {\bm\nu}^{\st (A_1)},...,  {\bm\nu}^{\st (A_m)}$ and the intersection of the sets $S'_{\st \cap}$ to the view. 
\end{enumerate}

 Next, we will show that the two views are computationally indistinguishable. $D$'s input $S^{\st (D)}$ is
 identical in both models; therefore, they have identical distributions. Also, the contract's address has the same distribution in both views, and so has the integer $ m\cdot(\yc+\chc)$. Since the real-model semi-honest adversary samples its randomness according to the protocol’s description, the random coins in both models (i.e., $r_{\st D}$  and $r'_{\st D}$) have identical distributions. Moreover, due to the security of  \ct, $\mathsf{View}^{\st \ct}_{\st D}$ and $\mathsf{Sim}^{\st \ct}_{\st D}$ have identical distributions. Keys $k$ and $  k'$ have identical distributions, as both have been picked uniformly at random from the same domain.  In the real model, each element of the pair $(g, p)$ is the output of a deterministic function on a random value $k$. We know that $k$ in the real model has identical distribution to $  k'$ in the ideal model, so do the evaluations of deterministic functions (i.e., Merkle tree, $\mathtt{H}$, and $\mathtt {PRF}$) on them. Therefore, each pair $(g, q)$ in the real model component-wise has an identical distribution to each pair $(g', q')$ in the ideal model.  
 Furthermore, due to the security of the \vopr, $\mathsf{View}^{\st \vopr}_{\st D}$ and $\mathsf{Sim}^{\st \vopr}_{\st D}$ have identical distributions.

 In the real model, each $\bm\nu^{\st (C)}$ has been blinded by a pseudorandom polynomial (i.e., derived from $\mathtt{PRF}$'s output) unknown to client $D$. In the ideal model, however, each $ {\bm\nu}^{\st (C)}$ has been blinded by a random polynomial unknown to client $D$. Due to the security of $\mathtt{PRF}$,  its outputs are computationally indistinguishable from truly random values.  Therefore, ${\bm\nu}^{\st (C)}$ in the real model and $ {\bm\nu}^{\st (C)}$ in the ideal model are computationally indistinguishable. Now we focus on the sum of all ${\bm\nu}^{\st (C)}$ in the real model and the sum of all $ {\bm\nu}^{\st (C)}$ in the ideal model, as adding them together would remove the above blinding polynomials that are unknown to client $D$.  Specifically, in the real model, after client $D$ sums all  ${\bm\nu}^{\st (C)}$ and removes the blinding factors and $\bm\zeta$ that it initially imposed, it would get a polynomial of the  form $\bm{\hat\phi}= \frac{ \sum\limits^{ \st {   A}_{ \st {   m}}}_{  \st {  {C }= }  \st {   A}_{ \st {  1}}}\bm\nu^{ \st {  {(C)}} }-\sum\limits^{ \st {   A}_{ \st {   m}}}_{  \st {  {C }= }  \st {   A}_{ \st {  1}}}(\bm\gamma^{ \st {  {(D,C)}}} + \bm\delta^{ \st {  {(D,C)}}})}{\bm\zeta}= \sum\limits^{ \st {   A}_{ \st {   m}}}_{ \st {  {C }= }  \st {   A}_{ \st {  1}}}(\bm\omega^{ \st {  {(D,C)}}} \cdot \bm\omega^{ \st {  {(C,D)}}}\cdot \bm\pi^{ \st {  (C)}}) +\bm\pi^{ \st {  {(D)}}}\cdot\sum\limits^{ \st {   A}_{ \st {   m}}}_{ \st {  {C }= }  \st {   A}_{ \st {  1}}}(\bm\rho^{ \st {  {(D,C)}}} \cdot \bm\rho^{ \st {  {(C,D)}}}) $, where $\bm\omega^{ \st {  {(C,D)}}}$ and $\bm\rho^{ \st {  {(C,D)}}}$ are random polynomials unknown to client $D$. 
 In the ideal model, after summing all $ {\bm\nu}^{\st (C)}$ and removing the random polynomials that it already knows, it would get a polynomial of the following form: 
 $\bm{\hat\phi}'= \frac{\sum\limits^{ \st {   A}_{ \st {   m}}}_{  \st {  {C }= }  \st {   A}_{ \st {  1}}} {\bm\nu}^{\st(C)}-\sum\limits^{ \st {   A}_{ \st {   m}}}_{  \st {  {C }= }  \st {   A}_{ \st {  1}}}( {\bm\gamma}^{\st(C)} +  {\bm\delta}^{\st(C)})}{ {\bm\zeta}}=  {\bm\pi}\cdot\sum\limits^{ \st {   A}_{ \st {   m}}}_{ \st {  {C }= }  \st {   A}_{ \st {  1}}}( {\bm\omega}'^{\st(C)} \cdot  {\bm\omega}^{\st(C)})+( {\bm\rho}'^{\st(C)} \cdot  {\bm\rho}^{\st(C)})  $.

 As shown in Section \ref{sec::poly-rep},  polynomial $\bm{\hat\phi}$ has the form $\bm\mu\cdot gcd( \bm\pi^{ \st {  {(D)}} },\bm\pi^{\st (A_1)}, ..., \bm\pi^{\st (A_m)})$, where $\bm\mu$ is a uniformly random polynomial and $gcd( \bm\pi^{ \st {  {(D)}} },\bm\pi^{\st (A_1)}, ..., \bm\pi^{\st (A_m)})$ represents the intersection of the input sets. Moreover, it is evident that $\bm{\hat\phi}'$ has the form $ {\bm\mu} \cdot  {\bm\pi}$, where $ {\bm\mu} $ is a random polynomial and $ {\bm\pi}$ represents the intersection. We know that both $gcd( \bm\pi^{ \st {  {(D)}} },\bm\pi^{\st (A_1)}, ..., \bm\pi^{\st (A_m)})$ and $ {\bm\pi}$  represent the same intersection, also  ${\bm\mu} $ in the real model and $ {\bm\mu} $ in the ideal model have identical distribution as they are uniformly random polynomials. Thus, two polynomials $\bm{\hat\phi}$ and $\bm{\hat\phi}'$ are indistinguishable. Also, the output $S_{\st \cap}$ is identical in both views. We conclude that the two views are computationally indistinguishable.

 \

\noindent\textbf{Case 3: Corrupt auditor}.  In this case, by using the proof that we have already provided for Case 1 (i.e.,  $m-1$ client $A_{\st j}$s are corrupt), we can easily construct a simulator that generates a view computationally distinguishable from the real-model semi-honest auditor. 
 The reason is that, in the worst-case scenario where $m-1$ malicious client $A_{\st j}$s reveal their input sets and randomness to the auditor, the auditor's view would be similar to the view of these corrupt clients, which we have shown to be indistinguishable. The only extra messages the auditor generates, that a corrupt client $A_{\st j}$ would not see in plaintext, are random blinding polynomials $(\bm\xi^{\st (A_1)},..., \bm\xi^{\st (A_m)})$ generated during the execution of $\mathtt{Audit}(.)$ of \zspaa; however, these polynomials are picked uniformly at random and independent of the parties' input sets. Thus, if the smart contract detects misbehaviour and invokes the auditor, even if $m-1$ corrupt client $A_{\st j}$ reveals their input sets, then the auditor cannot learn anything about honest parties' input sets.    
 
 \

\noindent\textbf{Case 4: Corrupt public}. In the real model, the view of the public (i.e., non-participants of the protocol) is defined as below:

 $$ \mathsf{View}_{\st Pub}^{\st \fpsi} \Big(\bot, S^{\st (D)},(S^{\st (A_1)},..., S^{\st (A_m)})\Big)=$$ $$ \{\bot, adr_{\st sc}, (m+1)\cdot(\yc+\chc), k, g, q, \bm\nu^{\st (A_1)},..., \bm\nu^{\st (A_m)}, \bm\nu^{\st (D)}\}$$

 Now, we describe how the simulator $\mathsf{Sim}^{\st \fpsi}_{\st Pub}$  works. 
 
 \begin{enumerate}
 \item generates an empty view and appends to it an empty symbol, $\bot$. It constructs and deploys a smart contract. It appends the contract's address, $  {adr}_{\st sc}$ and integer $(m+1)\cdot(\yc+\chc)$ to the view.
\item picks a random key, $  k'$, and derives pseudorandom values $z'_{\st i,j}$ from the key,  in the same way, done in Figure \ref{fig:ZSPA}. It constructs a Merkle tree on top of the $z'_{\st i, j}$ values. Let $g'$ be the root of the resulting tree. It appends $  k', g'$, and $q' = \mathtt{H}(  k')$ to the view. 
\item for each client $C\in \{  {  A}_{ \st {   1}}, ...,   {  A}_{ \st {   m}}\}$ and client $D$ generates a random polynomial of degree $3d+1$ (for each bin), i.e.,  $ {\bm\nu}^{\st (A_1)}, ...,  {\bm\nu}^{\st (A_m)},  {\bm\nu}^{\st (D)}$. 
 \end{enumerate}
 
 Next, we will show that the two views are computationally indistinguishable. In both views, $\bot$ is identical. Also, the contract’s addresses (i.e., ${adr}_{\st sc}$) has the same
distribution in both views, and so has the integer $(m+1)\cdot(\yc+\chc)$. Keys $k$ and $  k'$ have identical distributions as well, because both of them have been picked uniformly at random from the same domain.  In the real model, each element of pair $(g, p)$ is the output of a deterministic function on the random key $k$. We know that $k$ in the real model has identical distribution to $  k'$ in the ideal model, and so do the evaluations of deterministic functions on them. Hence, each pair $(g, q)$ in the real model component-wise has an identical distribution to each pair $(g', q')$ in the ideal model. In the real model, each polynomial  $\bm\nu^{\st (C)}$ is a blinded polynomial comprising of two uniformly random blinding polynomials (i.e., $\bm\gamma^{ \st {  {(D,C)}}}$ and $\bm\delta^{ \st {  {(D,C)}}}$) unknown to the adversary. In the ideal model, each polynomial  $ {\bm\nu}^{\st (C)}$ is a random polynomial; thus, polynomials $\bm\nu^{\st (A_1)},..., \bm\nu^{\st (A_m)}$ in the real model have identical distribution to  polynomials $ {\bm\nu}^{\st (A_1)}, ...,  {\bm\nu}^{\st (A_m)}$ in the ideal model.  Similarly,  polynomial $\bm\nu^{\st (D)}$ has been blinded in the real model; its blinding factors are the additive inverse of  the sum of the random polynomials $\bm\gamma^{ \st {  {(D,C)}}}$ and $\bm\delta^{ \st {  {(D,C)}}}$ unknown to the adversary. In the ideal model, polynomial $ {\bm\nu}^{\st (D)}$ is a uniformly random polynomial; thus, ${\bm\nu}^{\st (D)}$ in the real model and $ {\bm\nu}^{\st (D)}$ in the ideal model have identical distributions.
Moreover, in the real model even though the sum $\bm\phi$ of polynomials  $\bm\nu^{\st (A_1)},..., \bm\nu^{\st (A_m)}, \bm\nu^{\st (D)}$ would remove some of the blinding random polynomials, it is still a blinded polynomial with a pseudorandom blinding factor $\bm\gamma'$ (derived from the output of $\mathtt{PRF}$), unknown to the adversary. In the ideal model, the sum of polynomials $ {\bm\nu}^{\st (A_1)}, ...,  {\bm\nu}^{\st (A_m)},  {\bm\nu}^{\st (D)}$ is also a random polynomial. Thus, the sum of the above polynomials in the real model is computationally indistinguishable from the sum of those polynomials in the ideal model. We conclude that the two views are computationally indistinguishable. 
  \hfill\(\Box\)\end{proof}

% !TEX root =main.tex

%\section{Sub Protocols}

\section{Definition of Multi-party PSI with Fair Compensation and Reward}

In this section, we upgrade \p to ``multi-party PSI with Fair Compensation and Reward'' (\ep), which (in addition to offering the features of \p) allows honest clients who contribute their set to receive a reward by a buyer who initiates the PSI computation and is interested in the result.

%In this section, we present the notion of ``Earn while You Reveal PSI'' (\ep), which allows honest clients who contribute their set to get paid by a buyer who initiates the PSI computation and is interested in the result. 

%In this section, we present an efficient  PSI that allows honest parties who contribute their set to get paid by a buyer who initiates the PSI computation and is interested in the result. 

%\subsection{The Model}

%In this section, we provide the security model of our smart-PSI protocol. There are two kind of parties involved in the protocol. Namely, (1) a set of clients $\{A_{\st 1},...,A_{\st m}\}$ potentially \emph{rational} (i.e. an adversary that picks the best strategy to maximise its profit) and all may collude with each other, and (2) a non- colluding dealer: client $D$, potentially semi-honest (i.e. a passive adversary).  Similar to F-PSI, we consider static adversary, we assume there is an authenticated private (off-chain) channel between the clients and we consider a standard public blockchain.
%  

In \ep, there are (1) a set of clients $\{A_{\st 1},...,A_{\st m}\}$ a subset of which is potentially active adversaries and may collude with each other, (2) a non-colluding dealer, $D$, potentially semi-honest, and (3) an auditor $Aud$ potentially semi-honest, where all clients (except \aud) have input set. Furthermore,  in \ep  there are two ``extractor'' clients, say $A_{\st 1}$ and $A_{\st 2}$, where $(A_{\st 1},A_{\st 2})\in \{A_{\st 1},...,A_{\st m}\}$. These extractor clients volunteer to extract the (encoded) elements of the intersection and send them to a public bulletin board, i.e., a smart contract. In return, they will be paid. 
%
%We assume these two extractors are corrupted by an active adversary during interacting with other parties (and clients) until they collaborate with the rest of the clients to compute the intersection; 
%
We assume these two extractors act rationally only when they want to carry out the paid task of extracting the intersection and reporting it to the smart contract, so they can maximise their profit.\footnote{Thus, similar to any $A_{\st i}$ in \p, these extractors might be corrupted by an active adversary during the PSI computation.} For simplicity, we let client $A_{\st m}$ be the buyer, i.e., the party which initiates the PSI computation and is interested in the result.

 The formal definition of \ep is built upon the definition of \p (presented in Section \ref{sec::F-PSI-model}); nevertheless, in \ep, we ensure that honest non-buyer clients receive a \emph{reward} for participating in the protocol and revealing a portion of their inputs deduced from the result. We:  (i)  upgrade the predicate \qdel to  \qdelwr to ensure that when honest clients receive the result, then an honest non-buyer client receives its deposit back plus a reward and a buyer client receives its deposit back minus the paid reward, and (ii) upgrade the predicate  \qUnFAbt to \qUnFAbtwr to ensure when an adversary aborts in an unfair manner (i.e., aborts but learns the result) then an honest party receives its deposit back plus a predefined amount of compensation plus a reward.  The other two predicates (i.e., \qinit and \qFAbt) remain unchanged. Given the above changes, we denote the four predicates as $\bar Q:=(\qinit,  \qdelwr, \qUnFAbtwr, \qFAbt)$. Below, we present the formal definition of predicates \qdelwr and \qUnFAbtwr.

    \begin{definition}  [\qdelwr:
    Delivery-with-Reward predicate] Let $\mathcal{G}$ be a stable ledger, $adr_{\st sc}$ be smart contract $sc$'s address, $adr_{\st i}\in Adr$ be the address of an honest party, $\xc$ be a fixed amount of coins, and $pram:=(\mathcal{G}, adr_{\st sc}, \xc)$. Let $R$ be a reward function that takes as input the computation result: $res$, a party's address: $adr_{\st i}$, a reward a party should receive for each unit of revealed information:  $\lc$, and input size: $inSize$.  Then $R$ is defined as follows, if $adr_{\st i}$ belongs to a non-buyer, then it returns the total amount that $adr_{\st i}$ should be rewarded and if $adr_{\st i}$ belongs to a buyer client, then it returns the reward's leftover that the buyer can collect, i.e., $R(res, adr_{\st i}, \lc, inSize)\rightarrow \rewci$.    Then, the delivery with reward predicate $\qdelwr(pram,  adr_{\st i}, res, \lc, inSize)$ returns $1$ if $adr_{\st i}$ has sent $\xc$ amount to $sc$ and received at least $\xc+\rewci$ amount from it. Else, it returns $0$.

    %
%    Let also $G$ be a compensation function that takes as input  two parameters $(deps, m')$, where $deps$ is the amount of coins  that all $m+1$ parties  deposit; it returns the amount of compensation each honest party must receive, i.e., $G(deps, m')\rightarrow c'$. 
    %

 %
  \end{definition}

   \begin{definition}  [\qUnFAbtwr: UnFair-Abort-with-Reward predicate]
 Let $pram:=(\mathcal{G}, adr_{\st sc}, \xc)$ be the parameters defined above, and $Adr'\subset Adr$ be a set containing honest parties' addresses, $m' = |Adr'|$,  and   $adr_{\st i}\in Adr'$. Let also $G$ be a compensation function that takes as input  three parameters $(\depsc, adr_{\st i}, m')$, where $\depsc$ is the amount of coins that all $m+1$ parties deposit, $adr_{\st i}$ is an honest party's address, and $m' = |Adr'|$; it returns the amount of compensation each honest party must receive, i.e., $G(\depsc, ard_{\st i}, m')\rightarrow \xci$. Let $R$ be the reward function defined above, i.e., $R(res, adr_{\st i}, \lc, inSize)\rightarrow \rewci$, and let $\hat {pram}:=(res, \lc, inSize)$.  Then, predicate \qUnFAbtwr is defined as $\qUnFAbtwr(pram, \hat {pram}, G, R, \depsc, m', adr_{\st i})\rightarrow (a,b)$, where $a=1$ if $adr_{\st i}$ is an honest party's address which has sent $\xc$ amount to $sc$ and received  $\xc+\xci+\rewci$  from it, and $b=1$ if $adr_{\st i}$ is an auditor's address which received $\xci$  from $sc$. Otherwise, $a=b=0$. 
  \end{definition}

Next, we present the formal definition of multi-party PSI with Fair Compensation and Reward, \ep.

%%%%%%%%

\begin{definition}[\ep]\label{def::PSI-Q-fair-reward}
Let $f^{\st \text{PSI}}$ be the multi-party PSI functionality defined in Section \ref{sec::F-PSI-model}. We say  protocol $\Gamma$ realises  $f^{\st \text{PSI}}$ with $\bar Q$-fairness-and-reward in the presence of $m-3$ static active-adversary clients $A_{\st j}$s and $two$ rational clients $A_{\st i}s$ or a static passive dealer  $D$ or passive auditor $Aud$, if for every non-uniform probabilistic polynomial time adversary $\mathcal{A}$ for the real model, there exists a non-uniform probabilistic polynomial-time adversary (or simulator) $\mathsf{Sim}$ for the ideal model, such that for every $I\in \{A_{\st 1},...,A_{\st m}, D, Aud\}$, it holds that: 

\begin{equation*}
\{\mathsf {Ideal}^{\st \mathcal{W}(f^{\st \text{PSI}}, \bar Q)}_{\st \mathsf{Sim}(z), I}(S_{\st 1},..., S_{\st m+1})\}_{\st S_{\st 1},..., S_{\st m+1},z}\stackrel{c}{\equiv} \{\mathsf{Real}_{\st \mathcal{A}(z), I}^{\st \Gamma}(S_{\st 1},..., S_{\st m+1}) \}_{\st S_{\st 1},..., S_{\st m+1},z}
\end{equation*}
where  $z$ is an auxiliary input given to $\mathcal{A}$ and  $\mathcal{W}(f^{\st \text{PSI}}, \bar Q)$ is a functionality that wraps $f^{\st \text{PSI}}$ with predicates $\bar Q:=(\qinit,  \qdelwr, \qUnFAbtwr, \qFAbt)$. 
  \end{definition}

%%%%%%%%

\section{\withRew: A Concrete Construction of \ep}

\subsection{Main Challenges to Overcome}

\subsubsection{Rewarding Clients Proportionate to the Intersection Cardinality.}
In PSIs, the main private information about the clients which is revealed to a result recipient is the private set elements that the clients have in common. Thus, honest clients must receive a reward proportionate to the intersection cardinality, from a buyer. To receive the reward, the clients need to reach a consensus on the intersection cardinality. The naive way to do that is to let every client find the intersection and declare it to the smart contract. Under the assumption that the majority of clients are honest, then the smart contract can reward the honest result recipient (from the buyer's deposit). Nevertheless, the honest majority assumption is strong in the context of multi-party PSI. Moreover, this approach requires all clients to extract the intersection, which would increase the overall costs.  Some clients may not even be interested in or available to do so. This task could also be conducted by a single entity, such as the dealer; but this approach would introduce a single point of failure and all clients have to depend on this entity.  
To address these challenges, we allow any two clients to become extractors.  Each of them finds and sends to the contract the (encrypted) elements in the intersection. It is paid by the contract if the contract concludes that it is honest. This allows us to avoid (i) the honest majority assumption, (ii) requiring all clients to find the intersection, and (iii) relying on a single trusted/semi-honest party to complete the task.

\subsubsection{Dealing with Extractors' Collusion.}
Using two extractors itself introduces another challenge; namely, they may collude with each other (and with the buyer) to provide a consistent but incorrect result, e.g., both may declare that only $s_{\st 1}$ is in the intersection while 
 the actual intersection contains $100$ set elements, including $s_{\st 1}$.  This behaviour will not be detected by a verifier unless the verifier always conducts the delegated task itself too, which would defeat the purpose of delegation. To efficiently address this issue, we use the counter-collusion smart contracts (outlined in Section \ref{Counter-Collusion-Smart-Contracts}) which creates distrust between the two extractors and incentivises them to act honestly. 

%\subsubsection{Preserving the Intersection Privacy.} As stated above, the extractors are required to prove to a smart contract that they know 

\subsection{Description of \withRew (\epsi)}

%An Overview of E-PSI}

%This section presents a protocol, called E-PSI, that realises \ep. 

\subsubsection{An Overview.} To construct  \epsi, we mainly use \fpsi, deterministic encryption, ``double-layered'' commitments, the hash-based padding technique (from Section \ref{sec::poly-rep}), and the counter-collusion smart contracts. % (described in Section \ref{Counter-Collusion-Smart-Contracts}). 
At a high level, \epsi works as follows. First, all clients run step \ref{gen-FPSI-cont} of \fpsi to agree on a set of parameters and \fpsi's smart contract.  They deploy another smart contract, say $\SCe$. They also agree on a secret key, $mk'$. Next, the buyer places a certain deposit into $\SCe$. This deposit will be distributed among honest clients as a reward. 
%
%All clients check the buyer's deposit and proceed to the next step if they agree with the deposit amount.  
%
The extractors and $D$ deploy one of the counter-collusion smart contracts, i.e., \SCpc. These three parties deposit a certain amount on this contract.  Each honest extractor will receive a portion of $D$'s deposit for carrying out its task honestly and each dishonest extractor will lose a portion of its deposit for acting maliciously. 
Then, each client encrypts its set elements (under $mk'$ using deterministic encryption) and then represents the encrypted elements as a polynomial. The reason each client encrypts its set elements is to ensure that the privacy of the plaintext elements in the intersection will be preserved from the public.

%Then, each client represents the encryption of its set elements as a polynomial. 

%
Next, the extractors commit to the encryption of their set elements and publish the commitments. 
All clients (including $D$) take the rest of the steps in \fpsi using their input polynomials. This results in a blinded polynomial,  whose correctness is checked by \fpsi's smart contact. 

If  \fpsi's smart contact approves the result's correctness, then all parties receive the money that they deposited in \fpsi's contract. In this case, each extractor finds the set elements in the intersection. Each extractor proves to $\SCe$ that the encryptions of the elements in the intersection are among the commitments that the extractor previously published. 
If $\SCe$ accepts both extractors' proofs, then it pays each client (except the buyer) a reward, where the reward is taken from the buyer's deposit. The extractors receive their deposits back and are paid for carrying out the task honestly. Nevertheless, if $\SCe$ does not accept one of the extractors' proofs (or one extractor betrays the other), then it invokes the auditor in the counter-collusion contracts to identify the misbehaving extractor.  Then, $\SCe$ pays each honest client (except the buyer) a reward, taken from the misbehaving extractor. $\SCe$ also refunds the buyer's deposit.

If  \fpsi's smart contact does not approve the result's correctness and \aud identified misbehaving clients, then honest clients will receive (1) their deposit back from \fpsi's contract, and (2)  compensation and reward, taken from misbehaving clients. Moreover, the buyer and extractors receive their deposit back from $\SCe$. Figure \ref{fig:parties-interactions-in-ANE} outlines the interaction between parties.

\begin{figure}[htp]
    \centering
    \includegraphics[width=14cm]{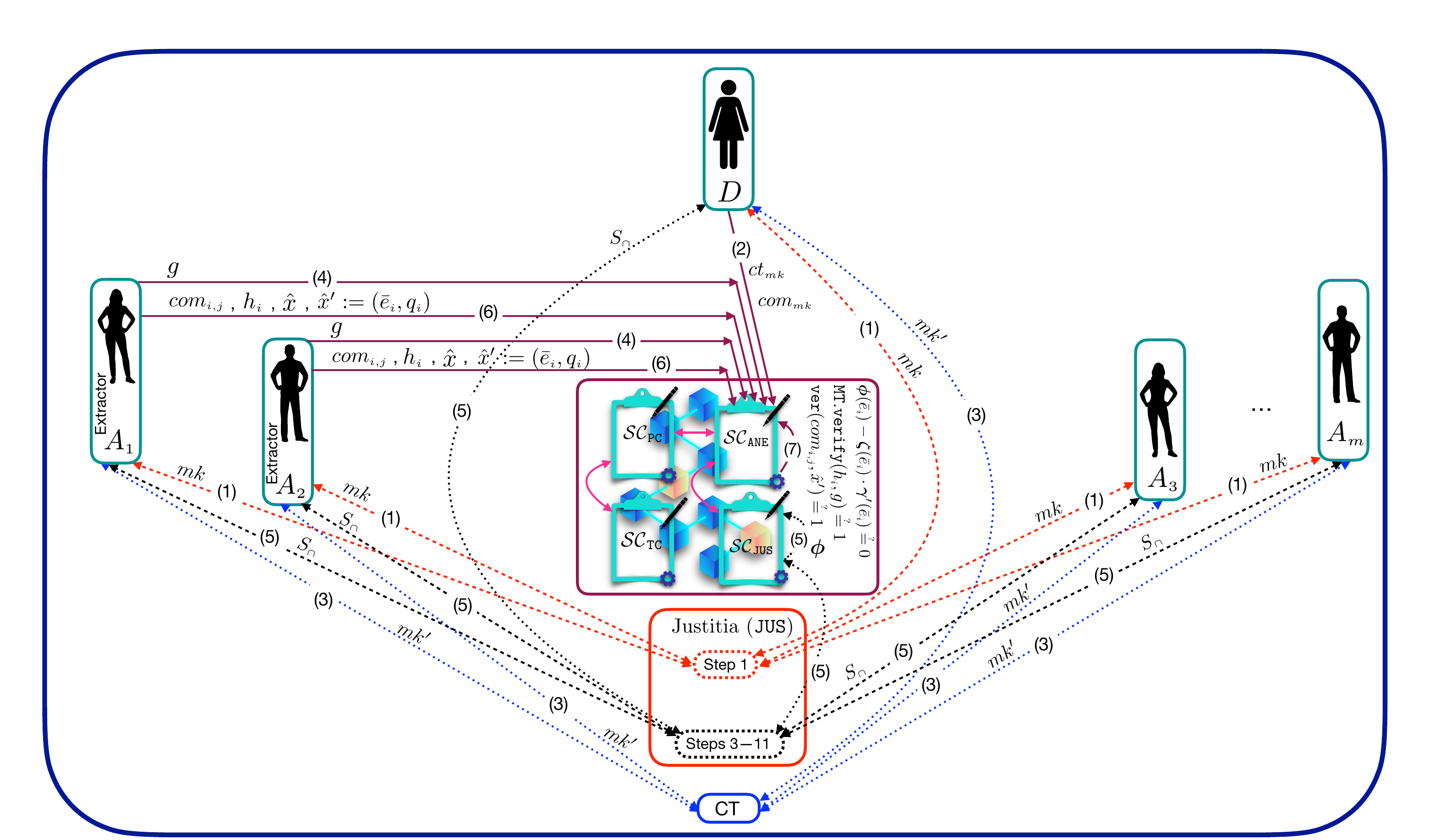}
    \caption{Outline of the interactions between parties in \withRew}\label{fig:parties-interactions-in-ANE}
\end{figure}

%However, we have to apply two essential modifications to F-PSI protocol. Generally speaking, the modifications are applied to (a) preserve the privacy of the elements in the intersection when each extractor proves the knowledge of them to the contract and (b) to identify misbehaving extractor without needing to have access to the extractor's set elements.  The role of collusion smart contracts is to create a distrust  between the extractors and buyer who may collude (out of the bound) to increase their profit.  

\subsubsection{Detailed Description of \epsi.} Next, we describe the protocol in more detail (Table \ref{table:notation-table} summarises the main notations used).

\begin{enumerate}

%\item All clients in $\textbf{P}$ sign a smart contract $\mathcal{SC}$ and deploy it to a blockchain. Then, the buyer, $ { A}_{\st {  m}}$, deposits $v$ amounts in the contract.

\item\label{e-psi::call-F-PSI-stepOne}  All clients in $\cl=\{ A_{\st 1},...,   A_{\st m},  D\}$ together run step \ref{gen-FPSI-cont} of \fpsi (in Section \ref{Fair-PSI-Protocol}) to deploy \fpsi's contract $\mathcal{SC}_{\fpsi}$ and agree on a  master key, $mk$. 

\item\label{e-psi::deploy-SC-E-PSI} All clients in $\cl$  deploy a new smart contract, $\SCe$. The address of $\SCe$ is given to all clients. 

\item The buyer, client $ { A}_{\st {  m}}$, before time $t_{\st 1}$ deposits $\Smin\cdot \vc$  amount to $\SCe$. 
\item\label{e-PSI::buyer-deposit} All clients after  time $t_{\st 2}>t_{\st 1}$ ensure that the buyer has deposited $\Smin\cdot \vc$ amount on $\SCe$. Otherwise, they abort.

\item\label{e-PSI::extractor-deposit} $D$ signs \SCpc with the extractors. $\SCe$ transfers $\Smin\cdot \rc$ amount (from the buyer deposit) to \SCpc for each extractor. This is the maximum amount to be paid to an honest extractor for honestly declaring the elements of the intersection. %Each honest extractor will be paid $\hat w'= r\cdot |{ { {S}}}_{\st\cap}|$, where 
Each extractor  deposits $\dc'=\dc+\Smin\cdot \fc$ amount in \SCpc at time $t_{\st 3}$. At time $t_{\st 4}$ all clients ensure that the extractors deposited enough coins; otherwise, they withdraw their deposit and abort. 

\item\label{e-psi::commit-to-mk} $D$ encrypts $mk$ under the public key of the dispute resolver (in \SCpc); let $ct_{\st mk}$ be the resulting ciphertext.  It also generates a commitment of $mk$ as follows: $z'=\mathtt{PRF}(mk, 0),\ com_{\st mk}=\comcom(mk, z')$. It stores $ct_{\st mk}$  and ${com}_{\st mk}$ in $\SCe$.

\item\label{e-psi::gen-mk-prime} All clients in $\cl$ engage in \ct to agree on another key, $mk'$.
\item\label{Smart-PSI:encode-elem} Each client  in $\cl$ maps the elements of its set $S:\{ s_{\st 1},..., s_{\st c}\}$ to random values by encrypting them as: $\forall i, 1\leq i\leq c: e_{\st i}=\mathtt{PRP}(mk', s_{\st i})$. 
Then, it encodes its encrypted set element as $\bar{e}_{\st i} =e_{\st i} || \mathtt{H}(e_{\st i})$.  
After that, it constructs a hash table  $\mathtt{HT}$ and inserts the encoded elements into the table. $\forall i: \mathtt{H}( \bar{e}_{\st i})={ {  {j}}}$, then $\bar{e}_{\st i}\rightarrow \mathtt{HT}_{\st {  {j}}}$. It pads every bin with random dummy elements to $d$ elements (if needed). Then,  for every bin, it builds a polynomial whose roots are the bin's content: $\bm\pi^{\st { {(I)}}}=\prod\limits^{\st d}_{\st i=1} (x-e'_{\st i})$, where $e'_{\st i}$ is either $\bar{e}_{\st i}$, or a dummy value.

\item\label{merkel-tree-cons} Every extractor in $\{A_{\st 1}, A_{\st 2}\}$: 

\begin{enumerate}
%
%\item for each bin, derives a pseudorandom polynomial: $\gamma'_{\st {  {j}}}$, using key $mk$.
%
%\item for each bin, evaluates $\gamma'_{\st {  {j}}}$ at the encode set elements of that bin: $\gamma'_{\st {  {j,i}}}=\gamma'_{\st {  {j}}}(\bar{e}^{\st { {(I)}}}_{\st i})$.
\item\label{smart-PSI::commit-to-bin} for each $j$-th bin, commits to the bin's elements: $com_{\st{i,j}}=\comcom(e'_{\st i}, q_{\st i})$, where $q_{\st i}$ is a fresh randomness  used for the commitment and $e'_{\st i}$ is either $\bar{e}_{\st i}$, or a dummy value of the bin. %Thus, if the bin contains paddings, it  commits to the paddings too. 

%\item commits to every  encrypted element.  $\forall i, 1\leq i\leq d: \mathtt{a}^{\st { {(I)}}}_{\st i}=\mathtt{Com}(e^{\st { {(I)}}}_{\st i}, q^{\st { {(I)}}}_{\st i})$, where $q^{\st { {(I)}}}_{\st i}$ is a fresh randomness  used for the commitment.
\item  constructs a Merkel tree on top of all committed values: $\mkgen(com_{\st 1,1},...,com_{\st d,h})\rightarrow g$. %Let $\mathtt{MT}^{\st g}$ be a Merkel tree with a root node $g^{\st I}$. 
\item stores the Merkel tree's root $g$ on $\SCe$.
\end{enumerate}

\item\label{e-psi::invoke-remainer-F-PSI} All clients in $\cl$   run steps \ref{ZSPA}--\ref{compute-res-poly} of \fpsi, where each client now deposits (in the $\mathcal{SC}_{\fpsi}$) $\yc'$ amount where $\yc'>\Smin\cdot \vc+{\chc}$. Recall, at the end of step \ref{compute-res-poly}  of \fpsi for each $j$-th bin (i) a random polynomial $\bm\zeta$ has been registered in $\mathcal{SC}_{\fpsi}$, (ii) a polynomial $\bm\phi$ (blinded by a random polynomial $\bm\gamma'$) has been extracted by $\mathcal{SC}_{\fpsi}$, and (iii) $\mathcal{SC}_{\fpsi}$  has checked this polynomial's  correctness. If the latter check:

\begin{itemize}
\item[$\bullet$] passes (i.e., $Flag=True$): all parties run step \ref{F-PSI::flag-is-true} of \fpsi (with a minor difference, see Section \ref{sec::Discussion-Anesidora}).  In this case, each party receives $\yc'$ amount it deposited in $\mathcal{SC}_{\fpsi}$. They proceed to step \ref{smart-PSI::extractors} below.
\item[$\bullet$]  fails (i.e., $Flag=False$): all parties run step \ref{F-PSI::flag-is-false}  of \fpsi. In this case,
(as in \fpsi) \aud is paid $\chc$ amount, and each honest party receives back its deposit, i.e., $\yc'$ amount. Also,  from the misbehaving parties' deposit  $\frac{m'\cdot \yc'-\chc}{m-m'}$ amount is sent to each honest client,  to reward and compensate the client $\Smin\cdot \lc$ and $\frac{m'\cdot \yc'-\chc}{m-m'}- \Smin\cdot \lc$ amounts respectively, where $m'$ is the total number of misbehaving parties.  Moreover, $\SCe$ returns to the buyer its deposit (i.e., $\Smin\cdot \vc$ amount paid to $\SCe$), and returns to each extractor its deposit, i.e., $\dc'$ amount paid to \SCpc. Then, the protocol halts. 
\end{itemize}

\item\label{smart-PSI::extractors} Every extractor client: 
\begin{enumerate}

\item finds the elements in the intersection. To do so, it first encodes each of its set elements to get $\bar e_{\st i}$, as explained in step \ref{Smart-PSI:encode-elem}.  
%i.e.,  it first computes $e^{\st { {(I)}}}_{\st i}=\mathtt{PRP}(mk',s^{\st { {(I)}}}_{\st i})$ and then computes $\bar{e}^{\st { {(I)}}}_{\st i} =e^{\st { {(I)}}}_{\st i} || \mathtt{H}(e^{\st { {(I)}}}_{\st i})$.
%
% and then encodes it:   (i.e. $ {e}^{\st { {(I)}}}_{\st i} =e^{\st { {(I)}}}_{\st i} || \mathtt{H}(e^{\st { {(I)}}}_{\st i})$). 
 %
 Then, it determines to which bin the encrypted value belongs, i.e., ${ {  {j}}}=\mathtt{H}( \bar{e}_{\st i})$. Next, it evaluates the resulting polynomial (for that bin) at the encrypted element. It considers the element in the intersection if the evaluation is zero, i.e., $\bm\phi( \bar{e}_{\st i})-\bm\zeta( \bar{e}_{\st i})\cdot \bm\gamma'( \bar{e}_{\st i})=0$. If the extractor is a traitor, by this point it should have signed \SCtc with $ { D}$ and provided all the inputs (e.g., correct result) to \SCtc. 

\item \label{extractor-proves} proves that every element in the intersection is among the elements it has committed to. Specifically, for each element in the intersection, say $\bar{e}_{\st i}$, it sends to $\SCe$:

\begin{enumerate}
%
%\item [$\bullet$]  the opening of commitment $\mathtt{a}'$, i.e., pair $\ddot {x}:=(mk, z')$. This is done only once for all elements in the intersection.  
%
\item [$\bullet$]  commitment $com_{\st i,j}$ (generated in step \ref{smart-PSI::commit-to-bin}, for $\bar{e}_{\st i}$) and its  opening ${\hat x}':=(\bar{e}_{\st i},  q_{\st i})$.

%
%\item [$\bullet$] the element's commitment: $\mathtt{a}^{\st { {(I)}}}_{\st {  {j,i}}}=\mathtt{Com}(\gamma'_{\st {  {j,i}}}, q^{\st { {(I)}}}_{\st i})$, where $q^{\st { {{(I)}}} }_{\st i}$ was generated in step \ref{merkel-tree-cons}.   
%
%\item[$\bullet$]  $ \bar{e}^{\st { {{(I)}}} }_{\st i}$ and it  commitment's opening:  $\mathtt{m}^{\st { {{(I)}}} }_{\st i}=(\gamma'_{\st {  {j,i}}}, q^{\st { {{(I)}}} }_{\st i})$. 

%
\item[$\bullet$]   proof $h_{\st i}$ that asserts $com_{\st i,j}$ is a leaf node of   a Merkel tree with  root $g$. 

%\item[$\bullet$] the index of the bin to which $ \bar{e}^{\st { {{(I)}}} }_{\st i}$ belongs, i.e., ${ {  {j}}}=\mathtt{H}(\bar {e}^{\st { {{(I)}}} }_{\st i})$. 
 \end{enumerate}
\item sends the opening of commitment $com_{\st mk}$, i.e., pair $\hat {x}:=(mk, z')$, to $\SCe$. This is done only once for all elements in the intersection.  

 \end{enumerate}
\item\label{e-psi::SC-verification} Contract $\SCe$:
\begin{enumerate}

\item\label{e-psi::SC-verification--derive-mk}  verifies the opening of the commitment for $mk$, i.e., $\comver(com_{\st mk},\hat{x})=1$. If the verification passes, it generates the index of the bin to which $ \bar{e}_{\st i}$ belongs, i.e., ${ {  {j}}}=\mathtt{H}(\bar {e}_{\st i})$. It  uses $mk$ to derive the pseudorandom polynomial $\bm\gamma'$ for $j$-th bin.

 \item\label{e-psi::SC-verification--check-three-vals} checks whether (i) the opening of commitment is valid,  (ii) the Merkle tree proof is valid, and (iii) the encrypted element is the resulting polynomial's root. Specifically, it ensures that the following relation holds:

$$\Bigg(\comver(com_{\st i,j}, \hat{x}')=1\Bigg) \hspace{1mm} \wedge \hspace{1mm} \Bigg(\mkver(h_{\st i},g)=1\Bigg) \hspace{1mm} \wedge \hspace{1mm}  \Bigg(\bm\phi( \bar{e}_{\st i})-\bm\zeta( \bar{e}_{\st i})\cdot \bm\gamma'( \bar{e}_{\st i})=0\Bigg)$$

%$$\mathtt{Ver_{\st com}}(\mathtt{a}^{\st { {{(I)}}} }_{\st {  {j,i}}},\mathtt{m}^{\st { {{(I)}}} }_{\st i})=1\ \ \ \ \wedge \ \ \ \ \mathtt{Ver_{\st MT}}(\mathtt{h}^{\st { {{(I)}}} }_{\st i},g^{\st { {{(I)}}} })=1 \ \ \ \ \wedge \ \ \ \  \phi( \bar{e}^{\st { {{(I)}}} }_{\st i})-\zeta( \bar{e}^{\st { {{(I)}}} }_{\st i})\cdot \gamma'_{\st {  {j,i}}}=0$$

%\item if all proofs of both extractors are valid and both extractors provide identical elements of the intersections (for each bin),  for each valid proof, it takes $m\cdot l$ coins from the buyer's deposit (in $\mathcal{SC}_{\st {  {EXT}}}$) and distributes it among all clients, except the buyer. 

\end{enumerate}

% !TEX root =main.tex

\item The parties are paid as follows. 

\begin{itemize}
\item[$\bullet$]  if all proofs of both extractors are valid, both extractors provided identical elements of the intersections (for each bin), and there is no traitor, then $\mathcal{SC}_{\epsi}$:
\begin{enumerate}
 \item takes $|S_{\st\cap}|\cdot m\cdot \lc$ amount from the buyer's deposit (in $\mathcal{SC}_{\epsi}$) and distributes it among all clients, except the buyer. 
 \item calls \SCpc which returns the extractors' deposit (i.e., $\dc'$ amount each) and pays each extractor $|S_{\st\cap}|\cdot \rc$ amount, for doing their job correctly. 
 \item checks if $|{ { {S}}}_{\scriptscriptstyle\cap}|<\Smin$. If the check passes, then it returns $(\Smin-|S_{\scriptscriptstyle\cap}|)\cdot \vc$ amount  to the buyer.
 \end{enumerate}
\item[$\bullet$] if both extractors failed to deliver any result, then $\mathcal{SC}_{\epsi}$:
\begin{enumerate}
\item refunds the buyer, by sending $\Smin\cdot \vc$ amount (deposited in $\mathcal{SC}_{\epsi}$) back to the buyer. 
\item retrieves each extractor's deposit (i.e., $\dc$ amount) from the \SCpc and distributes it among the rest of the clients (except the buyer and extractors).  
 \end{enumerate}
 \item[$\bullet$]\label{smart-PSI-inconsistency} Otherwise (e.g., if some proofs are invalid, if an extractor's result is inconsistent with the other extractor's result, or there is a traitor), $\mathcal{SC}_{\epsi}$ invokes (steps 8.c and 9 of) \SCpc and its auditor to identify the misbehaving extractor, with the help of $ct_{\st mk}$ after decrypting it. Then, $\mathcal{SC}_{\epsi}$ asks \SCpc to pay the auditor the total amount of $\chc$ taken from the deposit of the extractor(s) who provided incorrect result to $\mathcal{SC}_{\epsi}$. Moreover,

\begin{enumerate}
\item if \underline{both extractors cheated}:
\begin{enumerate}
\item\label{both-cheated-no-traitor} if there \underline{is no traitor}, then $\mathcal{SC}_{\epsi}$ refunds the buyer, by sending $\Smin\cdot \vc$ amount (deposited in $\mathcal{SC}_{\epsi}$) back to the buyer. It also distributes $2\cdot \dc'- \chc$ amount (taken from the extractors' deposit in \SCpc) among the rest of  clients  (except the buyer and extractors). %The Prisoner's contract pays its dispute resolver $ch$ amount. 
\item if there \underline{is a traitor}, then:
%%%%%
\begin{enumerate}
\item\label{both-cheated-honest-traitor} if the traitor delivered a \underline{correct result} in \SCtc, $\mathcal{SC}_{\epsi}$ retrieves $\dc'-\dc$ amount from the other dishonest extractor's deposit (in \SCpc) and distributes it among the rest of the clients (except the buyer and dishonest extractor). Also, it asks \SCpc to send $|S_{\st\cap}|\cdot \rc+\dc'+\dc-\chc$ amount to the traitor (via \SCtc). % and $\hat{ch}$ amount to the dispute resolver.
\SCtc refunds the traitor's deposit, i.e., $\chc$ amount. It refunds the buyer, by sending $\Smin\cdot \vc-|S_{\st\cap}|\cdot \rc$ amount (deposited in $\mathcal{SC}_{\epsi}$) back to it.

 %Otherwise (i.e., if it delivered an incorrect result in the Traitor's contract), the Traitor's contract refunds the traitor's deposit (i.e., $ch$ amount).
%
\item if the traitor delivered an \underline{incorrect result} in \SCtc, $\mathcal{SC}_{\epsi}$ pays the buyer and rest of clients in the same way it does in step \ref{both-cheated-no-traitor}. 
%
%distributes $2\cdot \hat d- \hat{ch}$ amount (taken from the extractors deposit in the Prisoner's contract) among the rest of the clients (except the buyer and extractors). 
%
%The Prisoner's contract pays its dispute resolver $\hat{ch}$ amount.
%
\SCtc refunds the traitor, i.e., $\chc$ amount.  %It refunds the buyer, by sending ${\resizeT {\textit {S}}}_{\resizeS {\textit  min}}\cdot v$ amount (deposited in $\mathcal{SC}_{\resizeS {\textit  {EXT}}}$) back to the buyer.

\end{enumerate}
%%%%%
\end{enumerate}
\item if \underline{one of the extractors cheated}: 
\begin{enumerate}
\item if there \underline{is no traitor}, then $\mathcal{SC}_{\epsi}$ calls \SCpc that (a) returns the honest extractor's deposit (i.e., $\dc'$ amount), (b) pays this extractor $|S_{\st\cap}|\cdot \rc$ amount, for doing its job honestly, and (c) pays this extractor $ \dc-\chc$ amount taken from the dishonest extractor's deposit. 
%
%and (d) pays its dispute resolver $\hat {ch}$ amount taken from the dishonest extractor's deposit. 
%
%Also, $\mathcal{SC}_{\resizeS {\textit  {EXT}}}$ sends ${\resizeT {\textit {S}}}_{\resizeS {\textit  min}}\cdot v-|S_{\st\cap}|\cdot r$ amount (deposited in $\mathcal{SC}_{\resizeS {\textit  {EXT}}}$) back to the buyer.
%
 $\mathcal{SC}_{\epsi}$ pays the buyer and the rest of the clients in the same way it does in step \ref{both-cheated-honest-traitor}.

% It retreives  $\hat d - \hat c - \hat{ch}$ amount from the dishonest extractor's deposit (in the Prisoner’s contract) and distributes it among the rest of clients (except the buyer and dishonest extractor). %
%
\item if there \underline{is a traitor}
%
%%%%%

\begin{enumerate}
\item\label{one-cheated-exists-traitor-honest-traitor}  if the traitor delivered a \underline{correct result} in \SCtc (but it cheated in $\mathcal{SC}_{\epsi}$), then $\mathcal{SC}_{\epsi}$ calls \SCpc that (a) returns the other honest extractor's deposit (i.e., $\dc'$ amount), (b) pays the honest extractor $|S_{\st\cap}|\cdot \rc$ amount taken from the buyer's deposit, for doing its job honestly,  (c) pays the honest extractor $\dc- \chc$ amount taken from the traitor's deposit,  
%
%(d) pays its dispute resolver $\hat{ch}$ amount taken from the traitor extractor's deposit (deposited in $\mathcal{SC}_{\resizeS {\textit  {EXT}}}$), 
%
 (d)
 pays to the traitor $|S_{\st\cap}|\cdot \rc$ amount taken from the buyer’s deposit (via the \SCtc), and (e) refunds the traitor $\dc'-\dc$ amount taken from its own deposit.  \SCtc refunds the traitor's deposit (i.e., $\chc$ amount).  $\mathcal{SC}_{\epsi}$ takes $|S_{\st\cap}|\cdot m\cdot \lc$ amount from the buyer's deposit (in $\mathcal{SC}_{\epsi}$) and distributes it among all clients, except the buyer. If $|{ { {S}}}_{\scriptscriptstyle\cap}|<\Smin$,  then  $\mathcal{SC}_{\epsi}$ returns $(\Smin-|{ { {S}}}_{\scriptscriptstyle\cap}|)\cdot \vc$ amount (deposited in $\mathcal{SC}_{\epsi}$) back  to the buyer.

\item  if the traitor delivered an \underline{incorrect result} in \SCtc (and it cheated in $\mathcal{SC}_{\epsi}$), then $\mathcal{SC}_{\epsi}$ pays the honest extractor in the same way it does in step \ref{one-cheated-exists-traitor-honest-traitor}.  
%
%calls Prisoner's contract that (a) returns the other honest extractor's deposit (i.e., $\hat d$ amount), (b) pays the honest extractor $|S_{\st\cap}|\cdot r$ amount, for doing its job honestly, and (c) pays the honest extractor $\hat c$ amount taken from the traitor extractor's deposit. 
%
%, and  (d) pays its dispute resolver $\hat{ch}$ amount taken from the traitor extractor's deposit (deposited in $\mathcal{SC}_{\resizeS {\textit  {EXT}}}$). 
%
\SCtc refunds the traitor's deposit, i.e., $\chc$ amount. 
%
%$\mathcal{SC}_{\resizeS {\textit  {EXT}}}$ takes ${\resizeT {\textit {S}}}_{\resizeS {\textit  min}} \cdot f$ amount from the traitor's deposit (in $\mathcal{SC}_{\resizeS {\textit  {EXT}}}$) and distributes it among all clients, except the buyer and traitor. 
%
%Also, $\mathcal{SC}_{\resizeS {\textit  {EXT}}}$ sends ${\resizeT {\textit {S}}}_{\resizeS {\textit  min}}\cdot v-|S_{\st\cap}|\cdot r$ amount (deposited in $\mathcal{SC}_{\resizeS {\textit  {EXT}}}$) back to the buyer. 
%
Also, $\mathcal{SC}_{\epsi}$ pays the buyer and the rest of the clients in the same way it does in step \ref{both-cheated-honest-traitor}.

\end{enumerate}

\end{enumerate}

\end{enumerate}
\end{itemize}

\end{enumerate}

 \begin{theorem}\label{theorem::E-PSI-security}
If  $\mathtt{PRP}$, $\mathtt{PRF}$, the commitment scheme, smart contracts, the Merkle tree scheme, \fpsi and the counter-collusion contracts are secure and the public key encryption is semantically secure,  then  \epsi realises  $f^{\st \text{PSI}}$ with $\bar Q$-fairness-and-reward (w.r.t. Definition \ref{def::PSI-Q-fair-reward}) in the presence of $m-3$ static active-adversary clients $A_{\st j}$s and $two$ rational clients $A_{\st i}s$ or a static passive dealer $D$ or passive auditor $Aud$, or passive public which sees the intersection cardinality.
 \end{theorem}

Before we prove Theorem \ref{theorem::E-PSI-security} in Section \ref{sec::E-PSI-proof}, we  present several remarks on the \epsi. 
%
%\begin{remark}

\subsection{Further Discussion on \withRew}\label{sec::Discussion-Anesidora}
There is a simpler but costlier approach to finding the intersection without involving the extractors; that is the smart contract finds the (encoded) elements of the intersection and distributes the parties' deposit according to the number of elements it finds. This approach is simpler, as we do not need the involvement of (i) the extractors and (ii) the three counter collusion contracts. Nevertheless, it is costlier, because the contract itself needs to factorise the unblinded resulting polynomial and find the roots, which would cost it $O(d^{\st 2})$ for each bin, where $d$ is the size of each bin. Our proposed approach however moves such a computation off-chain, leading to a lower monetary computation cost. 
%\end{remark}

%\begin{remark}\label{remark::element-encoding}
The reason each client uses the hash-based padding to encode each encrypted element $e_{\st i}$  as $\bar{e}_{\st i} =e_{\st i} || \mathtt{H}(e_{\st i})$ is to allow the auditor in the counter collusion contracts to find the error-free intersection, without having to access to one of the original (encrypted) sets. 

Compared to \fpsi, there is a minor difference in finding the result in \epsi. Specifically, because in \epsi each set element  $s_{\st i}$ is encoded as  (i) $e_{\st i}=\mathtt{PRP}(mk', s_{\st i})$ and then (ii) $\bar{e}_{\st i} =e_{\st i} || \mathtt{H}(e_{\st i})$ by a client, then when the client wants to find the intersection it needs to first regenerate $\bar{e}_{\st i}$ as above and then treat it as a set element to check if  $\bm\phi'(\bar e_{\st i})=0$, in step \ref{F-PSI::find-intersection} of \fpsi.

In \epsi, each extractor uses double-layered commitments (i.e., it first commits to the encryption of each element and then constructs a Merkle tree on top of all commitments) for efficiency and privacy purposes. Constructing a Merkle tree on top of the commitments allows the extractor to store only a single value in $\SCe$ would impose a much lower storage cost compared to the case where it would store all commitments in $\SCe$. Also, committing to the elements' encryption allows it to hide from other clients the encryption of those elements that are not in the intersection. Recall that encrypting each element is not sufficient to protect one client's elements from the rest of the clients, as they all know the decryption key.

To increase their reward, malicious clients may be tempted to insert ``garbage'' elements into their sets with the hope that those garbage elements appear in the result and accordingly they receive a higher reward. However, they would not succeed as long as there exists a semi-honest client (e.g., dealer $D$) which uses actual set elements. In this case, by the set intersection definition, those garbage elements will not appear in the intersection.

In \epsi, for the sake of simplicity, we let each party receive a fixed reward, i.e., $\lc$, for every element it contributes to the intersection. However, it is possible to make the process more flexible/generic. For instance, we could define a Reward Function $RF$ that takes $\lc$, an (encoded) set element $e_{\st i}$ in the intersection, its distribution/value $val_{\st e_{\st i}}$, and output a reward $rew_{\st e_{\st i}}$ that each party should receive for contributing that element to the intersection, i.e., $RF(\lc, e_{\st i}, val_{\st e_{\st i}})\rightarrow rew_{\st e_{\st i}}$. 

%\end{remark}

%\begin{remark}
%The Merkle three is used to reduce the contract-side storage cost  each time an instance of the PSI is run. 
%
%\end{remark}

%\begin{remark}
% The reason smart PSI has a  verification mechanism for the extractors (instead of solely relying on the arbiter in the counter collusion contracts) is to minimise the role the arbiter, and let $\mathcal{SC}_{\epsi}$ resolve most of  dispute. 
%\end{remark}

%\begin{remark}
%In \epsi (unlike \fpsi), the intersection cardinality is revealed to the public. The clients can use padding to hide the exact number of elements. Specifically, all clients can agree on a set of elements and all insert them into their set in the setup phase. 
%\end{remark}

% !TEX root =main.tex

\subsection{Proof of \epsi}\label{sec::E-PSI-proof}

In this section, we prove Theorem \ref{theorem::E-PSI-security}, i.e., the security of \epsi.

\begin{proof}
We prove the theorem by considering the case where each party is corrupt, at a time.

\

%\noindent\textbf{Case 1: Corrupt $m-3$ clients in $\{  {  A}_{ \st {   1}}, ...,   {  A}_{ \st {   m}}\}\setminus\{E_{\st 1},E_{\st 2}\}$}.  Let $G$ be a set of at most $m-3$ corrupt clients, where $G\subset \{  {  A}_{ \st {   1}}, ...,   {  A}_{ \st {   m}}\}\setminus\{E_{\st 1},E_{\st 2}\}$. Let set $\hat G$ be a set of honest clients excluding the (dealer and) extractors, i.e., $\hat G=\{  {  A}_{ \st {   1}}, ...,   {  A}_{ \st {   m}}\}\setminus\{G,E_{\st 1},E_{\st 2}\}$. Also, let $\mathsf{Sim}^{\st\text{E-PSI}}_{\st A_{\st j}}$ be the simulator, which uses a subroutine adversary, $\mathcal{A}_{\st A_{\st j}}$.  Next, we explain how $\mathsf{Sim}^{\st \text{E-PSI}}_{\st _{\st A_{\st j}}}$, which receives the input sets of honest dealer client $D$ and honest client(s) in $\hat G$,  works. 

\noindent\textbf{Case 1: Corrupt extractors $\{A_{\st 1}, A_{\st 2}\}$ and $m-3$ clients in $\{  {  A}_{ \st {   3}}, ...,   {  A}_{ \st {   m}}\}$}.  Let set $G$ include extractors $\{A_{\st 1}, A_{\st 2}\}$ and a set of at most $m-3$ corrupt clients in $\{  {  A}_{ \st {   3}}, ...,   {  A}_{ \st {   m}}\}$. Let set $\hat G$ be a set of honest clients in  $\{  {  A}_{ \st {   3}}, ...,   {  A}_{ \st {   m}}\}$. Also, let $\mathsf{Sim}^{\st\epsi}_{\st A}$ be the simulator. We let the simulator  interact with (i) active adversary  $\mathcal{A}'$ that may corrupt $m-3$ clients in $\{  {  A}_{ \st {   3}}, ...,   {  A}_{ \st {   m}}\}$, and (ii) two rational adversaries  $\mathcal{A}'':=(\mathcal{A}_{\st 1}, \mathcal{A}_{\st 2})$ that corrupt extractors $(A_{\st 1},A_{\st 2})$ component-wise.  
%
%We let $\mathcal{A}:=(\mathcal{A}', \mathcal{A}'')$. 
%
In the simulation, before the point where the extractors are invoked to provide proofs and results, the simulator directly deals with active adversary  $\mathcal{A}'$. However, when the extractors are involved (to generate proofs and extract the result) we require the simulator to interact with each rational adversary $\mathcal{A}_{\st 1}$ and $\mathcal{A}_{\st 2}$.  We allow these two subroutine adversaries $(\mathcal{A}', \mathcal{A}'')$  to internally interact with each other.  Now, we explain how $\mathsf{Sim}^{\st \epsi}_{\st A}$, which receives the input sets of honest dealer $D$ and honest client(s) in $\hat G$,  works.

\begin{enumerate}
\item constructs and deploys two smart contracts (for \fpsi and \epsi). It sends the contracts' addresses to $\mathcal{A}'$.  It also simulates \ct and receives the output value, $ {mk}$, from its functionality, $f_{\st \ct}$.
\item\label{sim::case1-buyer-deposit} deposits $S_{\st min}\cdot \vc$ amount to $\mathcal{SC}_{\epsi}$  if buyer $A_{\st m}$ is honest, i.e.,  $A_{\st m}\in \hat G$. Otherwise, $\mathsf{Sim}^{\st \epsi}_{\st _{\st A}}$ checks if  $\mathcal{A}'$ has deposited $\Smin\cdot \vc$ amount in $\mathcal{SC}_{\epsi}$. If the check fails, it instructs the ledger to refund the coins that every party deposited and sends message $abort_{\st 1}$ to TTP (and accordingly to all parties); it outputs whatever $\mathcal{A}'$ outputs and then halts.

%it aborts if $\mathcal{A}'$ has not deposited enough coins. 
%
\item\label{e-psi::deploy-prisoners} constructs and deploys a (Prisoner's) contract and transfers $\wc=\Smin\cdot \rc$  amount for each extractor. $\mathsf{Sim}^{\st \epsi}_{\st A}$ ensures that each extractor deposited $\dc'=\dc+\Smin\cdot \fc$ coins in this contract; otherwise, it instructs the ledger to refund the coins that every party deposited and sends message $abort_{\st 1}$ to TTP; it outputs whatever $\mathcal{A}'$ outputs and then halts.

\item encrypts $ mk$ under the public key of the dispute resolver;  let $ct_{\st {mk}}$ be the resulting ciphertext.   It also generates a commitment of ${mk}$ as follows: $ com_{\st {mk}}=\comcom({mk}, \mathtt{PRF}({mk}, 0))$. It stores $ct_{\st {mk}}$  and ${com}_{\st {mk}}$ in $\mathcal{SC}_{\epsi}$. 
\item  simulates \ct again and receives the output value $ {mk}'$ from  $f_{\st \ct}$.

\item\label{sim::case1-merkle-tree-root} receives from $\mathcal{A}'$ a Merkle tree's root $g'$ for each extractor.

\item\label{sim::case1-F-PSI} simulates the steps of \ref{ZSPA}--\ref{compute-res-poly} in \fpsi. For completeness, we include the steps that the simulator takes in this proof. Specifically, $\mathsf{Sim}^{\st \epsi}_{\st _{\st A}}$:

\begin{enumerate}
\item\label{sim::E-PSI-ZSPA-A-invocation-} simulates \zspaa for each bin and receives the output value $( k,  g,  q)$ from $f^{\st \zspaa}$.
\item deposits in the contract the  amount of $\yc'=\Smin\cdot \vc+\chc$ for client $D$ and each honest client in $\hat G$. It sends to $\mathcal{A}'$ the amount deposited in the contract. 
\item\label{sim::case1-check-Adv-deposited-y.|G|}  checks if $\mathcal{A}'$ has deposited $\yc'\cdot |G|$ amount (in addition to $\dc'$ amount deposited in step \ref{e-psi::deploy-prisoners} above). If the check fails, it instructs the ledger to refund the coins that every party deposited and sends message $abort_{\st 1}$ to TTP (and accordingly to all parties); it outputs whatever $\mathcal{A}'$ outputs and then halts.
\item picks a random polynomial ${\bm\zeta}$ of degree $1$, for each bin. $\mathsf{Sim}^{\st \epsi}_{\st A}$, for each client $  {  C}\in \{  {  A}_{ \st {   1}}, ...,   {  A}_{ \st {   m}}\}$ allocates to each bin two degree $d$ random polynomials: (${\bm\omega}^{ \st {  {(D,C)}}}, {\bm\rho}^{ \st {  {(D,C)}}}$), and   two  degree $3d+1$ random polynomials: (${\bm\gamma}^{ \st {  {(D,C)}}}$, ${\bm\delta}^{ \st {  {(D,C)}}}$). Also, $\mathsf{Sim}^{\st \epsi}_{\st A}$ for each honest client $C'\in \hat G$, for each bin, picks two  degree $d$ random polynomials: (${\bm\omega}^{ \st {  {(C',D)}}}$, ${\bm\rho}^{ \st {  {(C',D)}}}$). 

\item\label{E-PSI::sim-A-first-VOPR-invocation} simulates \vopr using inputs ${\bm\zeta} \cdot {\bm\omega}^{ \st {  {(D,C)}}}$ and ${\bm\gamma}^{ \st {  {(D,C)}}}$ for each bin. It receives the inputs of clients $C''\in G$, i.e., ${\bm\omega}^{ \st {  {(C'',D)}}}\cdot {\bm\pi}^{ \st {  {(C'')}}}$, from its functionality $f^{\st \vopr}$, for each bin.  
\item extracts the roots of polynomial ${\bm\omega}^{ \st {  {(C'',D)}}}\cdot {\bm\pi}^{ \st {  {(C'')}}}$ for each bin and appends those roots that are in the sets universe to a new set $S^{\st(C'')}$. 
\item simulates again \vopr using inputs ${\bm\zeta}\cdot {\bm \rho}^{ \st {  {(D,C)}}}\cdot {\bm\pi}^{ \st {  {(D)}}}$ and ${\bm\delta}^{ \st {  {(D,C)}}}$, for each bin.
\item sends to TTP the input sets of all parties; specifically, (i) client $D$'s input set: $S^{\st (D)}$, (ii) honest clients' input sets: $S^{\st (C')}$ for all $C'$ in $\hat G$, and (iii) $\mathcal{A}'$'s input sets: $S^{\st(C'')}$, for all $C''$ in $G$.  For each bin, it receives the intersection set, $S_{\st\cap}$, from TTP. 
\item represents the intersection set for each bin as a polynomial, ${\bm \pi}$, as follows. First, it encrypts each element $s_{\st i}$ of $S_{\st\cap}$ as $e_{\st i}=\mathtt{PRP}({mk}', s_{\st i})$. Second, it encodes each encrypted element as $\bar{e}_{\st i} =e_{\st i} || \mathtt{H}(e_{\st i})$. Third, it constructs ${\bm \pi}$ as ${\bm \pi}=\prod\limits^{\st |S_{\st\cap}|}_{\st i=1 }(x-s_{\st i})\cdot \prod\limits^{\st d-|S_{\st\cap}|}_{\st j=1 }(x-u_{\st j})$, where $u_{\st j}$ is a dummy value. 
\item constructs polynomials ${\bm\theta}^{ \st {  {(C')}}}_{\st 1}={\bm\zeta} \cdot {\bm\omega}^{ \st {  {(D,C')}}}\cdot {\bm\omega}^{ \st {  {(C',D)}}}\cdot {\bm\pi}+ {\bm\gamma}^{ \st {  {(D,C')}}}, {\bm\theta}^{ \st {  {(C')}}}_{\st 2}= {\bm\zeta} \cdot {\bm\rho}^{ \st {  {(D,C')}}}\cdot {\bm\rho}^{ \st {  {(C',D)}}}\cdot {\bm\pi}+ {\bm\delta}^{ \st {  {(D,C')}}}$, and $ {\bm\nu}^{ \st {  {(C')}} }=  {\bm\theta}^{ \st {  {(C')}}}_{\st 1}+ {\bm\theta}^{ \st {  {(C')}}}_{\st 2}+ {\bm \tau}^{ \st {  {(C')}} }$, for each bin and each honest client $C'\in\hat G$, such that $ {\bm\tau}^{ \st {  {(C)}}}=\sum\limits^{\st 3d+2}_{\st i=0}z_{\st i,c}\cdot x^{\st i}$ and each value $z_{\st i,c}$ is derived from key $ k$ generated in step \ref{sim::E-PSI-ZSPA-A-invocation-}.  It sends to $\mathcal{A}'$ polynomial  $ {\bm\nu}^{ \st {  {(C')}} }$ for each bin and each honest client $C'\in\hat G$. 
\item\label{E-PSI::sim-A-receive-nu-from-adv} receives ${\bm\nu}^{ \st {  {(C'')}} }$  from $\mathcal{A}'$, for each bin and each corrupt client $C''\in G$. It checks whether the output for every $C''$ has been provided. Otherwise, it halts. 
\item if there is any abort within steps \ref{E-PSI::sim-A-first-VOPR-invocation}--\ref{E-PSI::sim-A-receive-nu-from-adv}, then it sends $abort_{\st 2}$ to TTP and instructs the ledger to refund the coins that every party deposited.  It outputs whatever $\mathcal{A}'$ outputs and then halts. 
\item constructs polynomial  ${\bm\nu}^{ \st {  {(D)}}}={\bm\zeta} \cdot  {\bm\omega'}^{ \st {  {(D)}}}\cdot {\bm\pi} - \sum\limits^{ \st {   A}_{ \st {   m}}}_{  \st {  {C }= }  \st {   A}_{ \st {  1}}}({\bm\gamma}^{ \st {  {(D,C)}}} + {\bm\delta}^{ \st {  {(D,C)}}}) + {\bm\zeta} \cdot {\bm\gamma'}$ for each bin on behalf of client $D$, where ${\bm\omega'}^{ \st {  {(D)}}}$ is a fresh random polynomial of degree $d$ and ${\bm\gamma'}$ is a pseudorandom polynomial derived from $ {mk}$.
\item sends to $\mathcal{A}'$ polynomials ${\bm\nu}^{ \st {  {(D)}}}$ and ${\bm\zeta}$ for each bin. 
\item\label{sim::case1-sim-check-res} computes polynomial $ {\bm\phi'}$ as $ {\bm\phi'} = \sum\limits_{\st \forall C''\in G}{\bm\nu}^{ \st {  {(C'')}} }- \sum\limits_{\st \forall C''\in G}({\bm\gamma}^{ \st {  {(D,C'')}}} + {\bm\delta}^{ \st {  {(D,C'')}}})$, for every bin. Next, it checks if  ${\bm\zeta}$  divides ${\bm\phi'}$, for every bin. If the check passes, it sets $Flag=True$. Otherwise, it sets $Flag=False$. 
\item if $Flag=True$, then instructs the ledger to send back each party's deposit, i.e., $\yc'$ amount. It sends a message $deliver$ to TTP.  It proceeds to step \ref{sim::case1-what-adv-extract-sends} below. 

\item if $Flag=False$: 
\begin{enumerate}
 \item receives $|G|$ keys of the $\mathtt{PRF}$ from $\mathcal{A}'$, i.e., $\vv k'=[  k'_{\st 1}, ...,  k'_{\st |G|}]$, for every bin. %, which should be the output of $f^{\st \text{ZSPA-A}}$ in step \ref{sim::ZSPA-A-invocation} above. 
\item checks if $ k'_{\st j}= k$, for every $ k'_{\st j}\in\vv k'$. Recall,  $ k$ was generated in step \ref{sim::ZSPA-A-invocation}. It constructs an empty list $ L'$ and appends to it the indices (e.g., $j$) of the keys that do not pass the above check. 
 \item receives from $f^{\st \zspaa}$ the output containing a vector of random polynomials, $\vv{\mu}'$, for each valid key. 
 \item sends to  $\mathcal{A}'$, $ L'$ and $\vv{\mu}'$, for every bin. 
 \item  for each bin of client $  {  C}$ whose index is not in $ L'$ computes polynomial ${\bm\chi}^{ \st {  {(D, C)}}}$ as
 ${\bm\chi}^{ \st {  {(D, C)}}}={\bm\zeta}\cdot {\bm\eta}^{ \st {  {(D,C)}}}-({\bm\gamma}^{ \st {  {(D,C)}}}+{\bm\delta}^{ \st {  {(D,C)}}})$,   where ${\bm\eta}^{ \st {  {(D,C)}}}$ is a fresh random polynomial of degree $3d+1$. Note that  $C$ includes both honest and corrupt clients, except those clients whose index is in  $ L'$. $\mathsf{Sim}^{\st \epsi}_{\st A}$ sends every polynomial ${\bm\chi}^{ \st {  {(D, C)}}}$ to  $\mathcal{A}'$. 
 \item\label{sim::case1-check-final-res-for-each-client} given each ${\bm\nu}^{ \st {  {(C'')}} }$ (by $\mathcal{A}'$ in step \ref{E-PSI::sim-A-receive-nu-from-adv}), computes polynomial $ {\bm\phi'}^{ \st {  {(C'')}} }$ as follows: ${\bm\phi'}^{ \st {  {(C'')}} } = {\bm\nu}^{ \st {  {(C'')}} }- {\bm\gamma}^{ \st {  {(D,C'')}}} - {\bm\delta}^{ \st {  {(D,C'')}}}$, for every bin.  $\mathsf{Sim}^{\st \epsi}_{\st A}$ checks if  ${\bm\zeta}$  divides $ {\bm\phi'}^{ \st {  {(C'')}} }$, for every bin. It appends the index of those clients that did not pass the above check to a new list, $ L''$. 
 \item if  $ L'$ or $ L''$ is not empty, then instructs the ledger: (a) to refund $\yc'$ amount to each  client whose index is not in $ L'$ and $ L''$, (b) to retrieve $\chc$ amount from the adversary (i.e., one of the parties whose index is in one of the lists) and send the $\chc$ amount to \aud, and (c) to reward and compensate each honest party (whose index is not in the two lists)  $\frac{m'\cdot \yc'-\chc}{m-m'}$ amount, where $m'=| L'|+| L''|$.  Then, it sends message $abort_{\st 3}$ to TTP. 
\item outputs whatever $\mathcal{A}'$ outputs and halts.
 \end{enumerate}
 \end{enumerate}
\item\label{sim::case1-what-adv-extract-sends}  for each $I\in \{1,2\}$, receives from $\mathcal{A}_{\st I}$ (1) a set $E^{\st (I)}$ of encoded encrypted elements, e.g., $\bar e_{\st i}$, in the intersection, (2) each $\bar e_{\st i}$'s commitment $com_{\st i,j}$, (3) each $com_{\st i,j}$'s opening $\hat{x}'$, (4) a proof $h_{\st i}$ that each $com_{\st i,j}$ is a leaf node of a Merkle tree with root $g'$ (given to simulator in step \ref{sim::case1-merkle-tree-root} above), and (5) the opening $\hat x$ of commitment  $com_{\st {mk}}$.
\item encrypts each element $s_{\st i}$ of $S_{\st\cap}$ as $e_{\st i}=\mathtt{PRP}({mk}', s_{\st i})$. Then, it encodes each encrypted element as $\bar{e}_{\st i} =e_{\st i} || \mathtt{H}(e_{\st i})$. Let set $S'$ include all encoded encrypted elements in the intersection.
\item\label{sim::case1-traitor-cont} sings a \SCtc with $\mathcal{A}_{\st I}$, if $\mathcal{A}_{\st I}$ decides to be a traitor extractor. In this case, $\mathcal{A}_{\st I}$, provides the intersection to \SCtc.  $\mathsf{Sim}^{\st \epsi}_{\st A}$ checks this intersection's validity. Shortly (in step \ref{sim::case1-both-cheated-traitor-correct-res}), we will explain how $\mathsf{Sim}^{\st \epsi}_{\st A}$ acts based on the outcome of this check. 
\item\label{sim::case1-check-result-correctness} checks if each set $E^{\st (I)}$ equals set $S'$. 
\item checks if $com_{\st i,j}$ matches the opening $\hat{x}'$ and the opening corresponds to a unique element in $S'$. 
\item\label{sim::case1-ver-Merkle-tree-proof} verifies each commitment's proof, $h_{\st i}$. Specifically, given the proof and root $g$, it ensures the commitment $com_{\st i,j}$  is a leaf node of a Merkle tree with a root node $g'$. It also checks whether the opening $\hat x$ matches $com_{\st {mk}}$. 
%
%\item if any of the checks in steps \ref{sim::case1-check-result-correctness}-\ref{sim::case1-ver-Merkle-tree-proof} fail, it aborts.
%
\item \label{sim::case1-all-check-pass} if all the checks in steps \ref{sim::case1-check-result-correctness}--\ref{sim::case1-ver-Merkle-tree-proof} pass, then instructs the ledger (i) to take $|S_{\st\cap}|\cdot m\cdot \lc$ amount from the buyer's deposit and distributes it among all clients, except the buyer, (ii) to return the extractors deposit (i.e., $\dc'$ amount each) and pay each extractor $|S_{\st\cap}|\cdot \rc$ amount, and (iii) to return $(\Smin-|S_{\scriptscriptstyle\cap}|)\cdot \vc$ amount to the buyer.
\item \label{sim::case1-no-result} if neither extractor sends the extractor's set intersection $(E^{\st (A_{\st 1})}, E^{\st (A_{\st 2})})$  in step \ref{sim::case1-what-adv-extract-sends}, then instructs the ledger  (i) to refund the buyer, by sending $\Smin\cdot \vc$ amount back to the buyer and (ii) to
 retrieve each extractor's deposit (i.e., $\dc'$ amount) from \SCpc and distribute it among the rest of the clients (except the buyer and extractors).  
\item if the checks in step \ref{sim::case1-all-check-pass} fail or in step  \ref{sim::case1-no-result} both $\mathcal{A}_{\st 1}$  and $\mathcal{A}_{\st 2}$ send the extractors' set intersection but they are inconsistent with each other, then it tags the extractor whose proof or set intersection was invalid as a misbehaving extractor. $\mathsf{Sim}^{\st \epsi}_{\st A}$ instructs the ledger to pay the auditor (of \SCpc) the total amount of $\chc$ coins taken from the misbehaving extractor(s) deposit. Furthermore, $\mathsf{Sim}^{\st \epsi}_{\st A}$ takes the following steps. 
\begin{enumerate}
\item\label{sim::case1-both-cheated-no-traitor} if both extractors cheated and there is no traitor, then instructs the ledger (i) to refund the buyer $\Smin\cdot \vc$ amount, and (ii) to take $2 \dc'-\chc$ amount from the misbehaving extractors'  deposit and distribute it to the rest of the clients except the buyer and extractors. 
\item\label{sim::case1-both-cheated-traitor-correct-res} if both extractors cheated, there is a traitor, and the traitor delivered a correct result (in step \ref{sim::case1-traitor-cont}), then instructs the ledger (i) to take $\dc'-\dc$ amount from the other misbehaving extractor's deposit and distribute it among the rest of the clients (except the buyer and dishonest extractor), (ii) to distribute $|S_{\st\cap}|\cdot \rc+\dc'+\dc-\chc$ amount to the traitor, (iii) to refund the traitor  $\chc$ amount, and (iv) to refund the buyer $\Smin\cdot \vc-|S_{\st\cap}|\cdot \rc$ amount. 
\item if both extractors cheated, there is a traitor, and the traitor delivered an incorrect result (in step \ref{sim::case1-traitor-cont}), then instructs the ledger to distribute coins the same way it does in step \ref{sim::case1-both-cheated-no-traitor}. 
\item if one of the extractors cheated and there is no traitor, then instructs the ledger (i) to return the honest extractor’s deposit (i.e., $\dc'$ amount), (ii) to pay the honest extractor $|S_{\st\cap}|\cdot \rc$ amount,  (iii)  to pay this extractor $ \dc-\chc$ amount taken from the dishonest extractor's deposit, and (iv) to pay the buyer and the rest of the clients the same way it does in step \ref{sim::case1-both-cheated-traitor-correct-res}. 
\item if one of the extractors cheated, there is a traitor, and the traitor delivered a correct result (in step \ref{sim::case1-traitor-cont}), then instructs the ledger (i) to return the other honest extractor's deposit (i.e., $\dc'$ amount), (ii) to pay the honest extractor $|S_{\st\cap}|\cdot \rc$ amount, taken from the buyer's deposit,  (iii) to pay the honest extractor $\dc- \chc$ amount, taken from the traitor's deposit,   (iv) to pay to the traitor $|S_{\st\cap}|\cdot \rc$ amount, taken from the buyer’s deposit,  (v) to refund the traitor $\dc'-\dc$ amount, (vi) to refund the traitor $\chc$ amount,  (vii) to take $|S_{\st\cap}|\cdot m\cdot \lc$ amount from the buyer's deposit and distribute it among all clients, except the buyer, and (viii) to return $(\Smin-|S_{\scriptscriptstyle\cap}|)\cdot \vc$ amount back  to the buyer. 
\item if one of the extractors cheated, there is a traitor, and the traitor delivered an incorrect result (in step \ref{sim::case1-traitor-cont}), then instructs the ledger (i) to pay the honest extractor the same way it does in step \ref{one-cheated-exists-traitor-honest-traitor}, (ii) to refund the traitor  $\chc$ amount, and (iii)  to pay the buyer and the rest of the clients in the same way it does in step \ref{sim::case1-both-cheated-traitor-correct-res}.
\item outputs whatever $\mathcal{A}$ outputs and then halts.
\end{enumerate}
\end{enumerate}

Next, we show that the real and ideal models are computationally indistinguishable. We first focus on the adversary’s output. The addresses of the smart contracts have identical distribution in both models. In the real and ideal models, the adversary sees the transcripts of ideal calls to $f_{\st \ct}$ as well as the functionality outputs ($mk,  mk'$). Due to the security of \ct (as we are in the $f_{\st \ct}$-hybrid world), the transcripts of $f_{\st \ct}$ in both models have identical distribution, so have the random outputs of $f_{\st \ct}$, i.e., ($mk,  mk'$). Also, the deposit amounts $\Smin\cdot \vc$ and $\wc$ have identical distributions in both models. Due to the semantical security of the public key encryption, the ciphertext $ct_{\st mk}$ in the real model is computationally indistinguishable from the ciphertext $ct_{\st {mk}}$ in the ideal model. Due to the hiding property of the commitment scheme, commitment $com_{\st mk}$ in the real model is computationally indistinguishable from commitment $com_{\st {mk}}$ in the ideal model.  
Moreover, due to the security of \fpsi, all transcripts and outputs produced in the ideal model, in step \ref{sim::case1-F-PSI} above, have identical distribution to the corresponding transcripts and outputs produced in \fpsi in the real model.  
The address of \SCtc has the same distribution in both models. The amounts each party receives in the real and ideal models are the same, except when both extractors produce an identical and incorrect result (i.e., intersection) in the real model, as we will shortly discuss, this would not occur under the assumption that the extractors are rational and due to the security of the counter collusion smart contracts.

Now, we show that an honest party aborts with the same probability in the real and ideal models. As before, for the sake of completeness, we include the \fpsi in the following discussion as well.  Due to the security of \ct, an honest party, during \ct invocation, aborts with the same probability in both models; in this case, the adversary learns nothing about the parties' input set and the sets' intersection as the parties have not sent out any encoded input set yet. In both models, an honest party can read the smart contract and check if sufficient amounts of coins have been deposited. Thus, it would halt with the same probability in both models.  If the parties halt because of insufficient amounts of deposit, no one could learn about (i) the parties' input set and (ii) the sets' intersection because the inputs (representation) have not been dispatched at this point.  Due to the security of \zspaa, an honest party during \zspaa execution aborts with the same probability in both models.  In this case, an aborting adversary also learns nothing about the parties' input set and the sets' intersection. %Since all parties' deposit is public, an honest party can read from the smart contract and detect if not all parties have deposited a sufficient amount with the same probability in both models. 

Due to the security of \vopr, honest parties abort with the same probability in both models. In the case where a party aborts during the execution of \vopr, the adversary would learn nothing (i) about its counter party's input set, and (ii) about the rest of the honest parties' input sets and the intersection as the other parties' input sets remain blinded by random blinding factors known only to client $D$. In the real model, client $D$ can check if all parties provided their encoded inputs via reading the state of the smart contract.  The simulator can perform the same check to ensure  $\mathcal{A}'$ has provided the encoded inputs of all corrupt parties. So, in both models, an honest party with the same probability detects if not all encoded inputs have been provided. In this case, if an adversary aborts and does not provide its encoded inputs (i.e., polynomials ${\bm\nu}^{ \st {  {(C'')}} }$), then it learns nothing about the honest parties' input sets and the intersection, for the same reason explained above.

In the real model,  the contract sums every client $C$'s polynomial $\bm\nu^{ \st {  {(C)}} }$ with each other and with client $D$'s polynomial $\bm\nu^{ \st {  {(D)}} }$, that ultimately removes the blinding factors that  $D$ initially inserted (during the \vopr execution), and then checks if the result is divisible by  $\bm \zeta$. Due to (a) Theorem \ref{Unforgeable-Polynomials-Linear-Combination}, (b) the fact that the smart contract is given the random polynomial $\bm \zeta$ in plaintext, (c) no party (except honest $D$) knew polynomial $\bm \zeta$ before they send their input to the contract, and (d) the security of the contract (i.e., the adversary cannot influence the correctness of the smart contract's verifications), the contract can detect if a set of outputs of \vopr were tampered with, with a probability at least $1-\negl(\lambda)$. In the ideal model, $\mathsf{Sim}^{\st \epsi}_{\st A}$ (in step \ref{sim::case1-sim-check-res}) can remove the blinding factors and it knows the random polynomial ${\bm \zeta}$. So, $\mathsf{Sim}^{\st \epsi}_{\st A}$ can detect when $\mathcal{A}'$ tampers with a set of the outputs of \vopr (sent to  $\mathsf{Sim}^{\st \epsi}_{\st A}$) with a probability at least $1-\negl(\lambda)$,  due to Theorem \ref{Unforgeable-Polynomials-Linear-Combination}. Therefore, the smart contract in the real model and the simulator in the ideal model would abort with a similar probability.

Due to the security of \zspaa, the probability that in the real model an invalid $k_{\st i}\in \vv{k}$ is appended to $ L$ is similar to the probability that  $\mathsf{Sim}^{\st \epsi}_{\st A}$ detects an invalid $ k'_{\st i}\in \vv{k'}$ in the ideal model. In the real model, when $Flag=False$, the smart contract can identify each ill-structured output of \vopr (i.e.,  $\bm\nu^{ \st {  {(C)}} }$) with a probability of at least $1-\negl(\lambda)$ by checking whether $\bm\zeta$  divides $\bm\iota^{ \st {  {(C)}}}$, due to  (a) Theorem \ref{proof::unforgeable-poly} (i.e., unforgeable polynomial), (b) the fact that the smart contract is given $\bm \zeta$ in plaintext, (c) no party (except honest client $D$) knew anything about $\bm \zeta$ before they send their input to the contract, and (d) the security of the contract.  
In the ideal model, when $Flag=False$, given each ${\bm\nu}^{ \st {  {(C'')}} }$, $\mathsf{Sim}^{\st \epsi}_{\st A}$ can remove its blinding factors from  ${\bm\nu}^{ \st {  {(C'')}} }$ which results in $ {\bm\phi'}^{ \st {  {(C'')}} }$ and then can check if ${\bm\zeta}$  divides $ {\bm\phi'}^{ \st {  {(C'')}} }$, in step \ref{sim::case1-check-final-res-for-each-client}. $\mathsf{Sim}^{\st \epsi}_{\st A}$ can detect an ill-structured  ${\bm\nu}^{ \st {  {(C'')}} }$ with a probability of at least $1-\negl(\lambda)$, due to Theorem \ref{proof::unforgeable-poly}, the fact that the simulator is given ${\bm \zeta}$ in plaintext,  and the adversary is not given any knowledge about ${\bm \zeta}$ before it sends to the simulator the outputs of \vopr.  Therefore, the smart contract in the real model and $\mathsf{Sim}^{\st \epsi}_{\st A}$ in the ideal model can detect an ill-structured input of an adversary with the same probability.  
The smart contract in the real model and $\mathsf{Sim}^{\st \epsi}_{\st A}$ in the ideal model can detect and abort with the same probability if the adversary provides an invalid opening to each commitment $com_{\st i,j}$ and $com_{\st {mk}}$, due to the binding property of the commitment scheme. Also, the smart contract in the real model and $\mathsf{Sim}^{\st \epsi}_{\st A}$ in the ideal model, can abort with the same probability if a Merkle tree proof is invalid, due to the security of the Merkle tree, i.e., due to the collision resistance of Merkle tree's hash function.

Note that in the ideal model, $\mathsf{Sim}^{\st \epsi}_{\st A}$ can detect and abort with a probability of $1$, if $\mathcal{A}_{\st I}$ does not send to the simulator all encoded encrypted elements of the intersection,  i.e., when $E^{\st (I)}\neq S'$. Because the simulator already knows all elements in the intersection (and the encryption key). Thus, it can detect with a probability of $1$ if both the intersection sets that the extractors provide are identical but incorrect.   
 In the real world, if the extractors collude with each other and provide identical but incorrect intersections,   then an honest client (or the smart contract) cannot detect it. Thus, the adversary can distinguish the two models, based on the probability of aborting.  However, under the assumption that the smart contracts (of Dong \textit{et al.} \cite{dong2017betrayal}) are secure (i.e., are counter-collusion), and the extractors are rational, such an event (i.e., providing identical but incorrect result without one extractor betraying the other) would not occur in either model, as the real model and $(\mathcal{A}_{\st 1}, \mathcal{A}_{\st 2})$ rational adversaries follow the strategy that leads to a higher payoff. Specifically, as shown in \cite{dong2017betrayal}, providing incorrect but identical results is not the preferred strategy of the extractors; instead, the betrayal of one extractor by the other is the most profitable strategy in the case of (enforceable) collusion between the two extractors. This also implies that the amounts that the extractors would receive in both models are identical.

%$\ddot Q:=(\qinit,  \qdelwr, \qUnFAbtwr, \qFAbt)$

Now, we analyse the output of the predicates $\bar Q:=(\qinit,  \qdelwr, \qUnFAbtwr, \qFAbt)$ in the real and ideal models. In the real model, all clients proceed to prepare their input set only if the predefined amount of coins have been deposited by the parties; otherwise (if in steps \ref{sim::case1-buyer-deposit},\ref{e-PSI::buyer-deposit},\ref{e-PSI::extractor-deposit} of \epsi and step \ref{F-PSI::each-client-deposit} of \fpsi there is not enough deposit), they will be refunded and the protocol halts. In the ideal model, the simulator proceeds to prepare its inputs only if enough deposit has been placed in the contract. Otherwise, it would send message $abort_{\st 1}$ to TTP, during steps \ref{sim::case1-buyer-deposit}--\ref{sim::case1-check-Adv-deposited-y.|G|}. Thus, in both models, the parties proceed to prepare their inputs only if $\qinit(.) \rightarrow1$.  
In the real model, if there is an abort after the parties ensure there is enough deposit and before client $D$ provides its encoded input to the contract, then all parties can retrieve their deposit in full. In this case, the aborting adversary cannot learn anything about honest parties' input sets, as the parties' input sets have been blinded by random blinding polynomials known only to client $D$. In the ideal model, if there is any abort during steps \ref{E-PSI::sim-A-first-VOPR-invocation}--\ref{E-PSI::sim-A-receive-nu-from-adv}, then the simulator sends $abort_{\st 2}$ to TTP and instructs the ledger to refund every party's deposit. In the case of an abort, within the above two steps, the auditor is not involved, and paid. Therefore, in both models,  in the case of an abort within the above steps, we would have $\qFAbt(.)\rightarrow1$.

In the real model, if $Flag=True$, then all parties can locally extract the intersection, regardless of the extractors' behaviour. In this case, each honest party receives $\yc'$ amount that it initially deposited in $\mathcal{SC}_{\fpsi}$. Moreover, each honest party receives \emph{at least} $|S_{\st\cap}|\cdot \lc$ amount as a reward, for contributing to the result. In this case, the honest buyer always collects the leftover of its deposit. Specifically, if both extractors act honestly, and the intersection cardinality is smaller than $|\Smin|$, then the buyer collects its deposit leftover, after paying all honest parties. If any extractor misbehaves, then the honest buyer fully recovers its deposit (and the misbehaving extractor pays the rest). Even in the case that an extractor misbehaves and then becomes a traitor to correct its past misbehaviour, the buyer collects its deposit leftover if the intersection cardinality is smaller than $|\Smin|$.  In the ideal model, when $Flag=True$, then $\mathsf{Sim}^{\st \epsi}_{\st A}$ can extract the intersection by summing the output of \vopr provided by all parties and removing the blinding polynomials. In this case, it sends back each party's deposit placed in $\mathcal{SC}_{\fpsi}$, i.e., $\yc'$ amount. Also, in this case, each honest party receives at least $|S_{\st\cap}|\cdot \lc$ amount as a reward and the honest buyer always collects the leftover of its deposit. Thus, in both models in the case of $Flag=True$, we would have $\qdelwr(.)\rightarrow 1$.

In the real model, when $Flag=False$, only the adversary can learn the result. In this case, the contract sends (i) $\chc$ amount to \aud, and (ii) $\frac{m'\cdot \yc'-\chc}{m-m'}$ amount, as compensation and reward, to each honest party, in addition to each party's initial deposit. In the ideal model,  when $Flag=False$, $\mathsf{Sim}^{\st \epsi}_{\st A}$ sends $abort_{\st 3}$ to TTP and instructs the ledger to distribute the same amount the contract distributes among the auditor (e.g., with address $adr_{\st j}$) and every honest party (e.g., with address $adr_{\st i}$) in the real model.  Thus, in both models when $Flag=False$, we would have $\qUnFAbtwr(., ., ., ., adr_{\st i})\rightarrow (a=1, .)$ and  $\qUnFAbtwr(., ., ., ., adr_{\st j})\rightarrow (., b=1)$. 

We conclude that the distribution of the joint outputs of the honest client $C\in \hat G$, client $D$, \aud, and the adversary in the real and ideal models are computationally indistinguishable.

\

\noindent\textbf{Case 2: Corrupt dealer $D$}.  In the real execution, the dealer's view is defined as follows:

$$ \mathsf{View}_{\st D}^{\st \epsi} \Big(S^{\st (D)}, (S^{\st (1)},..., S^{\st (m)})\Big)=$$ $$ \{S^{\st (D)}, adr_{\st sc}, S_{\minn}\cdot \vc, 2\cdot \dc', 
 r_{\st D},  \mathsf{View}^{\st \fpsi}_{\st D}, \mathsf{View}^{\st \ct}_{\st D}, (com_{\st 1,j}, \bar{e}_{\st 1}, q_{\st 1}, h_{\st 1})..., (com_{\st sz, j'}, \bar{e}_{\st sz}, q_{\st {sz}}, h_{\st sz}), g, \hat{x}:=({mk}, z'), 
  S_{\st \cap}\}$$
where  $\mathsf{View}^{\st \ct}_{\st D}$ and $\mathsf{View}^{\st \vopr}_{\st D}$ refer to $D$'s real-model view during the execution of \ct and \vopr respectively. Also, $r_{\st D}$ is the outcome of internal random coins of $D$, $adr_{\st sc}$ is the address of contract, $\mathcal{SC}_{\epsi}$, $(j, ...,j')\in \{1,..., h\}$, $z'=\mathtt{PRF}(\bar{mk}, 0)$, $sz=|S_{\st \cap}|$, and $h_{\st i}$ is a Merkle tree proof asserting that $com_{\st i,j}$ is a leaf node of a Merkle tree with root node $g$. The simulator $\mathsf{Sim}^{\st \epsi}_{\st D}$, which receives all parties' input sets, works as follows. 

\begin{enumerate}

%\item receives from the subroutine adversary polynomials $\bar{\bm\zeta}, (\bar{\bm\gamma}^{\st(A_1)}, \bar{\bm\delta}^{\st (A_1)}),..., (\bar{\bm\gamma}^{\st(A_m)}, \bar{\bm\delta}^{\st (A_m)})$,  $(\bar{\bm\omega}'^{\st (A_1)}, \bar{\bm\rho}'^{\st (A_1)}),..., $ $(\bar{\bm\omega}'^{\st (A_m)}, \bar{\bm\rho}'^{\st (A_m)})$, where $deg(\bar{\bm\gamma}^{\st(C)})=deg(\bar{\bm\delta}^{\st(C)})=3d+1, deg(\bar{\bm\omega}'^{\st (C)}) =deg(\bar{\bm\rho}'^{\st (C)})= d$, and $deg(\bar{\bm\zeta})=1$, where $C\in  \{  {  A}_{ \st {   1}}, ...,   {  A}_{ \st {   m}}\} $.
%
\item generates an empty view. It appends to the view, the input set $S^{\st (D)}$. It constructs and deploys a smart contract. It appends the contract's address, $ {adr}_{\st sc}$, to the view.

\item appends to the view integer $\Smin\cdot \vc$, and $2\cdot \dc'$. Also, it appends uniformly random coins $r'_{\st D}$ to the view. 

\item extracts the simulation of \fpsi from \fpsi's simulator for client $D$. Let $\mathsf{Sim}^{\st \fpsi}_{\st D}$ be the simulation, that also includes a random key $mk$. It appends $\mathsf{Sim}^{\st \fpsi}_{\st D}$ to the view.

 \item extracts the simulation of \ct from \ct's simulator, yielding  the simulation $\mathsf{Sim}^{\st \ct}_{\st D}$ that includes its output $mk'$. It appends $\mathsf{Sim}^{\st \ct}_{\st D}$ to the view.

 \item\label{sim::case2-encode-intersection} encrypts each element $s_{\st i,j}$ in the intersection set $S_{\st \cap}$  as $e_{\st i,j}=\mathtt{PRP}(mk', s_{\st i,j})$ and then encodes the result as $\bar{e}_{\st i,j} = e_{\st i,j}||\mathtt{H}(e_{\st i, j})$. It commits to each encoded value as $com_{\st i,j}=\comcom(\bar{e}_{\st i,j}, q_{\st i,j})$, where $j$ is the index of the bin to which $\bar{e}_{\st i,j}$ belongs and  $q_{\st i, j}$ is a random value.

  \item\label{sim::case2-gen-commitments} It constructs $(com'_{\st 1,1}, ..., com'_{\st d, h})$ where each $com'_{\st i, j}$ is a value picked uniformly at random from the commitment scheme output range. For every $j$-th bin, it sets each  $com'_{\st i', j}$ to unique $com_{\st i,j}$ if value $com_{\st i,j}$ for $j$-th bin has been generated in step \ref{sim::case2-encode-intersection}. Otherwise, the original values of $com'_{\st i', j}$ remains unchanged.

% \item\label{sim::case2-gen-commitments} picks random values $(r_{\st 1,1}, ..., r_{\st d, h})$, from the $\mathtt{PRP}$'s range. It commits to them that yields $(com'_{\st 1,1}, ..., com'_{\st d, h})$ using fresh uniformly random values $q'=\{q'_{\st 1,1}, ..., q'_{\st d, h}\}$. For every $j$-th bin, it sets each  $com'_{\st i', j}$ to unique $com_{\st i,j}$ if value $com_{\st i,j}$ for $j$-th bin has been generated in step \ref{sim::case2-encode-intersection}. Otherwise, the original values of $com'_{\st i', j}$ remains unchanged. 
 
% \begin{itemize}
% 
%\item[$\bullet$] $r_{\st i, j}$ to a unique $\bar{e}_{\st ., j}$
%%
%\item[$\bullet$] $com'_{\st i, j}$ to unique $com_{\st .,j}$
%%
%\item[$\bullet$] $q'_{\st i, j}$ to unique $q_{\st ., j}$ 
%
% \end{itemize}
 
% Note the above allocation takes place only if values $\bar{e}_{\st ., j}$, $com_{\st .,j}$, and $q_{\st ., j}$ for $j$-th have been generated in step \ref{sim::case2-encode-intersection}. Otherwise, the original values of $r_{\st i, j}$, $com'_{\st i, j}$, and $q'_{\st i, j}$ remain unchanged. 

 \item constructs a  Merkle tree on top of the values $(com'_{\st 1,1}, ..., com'_{\st d, h})$ generated in step \ref{sim::case2-gen-commitments}. Let $g$ be the resulting tree's root.  
 \item for each element $s_{\st i,j}$ in the intersection, it constructs $(com'_{\st i',j}, \bar{e}_{\st i,j}, q_{\st i,j}, h_{\st i,j})$, where $com'_{\st i',j}$ is the commitment of $ \bar{e}_{\st i,j}$, $com'_{\st i',j}\in com$,  $(\bar{e}_{\st i,j}$ $q_{\st i,j})$ is the commitment's opening generated in step \ref{sim::case2-encode-intersection}, and $h_{\st i,j}$ is a Merkle tree proof asserting that $com'_{\st i',j}$ is a leaf of a Merkle tree with root $g$. It appends all  $(com'_{\st i',j}, \bar{e}_{\st i,j}, q_{\st i,j}, h_{\st i,j})$ along with $g$ to the view. 

\item generates a commitment to $mk$ as $com_{\st mk}=\comcom(mk, z')$ where $z'=\mathtt{PRF}(mk, 0)$. It appends $(mk, z')$ along with $S_{\st \cap}$ to the view. 

\end{enumerate}

 Now, we will discuss why the two views are computationally indistinguishable.  $D$'s input $S^{\st (D)}$ is
 identical in both models, so they have identical distributions.  The contract's address has the same distribution in both models. The same holds for the integers $\Smin\cdot \vc$ and $ 2\cdot \dc'$. Also, because the real-model semi-honest adversary samples its randomness according to the protocol’s description, the random coins in both models (i.e., $r_{\st D}$  and $r'_{\st D}$) have identical distribution. Due to the security of \fpsi, $\mathsf{View}^{\st \fpsi}_{\st D}$ and $\mathsf{Sim}^{\st \fpsi}_{\st D}$ have identical distribution, so do  $\mathsf{View}^{\st \ct}_{\st D}$ and $\mathsf{Sim}^{\st \ct}_{\st D}$ due to the security of \ct. The intersection elements in both models have identical distributions and the encryption scheme is schematically secure. Therefore, each intersection element's representation (i.e., $\bar{e}_{\st i}$ in the real model and $\bar{e}_{\st i, j}$ in the deal model) are computationally indistinguishable. Each $q_{\st i}$ in the real model and $q_{\st i, j}$ in the ideal model have identical distributions as they have been picked uniformly at random. Each commitment $com_{\st i,j}$ in the real model is computationally indistinguishable from commitment $com'_{\st i',j}$ in the ideal model.

 In the real model, each Merkle tree proof $h_{\st i}$ contains two leaf nodes (along with intermediary nodes that are the hash values of a subset of leaf nodes) that are themselves the commitment values. Also, for each $h_{\st i}$, only one of the leaf node's openings (that contains an element in the intersection) is seen by $D$. The same holds in the ideal model, with the difference that for each Merkle tree proof $h_{\st i, j}$ the leaf node whose opening is not provided is a random value, instead of an actual commitment. However, due to the hiding property of the commitment scheme, in the real and ideal models,  these two leaf nodes (whose openings are not provided) and accordingly the two proofs are computationally indistinguishable. In both models, values $mk$ and $z'$ are random values, so they have identical distributions. Furthermore, the intersection $S_{\st \cap}$ is identical in both models. Thus, we conclude that the two views are computationally indistinguishable.

 \

\noindent\textbf{Case 3: Corrupt auditor}.  In this case, the real-model view of the auditor is defined as  

$$ \mathsf{View}_{\st Aud}^{\st \epsi} \Big(\bot, S^{\st (D)}, (S^{\st (1)},..., S^{\st (m)})\Big)=$$

$$ \{\mathsf{View}^{\st \fpsi}_{\st Aud}, adr_{\st sc}, \Smin\cdot \vc, 2\cdot \dc', (com_{\st 1,j}, \bar{e}_{\st 1}, q_{\st 1}, h_{\st 1})..., (com_{\st sz, j'}, \bar{e}_{\st sz}, q_{\st {sz}}, h_{\st sz}), g, \hat{x}:=({mk}, z')\}$$

  Due to the security of \fpsi, \aud's view $\mathsf{View}^{\st \fpsi}_{\st Aud}$ during the execution of \fpsi can be easily simulated. As we have shown in Case 2, for the remaining transcript, its real-model view can be simulated too. However, there is a difference between Case 3 and Case 2; namely, in the former case, \aud does not have the $\mathtt{PRP}$'s key $mk'$ used to encrypt each set element. However, due to the security of  $\mathtt{PRP}$, it cannot distinguish each encrypted encoded element from a uniformly random element and cannot distinguish $\mathtt{PRP}(mk',.)$ from a uniform permutation. Therefore, each value $\bar{e}_{\st j}$ in the real model has identical distribution to each value $\bar{e}_{\st i,j}$ (as defined in Case 2) in the ideal model, as both are the outputs of $\mathtt{PRP}$.

%by using the proof that we have already provided for Case 1 (i.e.,  $m-1$ client $A_{\st j}$s are corrupt), we can easily construct a simulator that generates a view computationally distinguishable from the real world semi-honest auditor. 
% The reason is that, in the worst-case scenario where $m-1$ malicious client $A_{\st j}$s reveal their input sets and randomness to the auditor, the auditor's view would be similar to the view of these corrupt clients, which we have shown to be indistinguishable. The only extra messages the auditor generates, that a corrupt client $A_{\st j}$ would not see in plaintext, are random blinding polynomials $(\bm\xi^{\st (A_1)},..., \bm\xi^{\st (A_m)})$ generated during the execution of $\mathtt{Audit}(.)$ of the ZSPA-A; however, these polynomials are picked uniformly at random and independent of the parties' input sets. Thus, if the smart contract detects misbehaviour and invokes the auditor, even if $m-1$ corrupt client $A_{\st j}$ reveals their input sets, then the auditor cannot learn anything about honest parties' input sets.    
 
 \

\noindent\textbf{Case 4: Corrupt public}. In the real model, the view of the public (i.e., non-participants of the protocol) is defined as below:

 $$ \mathsf{View}_{\st Pub}^{\st \epsi} \Big(\bot, S^{\st (D)},(S^{\st (A_1)},..., S^{\st (A_m)})\Big)=$$ $$ 
\{\mathsf{View}^{\st \fpsi}_{\st Pub}, adr_{\st sc}, \Smin\cdot \vc, 2\cdot \dc', (com_{\st 1,j}, \bar{e}_{\st 1}, q_{\st 1}, h_{\st 1})..., (com_{\st sz, j'}, \bar{e}_{\st sz}, q_{\st {sz}}, h_{\st sz}), g, \hat{x}:=({mk}, z')\}$$

  Due to the security of \fpsi, the public's view $\mathsf{View}^{\st \fpsi}_{\st Pub}$ during \fpsi's execution can be simulated in the same way which is done in Case 4, in Section \ref{sec::F-PSI-proof}. The rest of the public's view overlaps with \aud's view in Case 3, excluding $\mathsf{View}^{\st \fpsi}_{\st Aud}$. Therefore, we can use the argument provided in Case 3 to show that the rest of the public's view can be simulated.     
   We  conclude that the public's real and ideal views are computationally indistinguishable. 
  \hfill\(\Box\)\end{proof}

% !TEX root =main.tex

\section{Evaluation}\label{sec::valuation}
In this section, we analyse the asymptotic costs of \epsi. We also compare its costs and features with the fastest two and multiple parties PSIs in \cite{AbadiDMT22,DBLP:conf/ccs/KolesnikovMPRT17,NevoTY21,RaghuramanR22}) and with the fair PSIs in \cite{DebnathD16,DBLP:conf/dbsec/DongCCR13}. Tables \ref{table::Asymptotic-Cost} and \ref{table::comparisonTable} summarise the result of the cost analysis and the comparison respectively.

% !TEX root =main.tex

 \begin{table}[!htb]

\caption{ {\small{Asymptotic costs of different parties in \epsi. In the table, $h$ is the total number of bins, $d$ is a bin's capacity (i.e., $d=100$), $m$ is the total number of clients (excluding $D$), $|S|$ is a set cardinality, and $\bar\xi$ is \ole's security parameter.
}}} \label{table::Asymptotic-Cost} 
% \vspace{-3mm}
\begin{center}
\renewcommand{\arraystretch}{1.2}
\begin{tabular}{|c|c|c|c|c|} 

   %\hline
        \cline{1-3}  
{\scriptsize {Party}}&{\scriptsize {Computation Cost}}&{\scriptsize {Communication Cost}}\\
     \cline{1-3}  
%&\scriptsize$e=1$&\scriptsize$e>1$\\
\hline

    %SO-PoR 1st row
\scriptsize Client $A_{\st  3},...,    A_{\st   m}$& \cellcolor{gray!50}   \scriptsize$O\Big(h\cdot d(m+d)\Big)$& \cellcolor{gray!50}  \scriptsize$O\Big(h\cdot d^{\st 2}\cdot \bar\xi\Big)$\\
 %  { }
     \cline{1-3}  
     %SO-PoR 2nd row
\scriptsize Dealer $D$&   \cellcolor{gray!20}\scriptsize$O\Big(h\cdot m(d^{\st 2}+d)\Big)$ &  \cellcolor{gray!20}\scriptsize$O\Big(h\cdot d^{\st 2}\cdot \bar\xi\cdot m\Big)$\\
      \cline{1-3}

       %[3] 1st row 
       
   \scriptsize   {Auditor $\aud$ }& \cellcolor{gray!50}\scriptsize$O\Big(h\cdot m\cdot d\Big)$&  \cellcolor{gray!50}\scriptsize$O\Big(h\cdot d\Big)$\\      
            \cline{1-3} 

 % \scriptsize \ \ \ \ \ \ \ \ --------------&&\\
 \scriptsize{Extractor} $A_{\st  1},    A_{\st   2}$& \cellcolor{gray!20}\scriptsize$O\Big(h\cdot d(m+d)\Big)$& \cellcolor{gray!20}\scriptsize$O\Big(|S_{\scriptscriptstyle\cap}|\cdot \log_{\st 2}|S|\Big)$\\
     \cline{2-3}
%{\scriptsize Auditor $\mathcal{D}_{\st n}$}&    \cellcolor{gray!20}\scriptsize$O(\sum\limits_{\st i=e}^{\st n}\frac{n!}{i!(n- i)!})$&    \cellcolor{gray!20}\scriptsize$ O(\sum\limits_{\st i=e}^{\st n}\frac{n!}{i!(n- i)!})$\\
     \cline{1-3}  
     
 \scriptsize Smart contract $\mathcal{SC}_{\epsi}$\ \&\ $\mathcal{SC}_{\fpsi}$& \cellcolor{gray!50}\scriptsize $O\Big( |S_{\st \cap}|(d+ \log_{\st 2} |S|)+h\cdot m\cdot d\Big)$& \cellcolor{gray!50}\scriptsize ---\\
 
   \hline
   
   \hline

     \scriptsize Overall Complexity & \cellcolor{gray!20}\scriptsize $O\Big(h\cdot d^{2}\cdot m \Big)=O\Big(|S|\cdot d\cdot m\Big)$& \cellcolor{gray!20}\scriptsize {$O\Big(h\cdot d^{\st 2}\cdot \bar\xi\cdot m\Big)$}\\
     
      \cline{1-3}  

\end{tabular}  
\end{center}
\end{table}

%!TEX root = main.tex

\begin{table} 

\begin{center}
\caption{ \small{Comparison of the asymptotic complexities and features of state-of-the-art PSIs. In the table, $t$ is a parameter that determines the maximum number of colluding parties and $\kappa$ is a security parameter.}}  \label{table::comparisonTable} 
\begin{tabular}{|c|c|c|c|c|c|c|c|c|c|} 
\hline

%\multicolumn{3}{c|}

%\multirow{2}{*} {\scriptsize {Schemes}} &{\scriptsize {Computation}}& \scriptsize{Communication}&{\scriptsize{Fairness}}&{ \scriptsize Rewarding}& {\scriptsize{ Sym-key based}}& {\scriptsize{Multi-party}}&\scriptsize Active Adversary \\
%\hline

\multirow{2}{*} {\scriptsize {Schemes}} &\multicolumn{2}{c|}{\scriptsize Asymptotic Cost}&\multicolumn{5}{c|}{\scriptsize{Features}} \\

\cline{2-8}

& \scriptsize{Computation}&\scriptsize{Communication}&{\scriptsize{Fairness}}&{ \scriptsize Rewarding}& {\scriptsize{ Sym-key based}}& {\scriptsize{Multi-party}}&\scriptsize Active Adversary\\

\hline 

%&\scriptsize {Modular expo.}&\cellcolor{gray!20}\scriptsize {$0$}&\cellcolor{gray!20}\scriptsize$5$&\cellcolor{gray!20}\scriptsize$12$\\

\scriptsize  \scriptsize{ \cite{AbadiDMT22}}&\cellcolor{gray!20}\scriptsize{$O( h\cdot d^{\st 2}\cdot m)$}&\cellcolor{gray!20}\scriptsize$O(h\cdot d\cdot m)$&\cellcolor{gray!20}\scriptsize\textcolor{red}{$\times$}&\cellcolor{gray!20}\scriptsize\textcolor{red}{$\times$}&\cellcolor{gray!20}\scriptsize\textcolor{blue}\checkmark  &\cellcolor{gray!20}\scriptsize\textcolor{blue}\checkmark&\cellcolor{gray!20}\scriptsize\textcolor{red}{$\times$} \\

\hline

\scriptsize \cite{DebnathD16}&\cellcolor{gray!50}\scriptsize{$O(|S|)$}&\cellcolor{gray!50}\scriptsize{$O(|S|)$}&\cellcolor{gray!50}\scriptsize\textcolor{blue}\checkmark&\cellcolor{gray!50}\scriptsize\textcolor{red}{$\times$}&\cellcolor{gray!50}\scriptsize\textcolor{red}{$\times$} &\cellcolor{gray!50}\scriptsize\textcolor{red}{$\times$}&\cellcolor{gray!50}\scriptsize\textcolor{blue}\checkmark \\

\hline

\scriptsize {\cite{DBLP:conf/dbsec/DongCCR13}}&\cellcolor{gray!20}\scriptsize{$O(|S|^{\st 2}$)}&\cellcolor{gray!20}\scriptsize$O(|S|)$&\cellcolor{gray!20}\scriptsize\textcolor{blue}\checkmark&\cellcolor{gray!20}\scriptsize\textcolor{red}{$\times$}  &\cellcolor{gray!20}\scriptsize\textcolor{red}{$\times$} &\cellcolor{gray!20}\scriptsize\textcolor{red}{$\times$}&\cellcolor{gray!20} \scriptsize\textcolor{blue}\checkmark\\ 

\hline
\scriptsize \cite{DBLP:conf/ccs/KolesnikovMPRT17}   &\cellcolor{gray!50}\scriptsize{$O(|S|\cdot m^{\st 2}+|S|\cdot m )$}&\cellcolor{gray!50}\scriptsize$O(|S|\cdot m^{\st 2})$&\cellcolor{gray!50}\scriptsize\textcolor{red}{$\times$}&\cellcolor{gray!50}\scriptsize\textcolor{red}{$\times$}  &\cellcolor{gray!50}\scriptsize\textcolor{blue}\checkmark &\cellcolor{gray!50}\scriptsize\textcolor{blue}\checkmark&\cellcolor{gray!50}\scriptsize\textcolor{red}{$\times$}\\ 

\hline

\scriptsize \cite{NevoTY21}&\cellcolor{gray!20}\scriptsize{$O(|S|\cdot \kappa(m+t^{\st 2}-t(m+1)))$}&\cellcolor{gray!20}\scriptsize{$O(|S|\cdot m\cdot \kappa)$}&\cellcolor{gray!20}\scriptsize{\textcolor{red}{$\times$}}&\cellcolor{gray!20}\scriptsize\textcolor{red}{$\times$}&\cellcolor{gray!20}\scriptsize\textcolor{blue}\checkmark  &\cellcolor{gray!20}\scriptsize\textcolor{blue}\checkmark&\cellcolor{gray!20}\scriptsize\textcolor{blue}\checkmark\\ 

\hline

\scriptsize \cite{RaghuramanR22}&\cellcolor{gray!50}\scriptsize{$O(|S|)$}&\cellcolor{gray!50}\scriptsize{$O(|S|\cdot \kappa)$}&\cellcolor{gray!50}\scriptsize{\textcolor{red}{$\times$}}&\cellcolor{gray!50}\scriptsize\textcolor{red}{$\times$} &\cellcolor{gray!50}\scriptsize\textcolor{blue}\checkmark &\cellcolor{gray!50}\scriptsize{\textcolor{red}{$\times$}} &\cellcolor{gray!50}\scriptsize\textcolor{blue}\checkmark\\ 

\hline

{\scriptsize \textbf{Ours:} \epsi}&\cellcolor{gray!20}\scriptsize{$O (h\cdot d^{2}\cdot m)$}&\cellcolor{gray!20}\scriptsize$O (h\cdot d^{\st 2}\cdot \bar\xi\cdot m )$&\cellcolor{gray!20}\scriptsize\textcolor{blue}\checkmark&\cellcolor{gray!20}\scriptsize \textcolor{blue}\checkmark&\cellcolor{gray!20}\scriptsize\textcolor{blue}\checkmark &\cellcolor{gray!20}\scriptsize\textcolor{blue}\checkmark&\cellcolor{gray!20}\scriptsize\textcolor{blue}\checkmark \\

\hline 

\end{tabular}
%}
%\renewcommand{\arraystretch}{1}
%\end{footnotesize}
\end{center}
%}
\end{table}

\subsection{Computation Cost}

\subsubsection{Client's and Dealer's Costs.}

In step \ref{e-psi::call-F-PSI-stepOne}, the cost of each client (including dealer $D$) is $O(m)$ and mainly involves an invocation of \ct. 
In steps \ref{e-psi::deploy-SC-E-PSI}--\ref{e-PSI::extractor-deposit}, the clients' cost is negligible as it involves deploying smart contracts and reading from the deployed contracts. 
In step \ref{e-psi::gen-mk-prime}, the clients' cost is  $O(m)$, as they need to invoke an instance of \ct. 
In step \ref{Smart-PSI:encode-elem}, each client invokes \prp and $\mathtt{H}$ linear with its set's cardinality. In the same step, it also constructs $h$ polynomials, where the construction of each polynomial involves $d$  modular multiplications and additions. Thus, its complexity in this step is $O(h\cdot d)$. It has been shown in \cite{AbadiDMT22} that $O(h\cdot d)=O(|S|)$ and  $d=100$ for all set sizes. 

In step \ref{e-psi::invoke-remainer-F-PSI}, each client   $A_{\st  1},...,    A_{\st   m}$ (excluding $D$): (i) invokes an instance of \zspaa which involves $O(h\cdot m)$ invocation of \ct, $3h\cdot m (d+1)$ invocation of \prf, $3h\cdot m (d+1)$ addition, and $O(h\cdot m\cdot d)$ invocation of $\mathtt{H}$ (in step \ref{ZSPA} of subroutine \fpsi), (ii) invokes $2h$ instances of \vopr, where each \vopr invocation involves $2d(1+d)$ invocations of $\ole^{\st +}$, multiplications, and additions  (in steps \ref{e-psi::D-randomises} and \ref{e-psi::C-randomises} of \fpsi), and (iii) performs $h(3d+2)$ modular addition (in step \ref{blindPoly-C-sends-to-contract} of  \fpsi).  
 Also, if $Flag=True$, each client (including $D$) invokes $\prf$ $h (3d+1)$ times, and performs $h (3d+1)$ additions, and performs polynomial evaluations linear with the total number of bins $h$, where each evaluation involves  $O(d)$  additions and $O(d)$ multiplications (in step \ref{F-PSI::flag-is-true} of \fpsi). 
 
 Step \ref{e-psi::commit-to-mk} involves only $D$ whose cost in this step is constant, as it involves invoking a public key encryption, \prf,  and commitment only once. Furthermore, $D$:  (a) invokes $2h\cdot m$ instances of \vopr  (in steps \ref{e-psi::D-randomises} and \ref{e-psi::C-randomises} of \fpsi), (b) invokes $\prf$ $h(3d+1)$ times (in step \ref{f-psi::D-gen-random-poly} of \fpsi), and (c) performs $h(d^{\st 2}+1)$ multiplications and $3h\cdot m\cdot d$ additions (in step \ref{f-psi::D-gen-switching-poly} of \fpsi). If $Flag=False$, then $D$ performs $O(h\cdot m\cdot d)$ multiplications and additions (in step \ref{F-PSI::flag-is-false} of \fpsi).

 \subsubsection{Auditor's Cost.}
 
If $Flag=False$, then \aud invokes $\prf$ $3h\cdot m(d+1)$ times  and  invokes $\mathtt{H}$ $O(h\cdot m\cdot d)$ times (in step \ref{F-PSI::flag-is-false} of \fpsi).

 \subsubsection{Extractor's Cost.}

 In step \ref{merkel-tree-cons}, each extractor invokes the commitment scheme linear with the number of its set cardinality $|S|$ and constructs a Merkle tree on top of the commitments. %Therefore, its complexity is $O(|S|)$.  
In step \ref{smart-PSI::extractors}, each extractor invokes $\mathtt{H}$ linear with its set cardinality $|S|$; it also performs polynomial evaluations linear with $|S|$.

 \subsubsection{Smart Contracts' Cost.}

In step \ref{e-psi::invoke-remainer-F-PSI}, the subroutine smart contract $\mathcal{SC}_{\fpsi}$ performs $h\cdot m(3d+1)$ additions and $h$ polynomial divisions,  where each division includes dividing a polynomial of degree $3d+1$ by a polynomial of degree $1$ (in step \ref{compute-res-poly} of \fpsi). In step \ref{e-psi::SC-verification--derive-mk}, $\mathcal{SC}_{\epsi}$ invokes the commitment's verification algorithm $\comver$ once,  $\mathtt{H}$ at most $|S_{\st \cap}|$ times, and $\prf$ $|S_{\st \cap}| (3d+1)$ times. In step \ref{e-psi::SC-verification--check-three-vals}, $\mathcal{SC}_{\epsi}$ invokes  $\comver$ at most $|S_{\st \cap}|$ times, and calls $\mathtt{H}$ $O(|S_{\st \cap}|\cdot \log_{\st 2} |S|)$ times. In the same step, it  performs polynomial evaluation linear with  $|S_{\st \cap}|$. Thus, its overall complexity is $O( |S_{\st \cap}|(d+ \log_{\st 2} |S|))$.
%

%In the same step, the subroutine smart contract $\mathcal{SC}_{\fpsi}$ performs $h\cdot m(3d+1)$ additions and $h$ polynomial divisions,  where each division includes dividing a polynomial of degree $3d+1$ by a polynomial of degree $1$ (in step \ref{compute-res-poly} of \fpsi). 

%Moreover, if $Flag=True$, then each client invokes $\prf$ $h (3d+1)$ times, and performs $h (3d+1)$ additions, and performs polynomial evaluations linear with its set cardinality, where each evaluation involves  $O(d)$  additions and $O(\frac{d^{\st 2}+d}{2})$ multiplications (in step \ref{F-PSI::flag-is-true} of \fpsi). If $Flag=False$, then (a) \aud invokes $\prf$ $3h\cdot m(d+1)$ times  and  invokes $\mathtt{H}$ $O(h\cdot m\cdot d)$ times, and (b) $D$ performs $O(h\cdot m\cdot d)$ multiplications and additions (in step \ref{F-PSI::flag-is-false} of \fpsi). 

%

  % (where each evaluation involves  $O(d)$  additions and $O(\frac{d^{\st 2}+d}{2})$ multiplications). 

%In step \ref{e-psi::SC-verification}, $\mathcal{SC}_{\epsi}$ invokes the commitment's verification algorithm, the hash function, at most linear with the intersection cardinality $|S_{\st \cap}|$, invokes $\prf$ $3d+1$ times, invokes the hash function $O(\log_{\st 2} (d\cdot m))$ times, and performs polynomial evaluations linear with the smallest set cardinality.% (where each evaluation involves  $O(d)$  additions and $O(\frac{d^{\st 2}+d}{2})$ multiplications). 

 %\scf  performs $h\cdot m(3d+1)$ modular additions and $h$ polynomial divisions (in step \ref{compute-res-poly} of F-PSI). 

\subsection{Communication Cost}

In steps  \ref{e-psi::call-F-PSI-stepOne} and \ref{e-psi::gen-mk-prime}, the communication cost of the clients is dominated by the cost of \ct which is $O(m)$. In steps \ref{e-psi::deploy-SC-E-PSI}--\ref{e-psi::commit-to-mk}, the clients' cost is negligible, as it involves sending a few transactions to the smart contracts, e.g., $\mathcal{SC}_{\fpsi}$, $\mathcal{SC}_{\epsi}$, and \SCpc. Step \ref{merkel-tree-cons} involves only extractors whose cost is $O(h)$ as each of them only sends to $\mathcal{SC}_{\epsi}$  a single value for each bin. In step \ref{e-psi::invoke-remainer-F-PSI}, the clients' cost is dominated by \vopr's cost; specifically, each pair of client and $D$ invokes \vopr $O(d^{\st 2})$ times for each bin; therefore, the cost of each client (excluding $D$) is $O(h\cdot d^{\st 2}\cdot \bar\xi)$ while the cost of $D$ is $O(h\cdot d^{\st 2}\cdot \bar\xi\cdot m)$, where $\bar\xi$ is the subroutine \ole's security parameter. 
Step \ref{smart-PSI::extractors}  involves only the extractors, where each extractor's cost is dominated by the size of the Merkle tree's proof it sends to $\mathcal{SC}_{\epsi}$, i.e., $O(|S_{\scriptscriptstyle\cap}|\cdot \log_{\st 2}|S|)$, where $|S|$ is the extractor's set cardinality. 
In step \ref{F-PSI::flag-is-false}, \aud sends $h$ polynomials of degree $3d+1$ to $\mathcal{SC}_{\fpsi}$; thus, its complexity is $O(h\cdot d)$. 
The rest of the steps impose negligible communication costs.

\subsection{Comparison}
Below we show that \epsi offers various features that the state-of-the-art PSIs do not offer simultaneously while keeping its overall overheads similar to the efficient PSIs.  

\subsubsection{Computation Complexity.} The  computation complexity  of \epsi is similar to that of PSI in \cite{AbadiDMT22}, but is better than the multiparty PSI's complexity in \cite{DBLP:conf/ccs/KolesnikovMPRT17} as  the latter's complexity is quadratic with the number of parties, i.e., $O(|S|\cdot d\cdot m)$ versus $O(|S|\cdot m^{\st 2}+|S|\cdot m )$. Also, \epsi's complexity  is better than the complexity of the PSI in  \cite{NevoTY21}  that is quadratic with parameter $t$. Similar to the two-party PSIs in \cite{DebnathD16,RaghuramanR22}, \epsi's complexity is linear with $|S|$.  The two-party PSI in \cite{DBLP:conf/dbsec/DongCCR13} imposes a higher computation overhead than \epsi does, as its complexity is quadratic with sets' cardinality. Hence, the complexity of \epsi is: (i) linear with the set cardinality, similar to the above schemes except the one in \cite{DBLP:conf/dbsec/DongCCR13} and (ii) linear with the total number of parties, similar to  the above multi-party schemes, except the one in \cite{DBLP:conf/ccs/KolesnikovMPRT17}. 
Hence, the computation complexity of \epsi is linear with the set cardinality and the number of parties, similar to the above schemes except for the ones in \cite{DBLP:conf/ccs/KolesnikovMPRT17,DBLP:conf/dbsec/DongCCR13} whose complexities are quadratic with the set cardinality or the number of parties.

\subsubsection{Communication Complexity.}  \epsi's communication complexity is slightly higher than the complexity of the PSI in \cite{AbadiDMT22}, by a factor of $d\cdot \bar\xi$. However, it is better than the  PSI's complexity in \cite{DBLP:conf/ccs/KolesnikovMPRT17} as the latter has a complexity quadratic with the number of parties. \epsi's complexity is slightly higher than the one in \cite{NevoTY21}, by a factor of $d$. Similar to the two-party PSIs in  \cite{DebnathD16,RaghuramanR22,DBLP:conf/dbsec/DongCCR13}, \epsi's complexity is linear with $c$. 
Therefore, the communication complexity of \epsi is linear with the set cardinality and number of parties, similar to the above schemes except the one in \cite{DBLP:conf/ccs/KolesnikovMPRT17} whose complexity is quadratic with the number of parties.

\subsubsection{Features.} \epsi is the only scheme that offers all the five features, i.e., supports fairness, rewards participants, is based on symmetric key primitives, supports multi-party, and is secure against active adversaries. After \epsi is the scheme in \cite{NevoTY21} which offers three of the above features. The rest of the schemes support only two of the above features. For the sake of fair comparison, we highlight that our \epsi and  \fpsi are the only PSIs that use smart contracts (that require additional but standard blockchain-related assumptions), whereas the rest of the above protocols do not use smart contracts.

%\input{Related-work}

 % !TEX root =main.tex

\section{Conclusion and Future Direction}\label{sec::concl}

Private Set Intersection (PSI) is a crucial protocol with numerous real-world applications. In this work, we proposed, \withFai, the first multi-party fair PSI that ensures that either all parties get the result or if the protocol aborts in an unfair manner, then honest parties will receive financial compensation. We then upgraded it to \withRew, the first PSI ensuring that honest parties who contribute their private sets receive a reward proportionate to the number of elements they reveal. Since an MPC that rewards participants for contributing their private inputs would help increase its real-world adoption, an interesting open question is:
\begin{center}
 \emph{How can we generalise the idea of rewarding participants to MPC?}
 \end{center}
 % !TEX root =main.tex

\section*{Acknowledgements}

Aydin Abadi was supported in part by REPHRAIN: The National Research Centre on Privacy, Harm Reduction and Adversarial Influence Online, under UKRI grant: EP/V011189/1. We would like to thank Yvo Desmedt for his suggestion of having a flexible reward mechanism. 

\bibliographystyle{splncs03}
\bibliography{ref}
\appendix
% !TEX root =Feather.tex

\section{Hash Tables}\label{Preliminary-Hash-Table}

We set the table's parameters appropriately to ensure the number of elements in each bin does not exceed a predefined capacity. Given the maximum number of elements $c$ and the bin's maximum size $d$, we can determine the number of bins by analysing hash tables under the balls into bins model  \cite{DBLP:conf/stoc/BerenbrinkCSV00}.
\vspace{-1mm}
\begin{theorem}\label{chernoff}\textbf{(Upper Tail in Chernoff Bounds)} Let $X_{\scriptscriptstyle i}$ be a random variable defined as $X_{\scriptscriptstyle i}=\sum\limits^{\scriptscriptstyle c}_{\scriptscriptstyle  i=1} Y_{\scriptscriptstyle i}$, where $Pr[Y_{\scriptscriptstyle i}=1]=p_{\scriptscriptstyle i}$, $Pr[Y_{\scriptscriptstyle i}=0]=1-p_{\scriptscriptstyle i}$, and all $Y_{\scriptscriptstyle i}$ are independent. Let the expectation be $\mu=\mathrm{E}[X_{\scriptscriptstyle i}]=\sum\limits ^{\scriptscriptstyle h}_{\scriptscriptstyle  i=1} p_{\scriptscriptstyle i}$, then 
$Pr[X_{\scriptscriptstyle i}>d=(1+\sigma)\cdot \mu]<\Big(\frac{e^{\scriptscriptstyle \sigma}}{(1+\sigma)^{\scriptscriptstyle (1+\sigma)}}\Big)^{\scriptscriptstyle \mu}, \forall \sigma>0$
\end{theorem}

In this model, the expectation is  $\mu=\frac{c}{h}$, where $c$ is the number of  balls and $h$ is the number of bins. The above inequality provides the probability that bin $i$ gets more than $(1+\sigma)\cdot \mu$ balls. Since there are $h$ bins, the probability that at least one of them is overloaded is bounded by the union bound:
\begin{equation}\label{equation:the-bound}
Pr[\exists i, X_{\scriptscriptstyle i}>d]\leq \sum\limits^{\scriptscriptstyle h}_{\scriptscriptstyle i=1}Pr[X_{\scriptscriptstyle i}>d] = h\cdot  \Big(\frac{e^{\scriptscriptstyle \sigma}}{(1+\sigma)^{\scriptscriptstyle (1+\sigma)}}\Big)^{\scriptscriptstyle \frac{c}{h}}
\end{equation}

Thus, for a hash table of length $h=O(c)$, there is always an  \emph{almost constant} expected number of elements, $d$,  mapped to the same bin  with a high probability \cite{DBLP:conf/ccs/PapamanthouTT08}, e.g., $1-2^{\scriptscriptstyle -40}$.
% !TEX root =main.tex

\section{Enhanced OLE's Ideal Functionality and Protocol}\label{apndx:F-OLE-plus}

The PSIs proposed in \cite{GhoshN19} use an enhanced version of the  \ole.  The enhanced \ole ensures that the receiver cannot learn anything about the sender's inputs,  in the case where it sets its input to $0$, i.e., $c=0$. The enhanced \ole's protocol (denoted by $\ole^{\st +}$) is presented in Figure \ref{fig:OLE-plus-protocol}. In this paper, $\ole^{\st +}$ is used as a subroutine in \vopr, in Section \ref{sec::vopr}.

\begin{figure}[ht]
\setlength{\fboxsep}{1pt}
\begin{center}
\begin{boxedminipage}{12.3cm}
\begin{small}
\begin{enumerate}
\item  Receiver (input $c \in \mathbb{F} $): Pick a random value, $r\stackrel{\st\$}\leftarrow  \mathbb{F} $, and send  $(\mathtt{inputS}, (c^{\st -1}, r))$ to the first $\mathcal{F}_{\st\ole}$.
%
%Upon receiving a message $(\mathtt{inputS},(a, b))$ from the sender, where $a, b \in \mathbb{F} $, verify that there is no  tuple stored; otherwise, ignore that message. Store $a$ and $b$ and send a message $(\mathtt{input})$ to the adversary, $\mathcal{A}$.
%
\item Sender (input $a, b \in \mathbb{F} $): Pick a random value, $u \stackrel{\st\$}\leftarrow  \mathbb{F} $, and send $(\mathtt{inputR}, u)$ to the first $\mathcal{F}_{\st\ole}$, to learn $t =  c^{\st -1}\cdot u
 + r$. Send $(\mathtt{inputS},(t + a, b - u))$ to the second $\mathcal{F}_{\st\ole}$.
\item Receiver: Send $(\mathtt{inputR}, c)$ to the second $\mathcal{F}_{\st\ole}$ and obtain $k = (t+a)\cdot c+(b-u)=a\cdot c + b + r\cdot c$. Output $s=k - r\cdot c=a\cdot c + b$.

%}
\end{enumerate}
\end{small}
\end{boxedminipage}
\end{center}
\caption{
\small {Enhanced Oblivious Linear function Evaluation  ($\ole^{\st +}$)  \cite{GhoshN19}}.} 
\label{fig:OLE-plus-protocol}
\end{figure}

% !TEX root =main.tex

\section{Counter Collusion Contracts}\label{appendix::Counter-Collusion-Contracts}

In this section, we present Prisoner’s Contract (\SCpc), Colluder’s Contract (\SCcc), Traitor’s Contract (\SCtc) originally proposed by Dong \textit{et al.} \cite{dong2017betrayal}. As we previously stated, we have slightly adjusted the contracts.  We will highlight the adjustments in blue. For the sake of completeness, below we restate the parameters used in these contracts and their relation.

\begin{itemize}
\item[$\bullet$] $\bc$: the bribe paid by the ringleader of the collusion to the other
server in the collusion agreement, in \SCcc.
\item[$\bullet$] $\cc$: a server’s cost for computing the task.
\item[$\bullet$] $\chc$: the fee paid to invoke an auditor for recomputing a task and resolving
disputes.
\item[$\bullet$] $\dc$: the deposit a server needs to pay to be eligible for getting the job.
\item[$\bullet$] $\tc$: the deposit the colluding parties need to pay in the collusion agreement, in \SCcc.
\item[$\bullet$] $\wc$: the amount that a server receives for completing the task.
\item[$\bullet$] $\wc \geq \cc$: the server would not accept underpaid jobs.
\item[$\bullet$] $\chc > 2\wc$: If it does not hold, then there would be no need to use the servers and the auditor would do the computation.
\item [$\bullet$] $(pk,sk)$: an asymmetric-key encryption's public-private key pair belonging to the auditor. 
\end{itemize}
\noindent The following relations need to hold when setting the contracts
in order for the desirable equilibria to hold:
(i) $\dc>\cc+\chc$, (ii) $\bc<\cc$, and (iii) $\tc<\wc-\cc + 2\dc - \chc -\bc$.

% !TEX root =main.tex

\subsection{Prisoner's Contract (\SCpc)}

\SCpc has been designed for outsourcing a certain computation. It is signed by a client who delegates the computation and two other parties  (or servers)  who perform the computation.  This contract tries to incentivize correct computation by using the following idea. It requires each server to pay a deposit before the computation is delegated. If a server behaves honestly, then it can withdraw its deposit. However, if a server cheats (and is detected), its deposit is transferred to the client. When one of the servers is honest and the other one cheats, the honest server receives a reward. This reward is taken from the deposit of the cheating server.  Hence, the goal of \SCpc is to create a Prisoner’s dilemma between the two servers in the following sense. Although the servers may collude with each other (to cut costs and provide identical but incorrect computation results) which leads to a higher payoff than both behaving honestly,  there is an even higher payoff if one of the servers manages to persuade the other server to collude and provide an incorrect result while itself provides a correct result. In this setting, each server knows that collusion is not stable as its counterparty will always try to deviate from the collusion to increase its payoff.  If a server tries to convince its counterparty (without a credible and enforceable promise), then the latter party will consider it as a trap; consequently, collusion will not occur. Below, we restate the contract. %We have slightly adjusted the contract such that misbehaving servers deposit is transferred to another contract (denoted with $\mathcal{SC}_{\epsi}$) instead of transferring directly to the client. 

\begin{enumerate}
\item The contract is signed by three parties; namely, client $ {D}$ and two other parties $E_{\st 1}$ and $E_{\st 2}$. A third-party auditor will resolve any dispute between $ { D}$ and the servers.  \textcolor{blue}{The address of another contract, called $\mathcal{SC}_{\epsi}$, is hardcoded in this contract.} 
\item The servers agree to run computation $f$ on input $x$, both of which have been provided by $ { D}$. 
\item The parties agree on three deadlines $(T_{\st 1}, T_{\st 2}, T_{\st 3})$, where $T_{\st i+1} > T_{\st i}$.
\item $ { D}$ agrees to pay $\wc$ to each server for the correct and on-time computation. Therefore, $ { D}$ deposits $2\cdot \wc$ amount in the contract. \textcolor{blue}{This deposit is transferred from $\mathcal{SC}_{\epsi}$ to this contract}.
\item Each server deposits \textcolor{blue}{$\dc'=$ }$\dc$\textcolor{blue}{$\ +\ \Xc$} amount in the contract. 
\item The servers must pay the deposit before $T_{\st 1}$. If a server fails to meet the deadline, then the contract would refund the parties' deposit (if any) and terminates.  
\item The servers must deliver the computation's result before $T_{\st 2}$. 
\item  The following steps are taken when (i) both servers provided the computation's result or (ii) deadline $T_{\st 2}$ elapsed. 
\begin{enumerate}
\item if both servers failed to deliver the computation's result, then the contract transfers their deposits to \textcolor{blue}{$\mathcal{SC}_{\epsi}$}.
\item if both servers delivered the result, and the results are equal, then (after verifying the results) \textcolor{blue}{this contract} must pay the agreed
amount $\wc$ and refund the deposit $\dc'$ to each server. 
\item\label{prisoner-cont-raise-disp} otherwise, $ { D}$ raises a dispute with the auditor.  
\end{enumerate}
\item When a dispute is raised, the auditor (which is independent of \aud in \fpsi) re-generates the computation's result, \textcolor{blue}{by using algorithm $\mathtt{resComp}(.)$ described shortly in Appendix \ref{sec::auditor-res-Comp}.} Let $y_{\st t}, y_{\st 1},$ and $y_{\st 2}$ be the result computed by the auditor, $E_{\st 1}$, and $E_{\st 2}$ respectively. The auditor uses the following role to identify the cheating party.  
\begin{itemize}
\item[$\bullet$] if $E_{\st i}$ failed to deliver the result (i.e., $y_{\st i}$ is null), then it has cheated. 
\item[$\bullet$] if a result $y_{\st i}$ has been delivered before the deadline and $y_{\st i}\neq y_{\st t}$, then $E_{\st i}$ has cheated. 
\end{itemize}
\textcolor{blue}{The auditor sends its verdict to \SCpc.} 
\item Given the auditor's decision, the dispute is settled according to the following rules.
\begin{itemize}
\item[$\bullet$] if none of the servers cheated, then \textcolor{blue}{this contract} transfers to each server (i) $\wc$ amount for performing the computation and (ii) its deposit, i.e., $\dc'$ amount. The client also pays the auditor $\chc$ amount.  
\item[$\bullet$] if both servers cheated, then \textcolor{blue}{this contract (i) pays the auditor the total amount of $\chc$, taken from the servers' deposit, and (ii) transfers to  $\mathcal{SC}_{\epsi}$ the rest of the deposit, i.e., $2\cdot \dc'-\chc$ amount.} 
\item[$\bullet$] if one of the servers cheated, then \textcolor{blue}{this contract (i) pays the auditor the total amount of $\chc$, taken from the misbehaving server's deposit, (ii) transfers the honest server's deposit (i.e., $\dc'$ amount) back to this server,  (iii) transfers to the honest server $\wc+\dc-\chc$ amount (which covers its computation cost and the reward), and  (iv) transfers to $\mathcal{SC}_{\epsi}$ the rest of the misbehaving server's deposit, i.e., $\Xc$ amount.} The cheating server receives nothing.  
\end{itemize} 
\item After deadline $T_{\st 3}$, if $ { D}$ has neither paid nor raised a dispute, then \textcolor{blue}{this contract} pays $\wc$ to each server which delivered a result before deadline $T_{\st 2}$ and refunds each server its deposit, i.e., $\dc'$ amount. Any deposit left after that will be transferred to \textcolor{blue}{$\mathcal{SC}_{\epsi}$}. 
\end{enumerate}

Now, we explain why we have made the above changes to the original \SCpc of Dong \et \cite{dong2017betrayal}. In the original \SCpc (a) the client does not deposit any amount in this contract; instead, it directly sends its coins to a server (and auditor) according to the auditor's decision,  (b) the computation correctness is determined only within this contract (with the involvement of the auditor if required), and (c)  the auditor simply re-generates the computation's result given the computation's plaintext inputs.  Nevertheless, in \epsi, (1) \emph{all clients} need to deposit a certain amount in $\mathcal{SC}_{\epsi}$ and only the contracts must transfer the parties' deposit,  (2) $\mathcal{SC}_{\epsi}$ also needs to verify a part of the computation's correctness without the involvement of the auditor and accordingly distribute the parties deposit based on the verification's outcome,  and (3) the auditor must be able to re-generate the computation's result without being able to learn the computation's plaintext input, i.e., elements of the set. Therefore, we have included the address of $\mathcal{SC}_{\epsi}$ in \SCpc to let the parties' deposit move between the two contracts (if necessary) and allowed \SCpc to distribute the parties' deposit; thus, the requirements in points (1) and (2) are met. To meet the requirement in point (3) above, we have included a new algorithm, called $\mathtt{resComp}(.)$, in \SCpc.  Shortly, we will provide more detail about this algorithm. Moreover, to make this contract compatible with \epsi, we increased the amount of each server's deposit by  $\Xc$. Nevertheless, this adjustment does not change the logic behind \SCpc's design and its analysis.

\subsubsection{Auditor's Result-Computation Algorithm.}\label{sec::auditor-res-Comp}

In this work,  we use \SCpc to delegate the computation of intersection cardinality to two extractor clients, a.k.a. servers in the original \SCpc. In this setting, the contract's auditor is invoked when an inconsistency is detected in step \ref{smart-PSI-inconsistency} of \epsi. For the auditor to recompute the intersection cardinality, we have designed $\mathtt{resComp}(.)$ algorithm. The auditor uses this algorithm for every bin's index $indx$,  where $1\leq indx\leq h$ and $h$ is the hash table's length. We present this algorithm in Figure \ref{fig::resComp}.  The auditor collects the inputs of this algorithm as follows: (a)  reads random polynomial $\bm\zeta$, and blinded polynomial $\bm\phi$ from contract $\mathcal{SC}_{\fpsi}$, (b) reads the ciphertext if secret key $mk$ from $\mathcal{SC}_{\epsi}$, and (c) fetches public parameters $(des_{\mathtt{H}}, h)$ from the hash table's public description.

Note that in the original  \SCpc of Dong \et \cite{dong2017betrayal}, the auditor is assumed to be fully trusted. However, in this work, we have relaxed such an assumption. We have designed \epsi and $\mathtt{resComp}(.)$ in such a way that even a semi-honest auditor cannot learn anything about the actual elements of the sets, as they have been encrypted under a key unknown to the auditor.

% !TEX root =main.tex

%\vspace{-2mm}
\begin{figure}%[!htbp]
\setlength{\fboxsep}{1pt}
\begin{center}
    \begin{tcolorbox}[enhanced,width=5.5in, 
    drop fuzzy shadow southwest,
    colframe=black,colback=white]
   % {\small{
    %\vspace{-2.5mm}
 \underline{$\mathtt{resComp}(\bm\zeta, \bm\phi, sk, ct_{\st mk}, indx, {des}_{\st \mathtt{H}})\rightarrow R$}\\
%
%\vspace{-2.2mm}
\begin{itemize}
\item \noindent\textit{Input}. $\bm\zeta$: a random polynomial of degree $1$, $\bm\phi$: a blinded polynomial of the form $\bm\zeta\cdot(\bm\epsilon + \bm\gamma')$ where $\bm\epsilon$ and $\bm\gamma'$ are arbitrary and  pseudorandom polynomials respectively, $deg(\bm\phi)-1=deg(\bm\gamma')$, $sk$: the auditor's secret key, $ct_{\st mk}$: ciphertext of $mk$ which is a  key of $\mathtt{PRF}$, $indx$: an input of $\mathtt{PRF}$, and   ${des}_{\st \mathtt{H}}$: a description of hash function $\mathtt{H}$.
\item\noindent\textit{Output}. $R$: a set containing valid roots of unblinded $\bm\phi$. 
\end{itemize}
\begin{enumerate}
\item decrypts the ciphertext $ct_{\st mk}$ under key $sk$. Let $mk$ be the result. 
\item unblinds polynomial $\bm\phi$, as follows:
\begin{enumerate}
\item re-generates pseudorandom polynomial $\bm\gamma'$ using key $mk$. Specifically, it uses $mk$ to derive a key: $k=\mathtt{PRF}(mk,  indx)$. Then, it uses the derived key to generate $3d+1$ pseudorandom coefficients, i.e.,  $ \forall j, 0\leq j \leq deg(\bm\phi)-1: g_{\st j}=\mathtt{PRF}(k, j)$. Next, it uses these coefficients to construct polynomial $\bm\gamma'$, i.e., $\bm\gamma'=\sum\limits_{\st j=0}^{\st deg(\bm\phi)-1} g_{\st j}\cdot x^{\st j}$.
\item removes the blinding factor from $\bm\phi$. Specifically, it computes polynomial $\bm\phi'$ of the following form $\bm\phi'= \bm\phi - \bm\zeta\cdot\bm\gamma'$. 
\end{enumerate}
\item extracts roots of polynomial $\bm\phi'$. 
\item finds valid roots, by (i) parsing each root $\bar{e}$ as $(e_{\st 1}, e_{\st 2})$ with the assistance of $des_{\st \mathtt{H}}$  and (ii) checking if $e_{\st 2}=\mathtt{H}(e_{\st 1})$. It considers a root valid, if this equation holds. 
\item returns set $R$ containing all valid roots.
%

%\vspace{-1mm}
\end{enumerate}
%}   }
\end{tcolorbox}
\end{center}
%\vspace{-3mm}
\caption{Auditor's result computation, $\mathtt{resComp}(.)$, algorithm} 
\label{fig::resComp}
\end{figure}

% !TEX root =main.tex

\subsection{Colluder's Contract (\SCcc)}

Recall that \SCpc aimed at creating a dilemma between the two servers. However, this dilemma can be addressed if they can make an  enforceable promise.  This enforceable promise can be another smart contract, called Colluder's Contract (\SCcc).  This contract imposes additional rules that ultimately would affect the parties’ payoffs and would make collusion the most profitable strategy for the colluding parties. In \SCcc, the party who initiates the collusion would pay its counterparty a certain amount (or bribe) if both follow the collusion and provide an incorrect result of the computation to \SCpc.  Note, \SCcc requires both servers to send a fixed amount of deposit when signing the contract.  The party who deviates from collusion will be punished by losing the deposit. Below, we restate \SCcc.

\begin{enumerate}
\item The contract is signed between the server who initiates the collusion, called ringleader (LDR) and the other server called follower (FLR). 
\item The two agree on providing to \SCpc a different result $res'$ than a correct computation of $f$ on $x$ would yield, i.e., $res'\neq f(x)$. Parameter $res'$ is recorded in this contract. 
\item LDR and FLR deposit $\tc+\bc$ and $\tc$ amounts in this contract respectively. 
\item The above deposit must be paid before the result delivery deadline in \SCpc, i.e., before deadline $T_{\st 2}$. If this condition is not met, the parties' deposit in this contract is refunded and this contract is terminated. 
\item When \SCpc is finalised (i.e., all the results have been provided), the following steps are taken.
\begin{enumerate}
\item \underline{both follow the collusion}: if both LDR and FLR provided $res'$ to \SCpc, then $\tc$ and $\tc+\bc$ amounts are delivered to LDR and FLR respectively. Therefore, FLR receives its deposit plus the bribe $\bc$. 
\item \underline{only FLR deviates from the collusion}:  if LDR and FLR provide $res'$ and $res''\neq res'$ to \SCpc respectively, then $2\cdot \tc+\bc$ amount is transferred to LDR while nothing is sent to FLR. 
\item \underline{only LDR deviates from the collusion}: if LDR and FLR provide $res''\neq res'$ and $res'$ to \SCpc respectively, then $2\cdot \tc+\bc$ amount is sent to FLR while nothing is transferred to LDR.  
\item \underline{both deviate from the collusion}: if LDR and FLR deviate and provide any result other than $res'$ to \SCpc, then  $2\cdot \tc+\bc$ amount is sent to LDR and $\tc$ amount is sent to FLR. 
\end{enumerate}
\end{enumerate}

We highlight that the amount of bribe a rational LDR is willing to pay is less than its computation cost (i.e., $\bc<\cc$); otherwise, such collusion would not bring a higher payoff. We refer readers to \cite{dong2017betrayal} for further discussion.

% !TEX root =main.tex

\subsection{Traitor's Contract (\SCtc)}

\SCtc incentivises a colluding server (who has had signed \SCcc) to betray its counterparty and report the collusion without being penalised by \SCpc.  The Traitor’s contract promises that the reporting server will not be
punished by \SCpc which makes it safe for the reporting server to follow the collusion strategy (of \SCcc), and get away from the punishment imposed by \SCpc. Below, we restate \SCtc.

\begin{enumerate}
\item This contract is signed among $ { D}$ and the traitor server (TRA) who reports the collusion. This contract is signed only if the parties have already signed \SCpc. 
\item  $ { D}$ signs this contract only with the first server who reports the collusion. 
\item The traitor TRA must also provide to this contract the result of the computation, i.e., $f(x)$. The result provided in this contract could be different than the one provided in \SCpc, e.g., when TRA has to follow \SCcc and provide an incorrect result to \SCpc. 
\item  $  D$ needs to pay  $\wc \ +\ $\textcolor{blue}{ $\dc' \ + $} $\dc-\chc$ amount to this contract. This amount equals the maximum amount TRA could lose in \SCpc plus the reward. \textcolor{blue}{This deposit will be transferred via $\mathcal{SC}_{\epsi}$ to this contract.}  TRA must deposit in this contract $\chc$ amount to cover the cost of resolving a potential dispute.  
\item  This contract should be signed before the deadline  $T_{\st 2}$  for the delivery of the computation result in \SCpc. If it is not signed on time, then this contract would be terminated and any deposit paid will be refunded.
\item It is required that the TRA provide the computation result to this contract before the above deadline $T_{\st 2}$.
\item If this contract is fully signed, then during the execution of \SCpc, $ { D}$ always raises a dispute, i.e., takes step \ref{prisoner-cont-raise-disp} in \SCpc.
\item After \SCpc is finalised (with the involvement of the auditor), the following steps are taken to pay the parties involved.   
\begin{enumerate}
\item  if none of the servers cheated in \SCpc (according to the auditor), then amount $\wc \ +\ $\textcolor{blue}{ $\dc' \ + $} $\dc-\chc$ is refunded to \textcolor{blue}{$\mathcal{SC}_{\epsi}$} and TRA's deposit (i.e., $\chc$ amount) is transferred to $ { D}$. Nothing is paid to TRA. 
\item if in \SCpc, the other server did not cheat and TRA  cheated; however, TRA provided a correct result in this contract, then \textcolor{blue}{ $\dc' \ + $} $\dc-\chc$ amount is transferred to  \textcolor{blue}{$\mathcal{SC}_{\epsi}$}. Also,    TRA gets its deposit back (i.e., $\chc$ amount) plus $\wc$ amounts for providing a correct result to this contract. 
\item if in \SCpc, both servers cheated; however, TRA delivered a correct computation result to this contract, then TRA gets its deposit back (i.e., $\chc$ amount), it also receives $\wc \ +\ $\textcolor{blue}{ $\dc' \ + $} $\dc-\chc$ amount.   
\item otherwise,  $\wc \ +\ $\textcolor{blue}{ $\dc' \ + $} $\dc-\chc$ and $\chc$ amounts are transferred to  \textcolor{blue}{$\mathcal{SC}_{\epsi}$} and TRA respectively. 
\end{enumerate}
\item If TRA provided a result to this contract, and deadline $T_{\st 3}$ (in \SCpc) has passed, then all deposits, if any left, will be transferred to TRA. %this ensures that the deposit will not be locked in the contract forever. 
\end{enumerate}

TRA must take the following three steps to report collusion: (i) it waits until \SCcc is signed by the other server, (ii) it reports the collusion to $ { D}$ before signing \SCcc, and (iii) it signs \SCcc only after it signed \SCtc with $ { D}$. Note, the original analysis of \SCtc does not require \SCtc to remain secret. Therefore, in our smart PSI, parties TRA and $ { D}$ can sign this contract and   then store its address in $\mathcal{SC}_{\epsi}$. Alternatively, to keep this contract confidential, $ { D}$ can deploy a template \SCtc to the blockchain and store the commitment of the contract's address (instead of the plaintext address) in $\mathcal{SC}_{\epsi}$. When a traitor (TRA) wants to report collusion, it signs the deployed \SCtc with  $ { D}$ which provides the commitment opening to TRA. In this case, at the time when the deposit is distributed,   either $ { D}$ or TRA   provides the opening of the commitment to  $\mathcal{SC}_{\epsi}$ which checks whether it matches the commitment. If the check passes, then it distributes the deposit as before.

% !TEX root =main.tex

\section{Proof of Theorem \ref{theorem:coef-poly-prod}}\label{sec::proof-of-poly-union}

Below, we restate the proof of Theorem \ref{theorem:coef-poly-prod}, taken from \cite{AbadiMZ21}.

\begin{proof}
Let $P=\{p_{\st 1},...,p_{\st t}\}$ and $Q=\{q_{\st 1},...,q_{\st t'}\}$ be the roots of polynomials $\mathbf{p}$ and   $\mathbf{q}$  respectively.  By the Polynomial Remainder Theorem,  polynomials $\mathbf{p}$ and $\mathbf{q}$  can be written as $\mathbf{p}(x)=\mathbf{g}(x)\cdot\prod\limits_{\st i=1}^{\st t}(x-p_{\st i})$ and $\mathbf{q}(x)=\mathbf{g}'(x)\cdot\prod\limits_{\st i=1}^{\st t'}(x-q_{\st i})$ respectively, where $\mathbf{g}(x)$ has degree $d-t$ and $\mathbf{g}'(x)$ has degree $d'-t'$. Let the product of the two polynomials be $\mathbf{r}(x)=\mathbf{p}(x)\cdot \mathbf{q}(x)$. For every $p_{\st i}\in P$, it holds  that $\mathbf{r}(p_{\st i})=0$. Because (a) there exists no non-constant polynomial in $\mathbb{F}_{\st p}[X]$ that has a multiplicative inverse (so it could cancel out factor $(x-p_{\st i})$ of $\mathbf{p}(x)$) and (b) $p_{\st i}$ is a root of $\mathbf{p}(x)$. The same argument  can be used to show for every $q_{\st i}\in Q$, it holds that $\mathbf{r}(q_{\st i})=0$. Thus, $\mathbf{r}(x)$ preserves  roots of  both  $\mathbf{p}$ and $\mathbf{q}$. Moreover, $\mathbf{r}$ does not have any other roots (than $P$ and $Q$). In particular, if $\mathbf{r}(\alpha)=0$, then $\mathbf{p}(\alpha)\cdot \mathbf{q}(\alpha)=0$. Since there is no non-trivial divisors of zero in $\mathbb{F}_{\st p}[X]$  (as it is an integral domain), it must hold that either $\mathbf{p}(\alpha)=0$ or $\mathbf{q}(\alpha)=0$. Hence, $\alpha\in P$ or $\alpha\in Q$.  %\hfill\(\Box\)
\hfill\(\Box\)
\end{proof}

%

% !TEX root =main.tex

\section{Error Probability}\label{sec::error-prob}

Recall that in \fpsi, in step \ref{JUS::check-non-zero-coeff},  each client $C$ needs to ensure polynomials $\bm\omega^{\st (C,D)}\cdot \bm\pi^{\st  {  {(C)}}}$ and  $\bm\rho^{\st (C,D)}$ do not contain any zero coefficient. This check ensures that client $C$ (the receiver in \vopr) does not insert any zero values to  \vopr (in particular to $\ole^{\st +}$ which is a subroutine of \vopr). If a zero value is inserted in \vopr, then an honest receiver will learn only a random value and more importantly cannot pass \vopr's verification phase. 
Nevertheless,  this check can be removed from step \ref{JUS::check-non-zero-coeff}, if we allow \fpsi to output an error with a small probability.\footnote{By error we mean even if all parties are honest,  \vopr halts when an honest party inserts $0$ to it.} In the remainder of this section, we show that this probability is negligible.

First, we focus on the product $\bm\omega^{\st (C,D)}\cdot \bm\pi^{\st  {  {(C)}}}$. We know that $\bm\pi^{\st  {  {(C)}}}$ is of the form $\prod\limits^{\st d}_{\st i=1} (x-s'_{\st i})$, where $s'_{\st i}$ is either a set element $s_{\st i}$ or a random value.  Thus, all of its coefficients are non-zero. Also, the probability that at least one of the coefficients  of $d$-degree random polynomial $\bm\omega^{\st (C,D)}$ equals $0$ is at most $\frac{d+1}{p}$.  Below, we state it formally.

\begin{theorem}\label{theorem::zero-coeff-in-ran-poly}
Let $\bm\delta=\sum\limits_{\st t=0}^{\st d}u_{\st  t} \cdot x^{\st t}$, where  $u_{\st  t}\stackrel{\st\$}\leftarrow\mathbb{F}_{\st p}$, for all $t, 0\leq t\leq d$. Then, the probability that at least one of the coefficients equals $0$ is at most $\frac{d+1}{p}$, i.e., 

$$Pr[\exists u_{\st t}, u_{\st t}=0]\leq \frac{d+1}{p}$$
\end{theorem}

\begin{proof}
The proof is straightforward. Since $\bm\delta$'s coefficients are picked uniformly at random from  $\mathbb{F}_{\st p}$, the probability that $u_{\st t}$ equals $0$ is $\frac{1}{p}$. Since $\bm\delta$ is of degree $d$, due to the union bound, the probability that at least one of $c_{\st t}$s equals $0$ is at most $\frac{d+1}{p}$. 
 \hfill\(\Box\)\end{proof}

Next, we show that the probability that the polynomial  $\bm\omega^{\st (C,D)}\cdot \bm\pi^{\st  {  {(C)}}}$ has at least one zero coefficient is negligible in the security parameter; we assume polynomial $\bm\omega^{\st (C,D)}$  has no zero coefficient.

%For the sake of simplicity,  we set $\bm{p}_1=\bm\omega^{\st (C,D)}=a_{\st 0}\cdot a_{\st 1}\cdot x+...+a_{\st d}\cdot x^{\st d}$ and $\bm{p}_2= \bm\pi^{\st  {  {(C)}}}$. 

\begin{theorem}\label{theorem::zero-coeff-in-product}
Let  $\bm\alpha=\sum\limits_{\st j=0}^{\st m}a_{\st  i} \cdot x^{\st i}$ and  $\bm\beta=\sum\limits_{\st j=0}^{\st n}b_{\st  j} \cdot x^{\st j}$,  where $a_{\st i}\stackrel{\st\$}\leftarrow \mathbb{F}_{\st p}$ and $a_{\st i}, b_{\st j}\neq0$, for all $i,j, 0\leq i \leq m$ and $0\leq j \leq n$. Also, let $\bm\gamma=\bm\alpha\cdot\bm\beta=\sum\limits_{\st j=0}^{\st m+n}c_{\st  j} \cdot x^{\st j}$. Then, the probability that at lest one of the coefficients of polynomial $\bm\gamma$ equals $0$ is at most $\frac{m+n+1}{p}$, i.e., 
$$Pr[\exists c_{\st j}, c_{\st j}=0]\leq \frac{m+n+1}{p}$$
\end{theorem}

\begin{proof}
Each coefficient $c_{\st  k}$ of $\bm\gamma$ can be defined as $c_{\st  k}=\sum\limits_{\substack{\st j=0\\ \st i=0}}^{\substack{\st i=m\\ \st j=n}}a_{\st i}\cdot b_{\st j}$, where $i+j=k$. We can rewrite $c_{\st  k}$ as  $c_{\st  k}=a_{\st w}\cdot b_{\st z}+ \sum\limits_{\substack{\st j=0, j\neq z\\ \st i=0, i\neq w}}^{\substack{\st i=m\\ \st j=n}}a_{\st i}\cdot b_{\st j}$, where $w+z=i+j=k$. 
We consider two cases for each $c_{\st  k}$:

\begin{itemize}

\item[$\bullet$]  {Case 1}: $\sum\limits_{\substack{\st j=0, j\neq z\\ \st i=0, i\neq w}}^{\substack{\st i=m\\ \st j=n}}a_{\st i}\cdot b_{\st j}=0$.  This is a trivial case, because with the  probability of $1$ it holds that $c_{\st  k}=a_{\st w}\cdot b_{\st z}\neq 0$, as by  definition $a_{\st w}, b_{\st z}\neq0$ and $\mathbb{F}_{\st p}$ is an integral domain. 

%In this case, since $a_{\st w}$ has been picked uniformly at random, the probability that $c_{\st  k}=a_{\st w}\cdot b_{\st z}=0$ is $\frac{1}{p}$.

%
% This is a trivial case, because with the  probability of $1$ it holds that $c_{\st  k}=a_{\st w}\cdot b_{\st z}\neq 0$. 

%In this case, because $a_{\st w}$ has been picked uniformly at random, with the  probability of $\frac{1}{p}$ it holds that $c_{\st  k}=a_{\st w}\cdot b_{\st z}=0$. 

\item[$\bullet$]  {Case 2}: $q=\sum\limits_{\substack{\st j=0, j\neq z\\ \st i=0, i\neq w}}^{\substack{\st i=m\\ \st j=n}}a_{\st i}\cdot b_{\st j}\neq 0$. In this case, for  event $c_{\st  k}=a_{\st w}\cdot b_{\st z}+q=0$ to occure, $a_{\st w}\cdot b_{\st z}$ must equal the additive inverse of $q$. Since $a_{\st w}$ has been picked uniformly at random, the probability that such an even occurs is $\frac{1}{p}$.
\end{itemize}

The above analysis is for a single $c_{\st  k}$. Thus, due to the union bound, the probability that at least one of the coefficients $c_{\st  k}$ equals $0$ is at most $\sum\limits^{\st m+n}_{\st j=0}\frac{1}{p}= \frac{m+n+1}{p}$.
 \hfill\(\Box\)\end{proof}

Next we turn out attention to $\bm\rho^{\st (C,D)}$. Due to Theorem \ref{theorem::zero-coeff-in-ran-poly}, the probability that at least one of the coefficients of $\bm\rho^{\st (C,D)}$ equals $0$ is at most $\frac{d+1}{p}$, . 

%\begin{theorem}\label{theorem::zero-coeff-in-ran-poly}
%Let $\bm\delta=\sum\limits_{\st t=0}^{\st d}u_{\st  t} \cdot x^{\st t}$, where  $u_{\st  t}\stackrel{\st\$}\leftarrow\mathbb{F}_{\st p}$, for all $t, 0\leq t\leq d$. Then, the probability that at least one of coefficients equals $0$ is at most $\frac{d}{p}$, i.e., 
%$Pr[\exists d_{\st j}, d_{\st j}=0]\leq \frac{d+1}{p}$.
%\end{theorem}
%
%
%\begin{proof}
%The proof is straightforward. Since $\bm\delta$'s coefficients are picked uniformly at random from  $\mathbb{F}_{\st p}$, the probability that $c_{\st j}$ equals $0$ is $\frac{1}{p}$. Since $\bm\delta$ is of degree $d$, due to the union bound, the probability that at least one of $c_{\st j}$s equals $0$ is at most $\frac{d+1}{p}$. 
% \hfill\(\Box\)\end{proof}

Hence, due to Theorems \ref{theorem::zero-coeff-in-ran-poly}, \ref{theorem::zero-coeff-in-product}, and union bound, the probability that at least one of the coefficients in $\bm\omega^{\st (C,D)}\cdot \bm\pi^{\st  {  {(C)}}}$ and  $\bm\rho^{\st (C,D)}$ equals $0$ is at most $\frac{3d+2}{p}$, which is negligible is the security parameter $p$.

\end{document}